\newif
\ifdraftmode
\draftmodefalse

\ifdraftmode
\else
\fi	
\RequirePackage{fix-cm}
\documentclass[reqno]{amsart}
\usepackage[margin=1.5in,bottom=1.25in]{geometry}	

\usepackage{amssymb}	
\usepackage{amsfonts}	
\usepackage{amsthm}	
\usepackage[foot]{amsaddr}	

\usepackage{mathtools}	
\mathtoolsset{%
centercolon=true,
}

\usepackage[
cal=cm,
]
{mathalfa}

\usepackage{dsfont}	

\usepackage[proportional,tabular,lining,sf,mono=false]{libertine}

\usepackage[scaled=.95]{AlegreyaSans}

\usepackage[T1]{fontenc}	

\usepackage{acronym}	
\newcommand{\acli}[1]{\emph{\acl{#1}}}	
\newcommand{\acdef}[1]{\define{\acl{#1}} \textup{(\acs{#1})}\acused{#1}}	

\usepackage[labelfont={bf,small},labelsep=colon,font=small]{caption}	

\usepackage{subcaption}	
\usepackage[svgnames]{xcolor}	

\definecolor{Bonfire}{HTML}{9E162E}
\definecolor{CardinalRed}{HTML}{C41E3A}
\definecolor{TuckOrange}{HTML}{D94415}

\definecolor{CadmiumGreen}{HTML}{097969}
\definecolor{Dartmouth}{HTML}{00693E}
\definecolor{ForestGreen}{HTML}{12312B}
\definecolor{RichForestGreen}{HTML}{0D1E1C}
\definecolor{SeaGreen}{HTML}{2E8B57}
\definecolor{SpringGreen}{HTML}{EAFAF1}
\definecolor{Jade}{HTML}{009900}

\definecolor{CobaltBlue}{HTML}{0047AB}
\definecolor{NavyBlue}{HTML}{000080}
\definecolor{RiverBlue}{HTML}{267ABA}
\definecolor{RiverNavy}{HTML}{003C73}
\definecolor{KleinBlue}{HTML}{002FA7}
\definecolor{OxfordBlue}{HTML}{002147}
\definecolor{SapphireBlue}{HTML}{0F52BA}
\definecolor{Zaffre}{HTML}{0819A8}

\colorlet{MyRed}{CardinalRed}
\colorlet{MyGreen}{Dartmouth}
\colorlet{MyBlue}{DodgerBlue}
\colorlet{MyViolet}{DarkMagenta}

\colorlet{MyDarkRed}{red!70!black}

\colorlet{MyLightRed}{MyRed!25}
\colorlet{MyLightGreen}{MyGreen!25}
\colorlet{MyLightBlue}{MyBlue!25}

\colorlet{PrimalColor}{MyBlue}
\colorlet{PrimalFill}{PrimalColor!25}
\colorlet{DualColor}{MyRed}

\colorlet{AlertColor}{MyRed}	
\colorlet{BadColor}{MyRed}	
\colorlet{GoodColor}{MyGreen}	
\colorlet{LinkColor}{MediumBlue}	
\colorlet{RevColor}{MediumBlue}	

\ifdraftmode
	\colorlet{DraftColor}{MyRed}	
\else
	\colorlet{DraftColor}{black}	
\fi

\newcommand{\afterhead}{.}	
\newcommand{\para}[1]{\smallskip\paragraph{\textbf{#1\afterhead}}}	

\usepackage{latexsym}	
\usepackage{fontawesome}	
\usepackage{pifont}	

\newcommand{\asterism}{\ding{70}}

\usepackage{tikz}	
\usetikzlibrary{calc,patterns,arrows.meta,positioning}	

\usepackage{array}	
\usepackage{booktabs}	
\usepackage[inline,shortlabels]{enumitem}	
\setlist[1]{topsep=\smallskipamount,itemsep=\smallskipamount,left=\parindent}
\setlist[2]{left=0pt}

\usepackage[kerning=true]{microtype}	

\usepackage{csquotes}
\usepackage{float}

\usepackage{balance}	
\usepackage{tabto}	
\usepackage{xspace}	

\usepackage[sort&compress,numbers]{natbib}	

\bibpunct[, ]{[}{]}{,}{}{,}{,}
\setcitestyle{numbers,square,comma}

\usepackage{hyperref}
\hypersetup{
final,
colorlinks=true,
linktocpage=true,
pdfstartview=FitH,
breaklinks=true,
pdfpagemode=UseNone,
pageanchor=true,
pdfpagemode=UseOutlines,
plainpages=false,
bookmarksnumbered,
bookmarksopen=false,
bookmarksopenlevel=1,
hypertexnames=false,
pdfhighlight=/O,
urlcolor=LinkColor,linkcolor=LinkColor,citecolor=LinkColor,	
pdftitle={},
pdfauthor={},
pdfsubject={},
pdfkeywords={},
pdfcreator={pdfLaTeX},
pdfproducer={LaTeX with hyperref}
}

\usepackage[sort&compress,capitalize,nameinlink,noabbrev]{cleveref}	

\crefname{algo}{Algorithm}{Algorithms}
\crefname{assumption}{Assumption}{Assumptions}

\makeatletter
\AtBeginDocument
{
	\def\ltx@label#1{\cref@label{#1}}	
	\def\label@in@display@noarg#1{\cref@old@label@in@display{#1}}	
}%
\makeatother

\usepackage{algorithm}	
\usepackage{algpseudocode}	

\usepackage[most]{tcolorbox}

\theoremstyle{plain}
\newtheorem{theorem}{Theorem}	
\newtheorem{corollary}{Corollary}	
\newtheorem{lemma}{Lemma}	
\newtheorem{proposition}{Proposition}	

\newtheorem*{theorem*}{Theorem}	
\newtheorem*{corollary*}{Corollary}	

\theoremstyle{definition}
\newtheorem{definition}{Definition}	
\newtheorem{assumption}{Assumption}	
\newtheorem{example}{Example}	

\newtheorem*{definition*}{Definition}	
\newtheorem*{assumption*}{Assumptions}	
\newtheorem*{example*}{Example}	

\theoremstyle{remark}
\newtheorem{remark}{Remark}	

\newtheorem*{remark*}{Remark}	
\newtheorem*{notation*}{Notation}	

\def\endenv{\hfill\asterism}	

\newcounter{proofstep}

\numberwithin{example}{section}	

\ifdraftmode
	\usepackage[showdeletions]{color-edits}	
\else
	\usepackage[suppress]{color-edits}	
\fi

\ifdraftmode
	\newcommand{\draft}[1]{{\color{DraftColor}#1}}	
\else
	\newcommand{\draft}[1]{#1}	
\fi

\newcommand{\define}[1]{\emph{\draft{#1}}}	

\newcommand{\newmacro}[2]{\newcommand{#1}{\draft{#2}}}	
\newcommand{\newop}[2]{\DeclareMathOperator{#1}{\draft{#2}}}	
\newcommand{\newoplims}[2]{\DeclareMathOperator*{#1}{\draft{#2}}}	

\DeclarePairedDelimiter{\braces}{\{}{\}}	
\DeclarePairedDelimiter{\bracks}{[}{]}	

\DeclarePairedDelimiterX{\setdef}[2]{\{}{\}}{#1:#2}	
\DeclarePairedDelimiterXPP{\exclude}[1]{\mathopen{}\setminus}{\{}{\}}{}{#1}	

\DeclarePairedDelimiterX{\braket}[2]{\langle}{\rangle}{#1,#2}	
\DeclarePairedDelimiterX{\inner}[2]{\langle}{\rangle}{#1,#2}	

\DeclarePairedDelimiter{\norm}{\lVert}{\rVert}	
\DeclarePairedDelimiterXPP{\dnorm}[1]{}{\lVert}{\rVert}{_{\draft{\ast}}}{#1}	
\DeclarePairedDelimiterXPP{\onenorm}[1]{}{\lVert}{\rVert}{_{\draft{1}}}{#1}	
\DeclarePairedDelimiterXPP{\twonorm}[1]{}{\lVert}{\rVert}{_{\draft{2}}}{#1}	
\DeclarePairedDelimiterXPP{\pnorm}[1]{}{\lVert}{\rVert}{_{\draft{p}}}{#1}	
\DeclarePairedDelimiterXPP{\qnorm}[1]{}{\lVert}{\rVert}{_{\draft{q}}}{#1}	
\DeclarePairedDelimiterXPP{\supnorm}[1]{}{\lVert}{\rVert}{_{\draft{\infty}}}{#1}	
\DeclarePairedDelimiterXPP{\matnorm}[1]{}{\lVert}{\rVert}{_{\draft{F}}}{#1}	

\newop{\defeq}{\coloneqq}	
\newop{\eqdef}{\eqqcolon}	

\newmacro{\from}{\colon}	
\newop{\too}{\rightrightarrows}	
\newop{\injects}{\hookrightarrow}	
\newop{\surjects}{\twoheadrightarrow}	

\newmacro{\Z}{\mathbb{Z}}	
\newmacro{\Q}{\mathbb{Q}}	
\newmacro{\C}{\mathbb{C}}	

\newoplims{\argmax}{arg\,max}	
\newoplims{\argmin}{arg\,min}	
\newoplims{\intersect}{\bigcap}	
\newoplims{\union}{\bigcup}	

\newop{\aff}{aff}	
\newop{\bd}{bd}	
\newop{\bigoh}{\mathcal{O}}	
\newop{\tildeoh}{\mathcal{\tilde O}}	
\newop{\baroh}{\mathcal{\bar O}}	
\newmacro{\littleoh}{o}	
\newop{\card}{\#}	
\newop{\cl}{cl}	
\newop{\conv}{conv}	
\newop{\crit}{crit}	
\newop{\curl}{curl}	
\newop{\diag}{diag}	
\newop{\diam}{diam}	
\newop{\dist}{dist}	
\newop{\diver}{div}	
\newop{\dom}{dom}	
\newop{\eig}{eig}	
\newop{\ess}{ess}	
\newop{\grad}{grad}	
\newop{\ind}{ind}	
\newop{\im}{im}	
\newop{\intr}{int}	
\newop{\Jac}{Jac}	
\newop{\one}{\mathds{1}}	
\newop{\proj}{proj}	
\newop{\prox}{prox}	
\newop{\rank}{rank}	
\newop{\relint}{ri}	
\newop{\sign}{sgn}	
\newop{\supp}{supp}	
\newop{\Sym}{Sym}	
\newop{\tr}{tr}	
\newop{\unif}{unif}	
\newop{\vol}{vol}	

\newcommand{\cf}{cf.\xspace}	
\newcommand{\eg}{e.g.,\xspace}	
\newcommand{\ie}{i.e.,\xspace}	
\newcommand{\viz}{viz.\xspace}	

\newcommand{\textpar}[1]{\textup(#1\textup)}	

\newmacro{\commentsymbol}{\triangleright}	
\newcommand{\eqstop}{\,.}	

\newcommand{\alt}[1]{#1'}	

\newmacro{\argdot}{\boldsymbol{\cdot}}	
\newmacro{\dd}{\kern0pt\:d}	
\newmacro{\ddt}{\frac{d}{dt}}	
\newmacro{\del}{\partial}	

\newcommand{\insum}{\sum\nolimits}	

\newmacro{\const}{c}	
\newmacro{\Const}{C}	

\newmacro{\param}{\theta}	
\newmacro{\params}{\Theta}	

\newmacro{\coef}{\lambda}	

\newmacro{\fn}{f} 

\newmacro{\pexp}{p}	
\newmacro{\qexp}{q}	
\newmacro{\rexp}{r}	

\newmacro{\iCount}{i}	
\newmacro{\jCount}{j}	
\newmacro{\kCount}{k}	
\newmacro{\nCounts}{n}	
\newmacro{\counts}{\mathcal{I}}	

\newmacro{\point}{x}	
\newmacro{\pointalt}{\alt\point}	
\newmacro{\pointaux}{u}	
\newmacro{\points}{\mathcal{X}}	
\newmacro{\intpoints}{\relint\points}	

\newmacro{\base}{q}	
\newmacro{\basealt}{q}	
\newmacro{\auxpoint}{u}	

\newmacro{\elem}{a}	
\newmacro{\elemalt}{b}	
\newmacro{\iElem}{a}	
\newmacro{\jElem}{b}	
\newmacro{\kElem}{c}	
\newmacro{\set}{S}	

\newmacro{\borel}{\mathcal{B}}	
\newmacro{\closed}{\mathcal{C}}	
\newmacro{\cpt}{\mathcal{K}}	
\newmacro{\nhd}{\mathcal{U}}	
\newmacro{\nhdalt}{\mathcal{V}}	
\newmacro{\open}{\mathcal{U}}	

\newmacro{\domain}{\mathcal{D}}	
\newmacro{\region}{\mathcal{R}}	

\newmacro{\interval}{\mathcal{I}}	
\newmacro{\rectangle}{\mathcal{R}}	
\newmacro{\cone}{\mathcal{K}}	

\newmacro{\tstart}{0}	
\renewcommand{\time}{\draft{t}}	
\newmacro{\timealt}{s}	
\newmacro{\timealtalt}{\tau}	
\newmacro{\horizon}{T}	

\newmacro{\curve}{\gamma}	
\DeclarePairedDelimiterXPP{\curveof}[1]{\curve}{(}{)}{}{#1}	
\DeclarePairedDelimiterXPP{\curveofX}[2]{\curve_{#1}}{(}{)}{}{#2}	
\DeclarePairedDelimiterXPP{\velof}[1]{\dot\curve}{(}{)}{}{#1}	
\DeclarePairedDelimiterXPP{\velofX}[2]{\dot\curve_{#1}}{(}{)}{}{#2}	

\newmacro{\flowmap}{\Phi}	
\DeclarePairedDelimiterXPP{\flowof}[2]{\flowmap_{#1}}{(}{)}{}{#2}	

\newmacro{\traj}{x}	
\newmacro{\dtraj}{\dot\traj}	

\DeclarePairedDelimiterXPP{\trajof}[1]{\traj}{(}{)}{}{#1}	
\DeclarePairedDelimiterXPP{\trajofX}[2]{\traj_{#1}}{(}{)}{}{#2}	
\DeclarePairedDelimiterXPP{\dtrajof}[1]{\dtraj}{(}{)}{}{#1}	
\DeclarePairedDelimiterXPP{\dtrajofX}[2]{\dtraj_{#1}}{(}{)}{}{#2}	

\newmacro{\vdim}{n}	
\newmacro{\realspace}{\R^{\vdim}}	

\newmacro{\iCoord}{i}	
\newmacro{\jCoord}{j}	
\newmacro{\kCoord}{k}	
\newmacro{\nCoords}{\vdim}	
\newmacro{\mCoords}{m}	

\newmacro{\unitvec}{u}	
\newmacro{\bvec}{e}	
\newmacro{\bvecs}{\mathcal{E}}	

\newmacro{\vecspace}{\mathcal{V}}	
\newmacro{\subspace}{\mathcal{W}}	

\newcommand{\dual}[1][\vecspace]{#1^{\ast}}	
\newmacro{\dspace}{\dual[\vecspace]}	

\newmacro{\hilbert}{\mathcal{H}}	
\newmacro{\banach}{\mathcal{X}}	

\newmacro{\mat}{M}	
\newmacro{\hmat}{H}	

\newmacro{\ones}{\mathbf{1}}	
\newmacro{\eye}{I}	
\newmacro{\zer}{\mathbf{0}}	

\newmacro{\eigval}{\lambda}	
\newmacro{\eigvec}{u}	

\newmacro{\ball}{\mathbb{B}}	
\newmacro{\sphere}{\mathbb{S}}	
\newmacro{\radius}{r}
\newmacro{\Radius}{R}

\newmacro{\mfld}{\mathcal{M}}	
\newmacro{\tanvec}{z}	
\newmacro{\form}{\omega}	

\newmacro{\gmat}{g}	
\newmacro{\gdist}{\dist_{\gmat}}	

\newmacro{\vertex}{v}	
\newmacro{\vertexalt}{w}	
\newmacro{\iVertex}{\vertex_{\iCount}}	
\newmacro{\jVertex}{\vertex_{\jCount}}	
\newmacro{\kVertex}{\vertex_{\kCount}}	
\newmacro{\nVertices}{V}	
\newmacro{\vertices}{\mathcal{V}}	

\newmacro{\edge}{e}	
\newmacro{\edgealt}{\alt\edge}	
\newmacro{\iEdge}{\edge_{\iCount}}	
\newmacro{\jEdge}{\edge_{\jCount}}	
\newmacro{\kEdge}{\edge_{\kCount}}	
\newmacro{\nEdges}{E}	
\newmacro{\edges}{\mathcal{\nEdges}}	

\newmacro{\graph}{\mathcal{G}}	
\newmacro{\graphfull}{\graph(\vertices,\edges)}	

\newop{\minimize}{minimize}	
\newop{\opt}{Opt}	
\newop{\gap}{Gap}	

\newmacro{\cvx}{\mathcal{C}}	
\newmacro{\obj}{f}	
\newmacro{\sobj}{F}	
\newmacro{\oper}{A}	
\newmacro{\vecfield}{v}	

\newmacro{\subd}{\partial}	
\newmacro{\subsel}{\nabla}	
\newmacro{\gvec}{g}	

\newmacro{\gbound}{G}	
\newmacro{\vbound}{V}	
\newmacro{\lips}{L}	
\newmacro{\strong}{\mu}	
\newmacro{\smooth}{\beta}	

\newop{\tcone}{TC}	
\newop{\dcone}{\tcone^{\ast}}	
\newop{\ncone}{NC}	
\newop{\pcone}{PC}	
\newop{\hull}{\Delta}	

\newop{\ex}{\mathbb{E}}	
\newop{\prob}{\mathbb{P}}	
\newop{\Var}{\mathbb{V}}	
\newop{\cov}{cov}	
\newop{\simplex}{\Delta}	

\DeclarePairedDelimiterXPP{\exof}[1]{\ex}{[}{]}{}{
 #1}

\DeclarePairedDelimiterXPP{\exwrt}[2]{\ex_{#1}}{[}{]}{}{
 #2}

\DeclarePairedDelimiterXPP{\probof}[1]{\prob}{(}{)}{}{
 #1}

\DeclarePairedDelimiterXPP{\probwrt}[2]{\prob_{#1}}{(}{)}{}{
 #2}

\DeclarePairedDelimiterXPP{\oneof}[1]{\one}{\{}{\}}{}{#1}	

\DeclarePairedDelimiterXPP{\varof}[1]{\var}{[}{]}{}{
 #1}

\DeclarePairedDelimiterXPP{\covof}[1]{\cov}{(}{)}{}{
 #1}

\newmacro{\event}{\mathcal{E}}       
\newmacro{\eventalt}{H}       

\newmacro{\sample}{\omega}	
\newmacro{\samples}{\Omega}	
\newmacro{\filter}{\mathcal{F}}	
\newmacro{\probspace}{(\samples,\filter,\prob)}	

\newmacro{\pdist}{P}	
\newmacro{\history}{\mathcal{H}}	

\newmacro{\mean}{\mu}	
\newmacro{\sdev}{\sigma}	
\newmacro{\variance}{\sdev^{2}}	
\newmacro{\covmat}{\Sigma}	

\newmacro{\seq}{a}	
\newmacro{\seqalt}{b}	

\newmacro{\beforestart}{0}	
\newmacro{\start}{1}	
\newmacro{\afterstart}{2}	
\newmacro{\running}{\start,\afterstart,\dotsc}	

\newmacro{\run}{t}	
\newmacro{\runalt}{s}	
\newmacro{\runaltalt}{\tau}	
\newmacro{\nRuns}{T}	
\newmacro{\runs}{\mathcal{\nRuns}}	

\newop{\Nash}{Nash}	
\newop{\CE}{CE}	
\newop{\CCE}{CCE}	
\newop{\NI}{NI}	

\newmacro{\stratdiam}{\norm{\strats}}
\newmacro{\distort}{\chi(\nPures)}

\newop{\brep}{BR}	
\newop{\val}{val}	

\newmacro{\play}{i}	
\newmacro{\playalt}{j}	
\newmacro{\iPlay}{i}	
\newmacro{\jPlay}{j}	
\newmacro{\kPlay}{k}	
\newmacro{\nPlayers}{N}	
\newmacro{\players}{\mathcal{\nPlayers}}	

\newmacro{\pure}{\alpha}	
\newmacro{\purealt}{\beta}	
\newmacro{\nPures}{A}	
\newmacro{\pures}{\mathcal{A}}	

\newmacro{\strat}{x}	
\newmacro{\stratalt}{\alt\strat}	
\newmacro{\strataux}{q}	
\newmacro{\strats}{\mathcal{X}}	
\newmacro{\intstrats}{\strats^{\circle}}	

\newmacro{\corr}{z}	
\newmacro{\corralt}{\alt\corr}	
\newmacro{\corrs}{\mathcal{Z}}	

\newmacro{\pay}{u}	
\newmacro{\loss}{\ell}	
\newmacro{\cost}{c}	
\newmacro{\pot}{f}	

\newmacro{\payvec}{w}	
\newmacro{\payv}{v}	
\newmacro{\payfield}{\payv}	
\newmacro{\paybound}{M}	
\newmacro{\payspace}{\mathcal{Y}}	
\newmacro{\payspacei}{\payspace_{\play}}	

\newmacro{\game}{\mathcal{G}}	
\newmacro{\gamefull}{\game(\players,\points,\pay)}	

\newmacro{\fingame}{\Gamma}	
\newmacro{\fingamefull}{\Gamma(\players,\pures,\pay)}	
\newmacro{\mixgame}{\Delta(\fingame)}	

\newmacro{\minmax}{L}	

\newmacro{\minvar}{\point_{1}}	
\newmacro{\minvaralt}{\alt\minvar}	
\newmacro{\minvars}{\points_{1}}	

\newmacro{\maxvar}{\point_{2}}	
\newmacro{\maxvaralt}{\alt\maxvar}	
\newmacro{\maxvars}{\points_{2}}	

\newmacro{\hreg}{h}	
\newmacro{\proxdom}{\points_{\hreg}}	

\newmacro{\breg}{D}	
\newmacro{\mprox}{P}	

\newmacro{\hconj}{h^{\ast}}	
\newmacro{\mirror}{Q}	
\newmacro{\fench}{F}	
\newmacro{\hker}{\theta}	

\newmacro{\hstr}{K}	
\newmacro{\hrange}{H}	

\newmacro{\learn}{\eta}	
\newmacro{\weight}{\lambda}	

\DeclarePairedDelimiterXPP{\bregof}[2]{\breg}{(}{)}{}{#1,#2}	
\DeclarePairedDelimiterXPP{\bregofX}[3]{\breg_{#1}}{(}{)}{}{#2,#3}	
\DeclarePairedDelimiterXPP{\fenchof}[2]{\fench}{(}{)}{}{#1,#2}	
\DeclarePairedDelimiterXPP{\fenchofX}[3]{\fench_{#1}}{(}{)}{}{#2,#3}	
\DeclarePairedDelimiterXPP{\proxof}[2]{\mprox_{#1}}{(}{)}{}{#2}	

\newmacro{\choice}{\mirror}	

\newmacro{\zone}{\mathbb{D}}	

\newop{\Eucl}{\Pi}	
\newop{\logit}{\Lambda}	
\newop{\dkl}{KL}	

\newmacro{\dvec}{w}	
\newmacro{\dpoint}{y}	
\newmacro{\dpointalt}{\alt\dpoint}	
\newmacro{\dpoints}{\mathcal{Y}}	

\newmacro{\score}{y}	
\newmacro{\scorealt}{\alt\score}	
\newmacro{\scoreaux}{z}	
\newmacro{\scores}{\payspace}	

\newmacro{\momexp}{p}	

\newmacro{\state}{X}	
\newmacro{\dstate}{Y}	

\newmacro{\drift}{b}	
\newmacro{\diffmat}{\Sigma}	

\newmacro{\brown}{W}	
\newmacro{\ito}{M}	
\newmacro{\levy}{L}	

\newmacro{\model}{\hat\vecfield}	
\newmacro{\sgrad}{\gvec}	
\newmacro{\step}{\gamma}	
\newmacro{\runtime}{\tau}	

\newmacro{\stepexp}{\ell_{\step}}	
\newmacro{\learnexp}{\ell_{\learn}}	

\newmacro{\apt}{X}	
\DeclarePairedDelimiterXPP{\aptof}[1]{\apt}{(}{)}{}{#1}	
\DeclarePairedDelimiterXPP{\aptofX}[2]{\apt_{#1}}{(}{)}{}{#2}	

\newop{\orcl}{\mathsf{G}}	
\newop{\err}{\mathsf{\noise}}	
\newmacro{\seed}{\omega}	
\newmacro{\seeds}{\Omega}	

\newmacro{\noise}{U}	
\newmacro{\bias}{b}	

\newmacro{\bbound}{B}	
\newmacro{\totbound}{\paybound}	
\newmacro{\mombound}{V}	

\newmacro{\snoise}{\xi}	
\newmacro{\sbias}{\chi}	

\newmacro{\mix}{\delta}	
\newmacro{\perturb}{z}	
\newmacro{\pivot}{\point}	

\newop{\reg}{Reg}	
\newop{\preg}{\overline{Reg}}	

\newmacro{\bench}{p}	
\newmacro{\test}{p}	

\addauthor[Pan]{PM}{MediumBlue}

\newop{\skel}{skel}	

\newcommand{\thisfigscale}{1}	
\newcommand{\thisfigheight}{1ex}	

\newcommand{\A}{\mathcal{A}}

\newcommand{\B}{\mathcal{B}}
\newcommand{\F}{\mathcal{F}}
\newcommand{\G}{\mathcal{G}}
\newcommand{\X}{\mathcal{X}}
\newcommand{\N}{\mathcal{N}}
\newcommand{\HH}{\mathcal{H}}
\newcommand{\cS}{\mathcal{S}}
\newcommand{\YY}{\mathcal{Y}}
\newcommand{\R}{\mathbb{R}}

\newcommand{\cont}{\mathrm{span}}

\newcommand{\ImQ}{\operatorname{Im}Q}

\newcommand{\FTRLtag}{\B}

\title
[The Role of Preferences in Multi-Agent Online Learning]
{What Preferences Can\textemdash and Cannot\textemdash Predict\\
in Multi-Agent Online Learning}

\author
[O.~Abbadi]
{Omar Abbadi$^{c,\ast, \sharp}$}
\address{$^{c}$\,%
Corresponding author.}
\address{$^{\ast}$\,%
Univ. Mohammed VI Polytechnic, CMSIS, Moroccan Center for Game Theory, 11103 Rabat, Morocco.}
\email{omar.abbadi@um6p.ma}
\author
[R.~Laraki]
{Rida Laraki$^{\ast}$}
\email{rida.laraki@um6p.ma}
\author
[P.~Mertikopoulos]
{Panayotis Mertikopoulos$^{\sharp}$}
\address{$^{\sharp}$\,%
Univ. Grenoble Alpes, CNRS, Inria, Grenoble INP, LIG, 38000 Grenoble, France.}
\email{panayotis.mertikopoulos@imag.fr}


\thanks{
We are deeply grateful to Josef Hofbauer for many fruitful and enlightening discussions, leading in particular to the counterexample of \cref{sec:prefs-not-enough}.
We are likewise obliged to Oliver Biggar and Christos Papadimitriou for identifying and correcting a claim in an earlier version of this paper.}

\newacro{LHS}{left-hand side}
\newacro{RHS}{right-hand side}
\newacro{iid}[i.i.d.]{independent and identically distributed}
\newacro{lsc}[l.s.c.]{lower semi-continuous}
\newacro{usc}[u.s.c.]{upper semi-continuous}
\newacro{rv}[r.v.]{random variable}
\newacro{whp}{with high probability}
\newacro{wp1}[w.p.$1$]{with probability $1$}

\newacro{curb}{closed under rational behavior}
\newacro{club}{closed under better replies}
\newacro{clubness}{closedness under better replies}

\newacro{sclub}[s-club]{closed under strict better replies}
\newacro{sclubness}[s-clubness]{closedness under strict better replies}

\newacro{rad}{resilient to aggregate deviations}
\newacro{radness}{resilience to aggregate deviations}
\newacro{srad}[s-rad]{strictly resilient to aggregate deviations}
\newacro{sradness}[s-radness]{strict resilience to aggregate deviations}

\newacro{NE}{Nash equilibrium}
\newacroplural{NE}[NE]{Nash equilibria}

\newacro{FTRL}{follow-the-regularized-leader}

\begin{document}

\begin{abstract}
%
%
We examine the interplay between ordinal, preference-based solution concepts in games and the long-run behavior of game dynamics, asking in particular to what extent the combinatorial data of a game\textemdash its \emph{preference graph}\textemdash determine the outcomes of no-regret learning dynamics\textemdash such as \acdef{FTRL}.
In one direction, we show that the skeleton of every \emph{dynamically stable} set (\ie the set of pure profiles it contains) must also be \emph{preferentially stable}, that is, it must be closed under profitable deviations.
We then ask the converse question:
\emph{when do preferences determine the long-run behavior of the players' learning dynamics?}
We begin by showing that preferences characterize asymptotic stability in the case of \emph{subgames}\textemdash \ie subsets of pure profiles obtained by restricting players' action sets.
Beyond this case however, the equivalence between dynamic and preferential stability collapses:
concretely, we construct a three-player game with a preferentially stable set whose span is dynamically \emph{unstable}, showing in this way that preferences \emph{do not suffice} as a criterion of dynamic stability.
We then bridge this gap via the notion of \acli{radness}, an easy-to-check payoff-based condition that guarantees asymptotic stability of arbitrary spans of pure strategies.

\end{abstract}

\allowdisplaybreaks	
\acresetall	
\maketitle
\acused{iid}
\acused{LHS}
\acused{RHS}

\section{Introduction}
\label{sec:introduction}

\begin{figure*}[t]
\renewcommand{\thisfigscale}{2.5}
\renewcommand{\thisfigheight}{28ex}
\centering
\begin{subfigure}[c]{.49\linewidth}
\centering
\makebox[\linewidth][c]{%
\resizebox{.42\linewidth}{!}{
%
%
\begin{tikzpicture}[
	scale=\thisfigscale,
	line join=round,
	line cap=round,
	edge/.style={semithick,draw=black},
	sinkspan/.style={draw=red!70!black,very thick},
	midorient/.style={-{Stealth},thick,draw=black},
	sinkorient/.style={{Stealth}-,thick,draw=MyDarkRed},
	sinkorientbi/.style={{Stealth}-{Stealth},thick,draw=MyDarkRed}
]

\coordinate (BLB) at (0,0);
\coordinate (BRB) at (1,0);
\coordinate (BRT) at (1,1);
\coordinate (BLT) at (0,1);

\coordinate (FLB) at ( 0.40,0.30);
\coordinate (FRB) at ( 1.40,0.30);
\coordinate (FRT) at ( 1.40,1.30);
\coordinate (FLT) at ( 0.40,1.30);

\fill[MyLightRed,opacity=0.25] (FLB)--(FRB)--(FRT)--(FLT)--cycle;
\fill[MyLightRed,opacity=0.25] (BRB)--(FRB)--(FRT)--(BRT)--cycle;
\fill[MyLightRed,opacity=0.25] (BLT)--(BRT)--(FRT)--(FLT)--cycle;

\draw[edge] (FLB)--(FRB)--(FRT)--(FLT)--cycle;
\draw[edge] (BLB)--(BRB)--(BRT)--(BLT)--cycle;
\draw[edge] (FLB)--(BLB);
\draw[edge] (FRB)--(BRB);
\draw[edge] (FRT)--(BRT);
\draw[edge] (FLT)--(BLT);

\draw[sinkspan] (FLB)--(FRB)--(FRT)--(FLT)--cycle;
\draw[sinkspan] (FRB)--(BRB)--(BRT)--(FRT)--cycle;
\draw[sinkspan] (FLT)--(FRT)--(BRT)--(BLT)--cycle;

\draw[midorient] ($(BLB)!0.48!(FLB)$) -- ($(BLB)!0.52!(FLB)$);
\draw[sinkorient] ($(BLT)!0.48!(FLT)$) -- ($(BLT)!0.52!(FLT)$);
\draw[sinkorient] ($(FRB)!0.48!(BRB)$) -- ($(FRB)!0.52!(BRB)$);
\draw[sinkorient] ($(FRT)!0.48!(BRT)$) -- ($(FRT)!0.52!(BRT)$);

\draw[midorient] ($(BLB)!0.48!(BRB)$) -- ($(BLB)!0.52!(BRB)$);
\draw[sinkorient] ($(FLB)!0.48!(FRB)$) -- ($(FLB)!0.52!(FRB)$);
\draw[sinkorient] ($(BRT)!0.48!(BLT)$) -- ($(BRT)!0.52!(BLT)$);
\draw[sinkorient] ($(FLT)!0.48!(FRT)$) -- ($(FLT)!0.52!(FRT)$);

\draw[midorient]($(BLB)!0.48!(BLT)$) -- ($(BLB)!0.52!(BLT)$);
\draw[sinkorient] ($(BRB)!0.48!(BRT)$) -- ($(BRB)!0.52!(BRT)$);
\draw[sinkorient] ($(FLT)!0.48!(FLB)$) -- ($(FLT)!0.52!(FLB)$);

\foreach \V in {FLB,FRB,FRT,FLT,BRB,BRT,BLT}{
\fill[red!70!black] (\V) circle (0.9pt);
}
\fill[black] (BLB) circle (0.9pt);

\end{tikzpicture}}%
\hfill
\includegraphics[width=.56\linewidth,height=\thisfigheight,keepaspectratio]{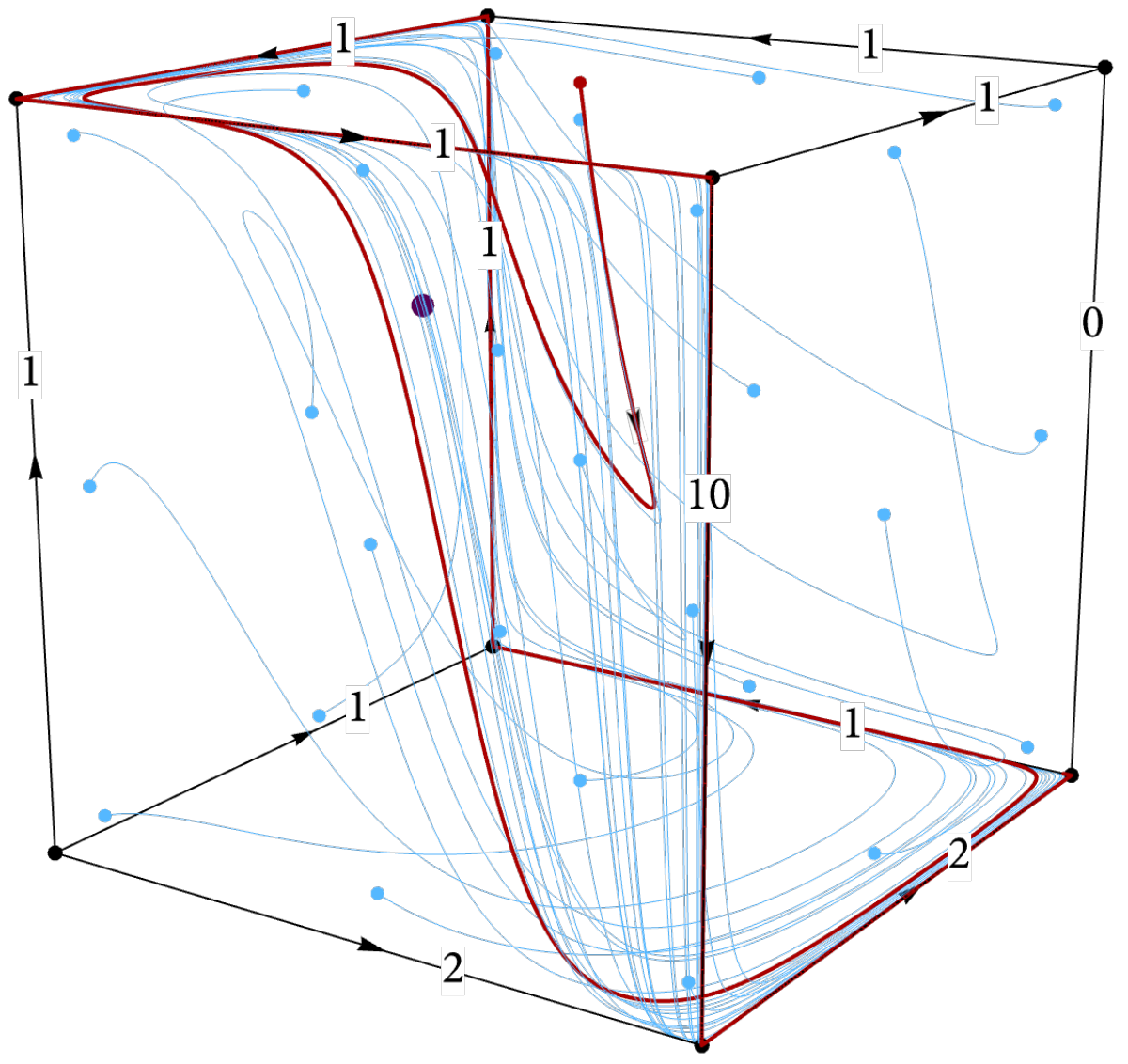}}%
\caption{When preferences and behavior \emph{do not} align.}
\label{fig:misalign}
\end{subfigure}%
\hfill
\begin{subfigure}[c]{.49\linewidth}
\centering
\makebox[\linewidth][c]{%
\resizebox{.42\linewidth}{!}{
%
%
\begin{tikzpicture}[
scale=\thisfigscale,
line join=round,
line cap=round,
edge/.style={semithick,draw=black},
rededge/.style={very thick,draw=MyDarkRed},
midorient/.style={-{Stealth},thick,draw=black},
sinkorient/.style={-Stealth,thick,draw=MyDarkRed}
]
\coordinate (FLB) at (0,0);
\coordinate (FRB) at (1,0);
\coordinate (FRT) at (1,1);
\coordinate (FLT) at (0,1);
\coordinate (BLB) at (0.40,0.30);
\coordinate (BRB) at (1.40,0.30);
\coordinate (BRT) at (1.40,1.30);
\coordinate (BLT) at (0.40,1.30);

\draw[edge]    (BLB)--(BRB)--(BRT);
\draw[rededge] (BRT)--(BLT)--(BLB);

\draw[rededge] (FLB)--(BLB);
\draw[edge]    (FRB)--(BRB);
\draw[rededge] (FRT)--(BRT);
\draw[edge]    (FLT)--(BLT);

\draw[rededge] (FLB)--(FRB)--(FRT);
\draw[edge]    (FRT)--(FLT)--(FLB);

\draw[midorient]
  ($(FLT)!0.48!(FRT)$) -- ($(FLT)!0.52!(FRT)$);
\draw[sinkorient]
  ($(BLT)!0.48!(BRT)$) -- ($(BLT)!0.52!(BRT)$);

\draw[midorient]
  ($(BRB)!0.48!(BLB)$) -- ($(BRB)!0.52!(BLB)$);
\draw[sinkorient]
  ($(FRB)!0.48!(FLB)$) -- ($(FRB)!0.52!(FLB)$);

\draw[midorient]
  ($(FLT)!0.48!(FLB)$) -- ($(FLT)!0.52!(FLB)$);
\draw[sinkorient]
  ($(FRT)!0.48!(FRB)$) -- ($(FRT)!0.52!(FRB)$);

\draw[sinkorient]
  ($(BLB)!0.48!(BLT)$) -- ($(BLB)!0.52!(BLT)$);
\draw[midorient]
  ($(BRB)!0.48!(BRT)$) -- ($(BRB)!0.52!(BRT)$);

\draw[midorient]
  ($(FLT)!0.48!(BLT)$) -- ($(FLT)!0.52!(BLT)$);
\draw[sinkorient]
  ($(BRT)!0.48!(FRT)$) -- ($(BRT)!0.52!(FRT)$);

\draw[sinkorient]
  ($(FLB)!0.48!(BLB)$) -- ($(FLB)!0.52!(BLB)$);
\draw[midorient]
  ($(BRB)!0.48!(FRB)$) -- ($(BRB)!0.52!(FRB)$);

\foreach \V in {FLB,FRB,FRT,BLB,BRT,BLT}{
  \fill[MyDarkRed] (\V) circle (0.9pt);
}
\foreach \V in {FLT,BRB}{
  \fill[black] (\V) circle (0.9pt);
}
\end{tikzpicture}}%
\hfill
\includegraphics[width=.56\linewidth,height=\thisfigheight,keepaspectratio]{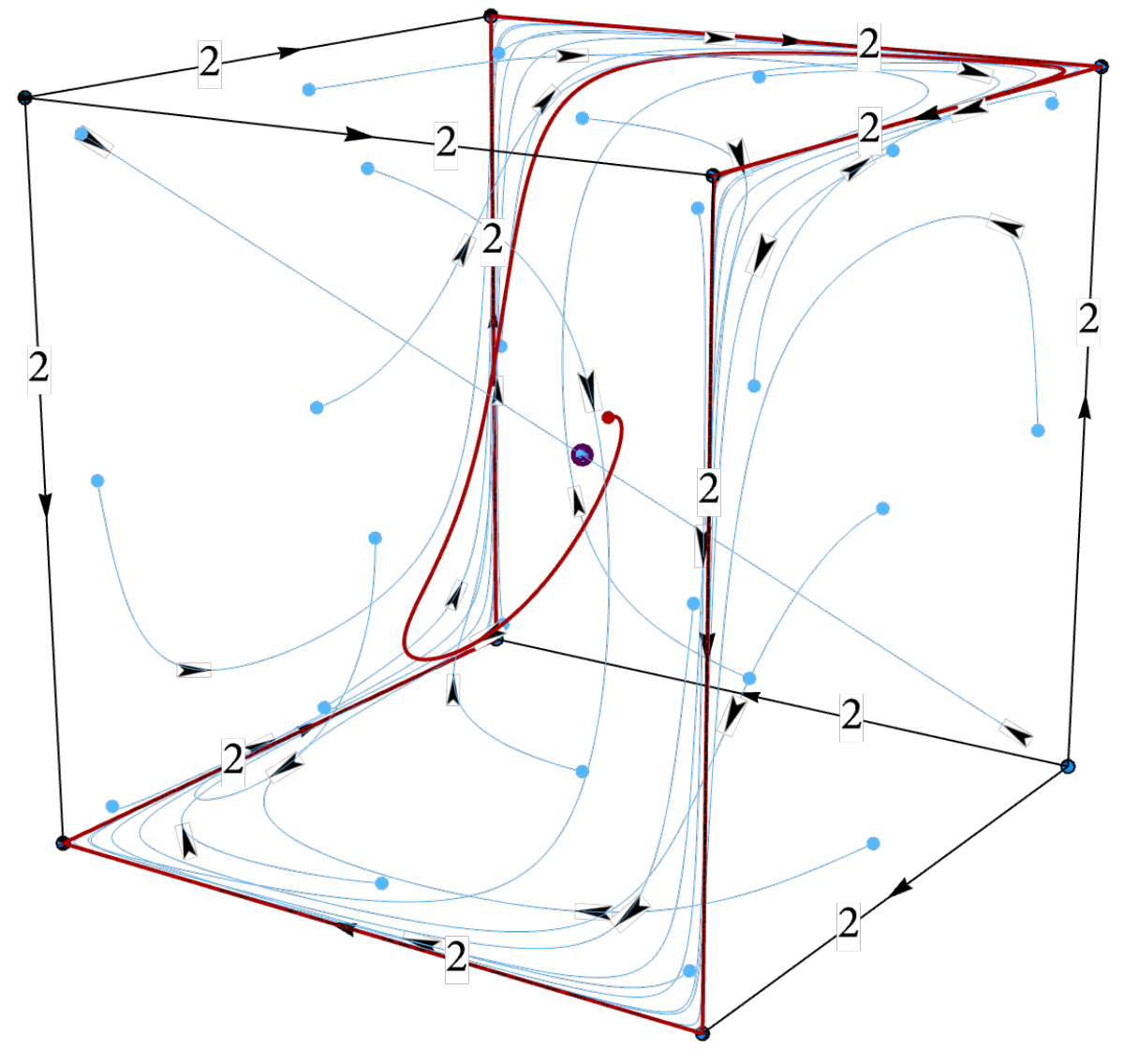}}%
\caption{When preferences and behavior \emph{do} align.}
\label{fig:align}
\end{subfigure}%
\caption{In each of these $2\times2\times2$ games, vertices are pure profiles and arrows indicate the direction of profitable unilateral deviations.
The part of the graph that is highlighted in red is \acf{club}, so players have no \emph{preferential} incentive to move away from these sets.
Translating this intuition to the players' dynamic behavior fails on the left:
trajectories starting arbitrarily close to the top face drift away toward the bottom face, contradicting the preferential viewpoint.
On the right the preferential and dynamical viewpoints are aligned: trajectories starting anywhere near the preferentially closed 6-edge cycle (highlighted in red) converge to it.}
\label{fig:examples}
\end{figure*}
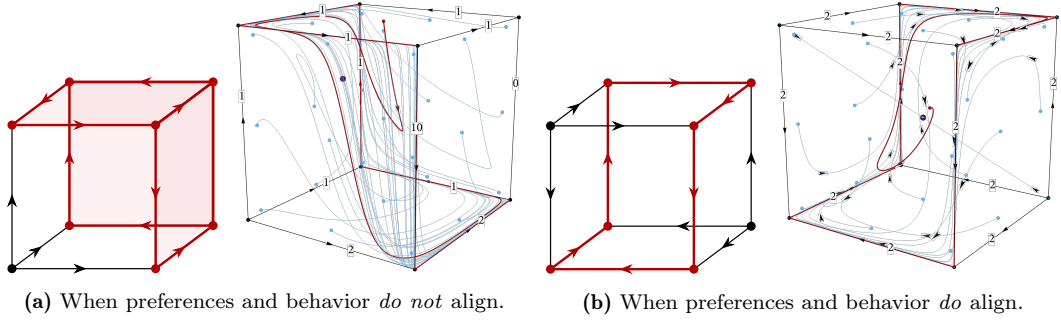

\begin{quote}
\itshape
\centering
Can players learn to behave rationally\\
by adapting to feedback from their environment?
\end{quote}
\medskip

This question has been at the forefront of non-cooperative game theory and multi-agent learning \cite{FL98,CBL06,NRTV07}, where interacting agents are commonly modeled through their governing dynamics \citep{HS98,San10}.
In particular, in many modern applications of game theory\textemdash from recommender platforms and auctions to large-scale multi-agent reinforcement learning \citep{MHSY+13,BLM20,GJV22,GSP19,OPPT+19}\textemdash the classical premises of full rationality, knowledge of the game, equilibrium computation, and consistently optimal play are unrealistic.
As such, it is natural to take the above question as a starting point and ask what it implies for the players' long-run behavior in a game-theoretic context.

In view of this, we focus on the dynamics of \acdef{FTRL} \citep{Haz16,SS11,SSS06}, a class of no-regret learning procedures which has become the gold standard for learning in games \citep{CBL06}.
Under \ac{FTRL}, players react to \emph{cumulative} scores rather than instantaneous payoffs \citep{MS16,CGM15}, so correlations emerge endogenously as the dynamics unfold, prompting various forms of conflict and cooperation \citep{Han57,LM13,FPW24}.
At the same time, this richness makes it difficult to characterize the dynamics' asymptotic behavior in \emph{general} finite games \citep{AFP21}:
indeed, while strict \aclp{NE} are precisely the locally stable and attracting points of regularized learning \citep{CGM15,FVGL+20,GVM21,GVM21b,LGMB24,LMBB25b}, they need not govern global behavior, and in their absence learning may lead to recurrence or chaos \citep{MPP18,CP19,NBP20,LMPP+24,LMP24,PS14,HMC21,MHC24}.
Even worse, several impossibility results rule out the existence of uncoupled dynamics that converge to \acl{NE} in all games \citep{HMC03,BHS12,MPPS23}.

These considerations reveal that the relevant asymptotic objects are often \emph{set-valued}, not necessarily pointwise.
Accordingly, in lieu of asking when learning converges to a particular equilibrium, we ask which \emph{regions of the strategy space} can arise as stable long-run outcomes, and how the structural features of the game affect them.

Along these lines, a common modeling stance in game theory is that the primitive object of interest is the players' \emph{preferences}, whereas payoffs are merely a numerical representation thereof \citep{Deb54,BS25,MSZ20,Osb04}.
Motivated by this viewpoint, \citet{PP19} conjectured that the attractors of the replicator dynamics\textemdash the continuous-time limit of exponential weights \citep{TJ78,HSV09,Sor09,MM10}\textemdash are all determined by a directed graph encoding profitable unilateral deviations.
This conjecture was made precise in subsequent work by \citet{BS23}, who, exploiting invariance properties of the replicator flow, showed that every  attractor of the dynamics must contain certain \enquote{terminal components} of said graph.
Nonetheless, even though this inclusion shows that preferences and learning are intimately linked, it does not yield sufficient predictive power;
more precisely, it remains unclear which sets actually \emph{are} stable outcomes of learning, and to what extent they can be characterized by the game's preference structure alone.

\para{Our contributions in the context of related work}

Our paper echoes a recent shift in perspective by \citet{BP25} to the effect that the \enquote{preferences only} intuition may \emph{fail} in general:
concretely, we show that, in a wide class of regularized learning dynamics, ordinal information need not determine stability, and the game's \emph{actual} payoff values can be essential for predicting the long-run behavior of the dynamics.
This is particularly relevant in the context of applications of game theory to reinforcement learning and machine learning, where a designer may need to instantiate \emph{numerical} rewards in order to implement an intended \emph{ordinal} specification\textemdash as in the case of preference learning \citep{CLBM+17,HMRA16} or reward specification \citep{AOSC+16} for AI safety\textemdash and different cardinal realizations of the same ordering can lead to drastically different outcomes in the long run.

In more detail, our paper aims to develop a theory for regularized dynamics that delineates when preferences do and do not pin down asymptotic outcomes, and how to recover stability and attraction when they do not.
We do so through the game's \emph{preference graph}~\citep{BS25}, the directed graph on pure profiles whose arcs encode the directions of (weakly) profitable unilateral deviations, and which therefore retains only \emph{ordinal} incentive information.
Our guiding premise is that, in spite of this deliberate minimality, preferences impose strong constraints on stable learning outcomes, prompting the following natural question:
\medskip
\begin{center}
\emph{To what extent do player preferences determine\\the long-run behavior of their learning dynamics?}
\end{center}
\medskip

Our results show that the preference graph places unavoidable restrictions on the stable outcomes of regularized learning while also identifying regimes in which ordinal information alone cannot determine such outcomes.
Concretely:
\begin{enumerate}
\item
For a broad class of dynamics, we show that the pure profiles contained in any asymptotically stable set\textemdash the \emph{skeleton} of the set\textemdash must be \acdef{club}, \ie there are no profitable unilateral deviations leading outside this set.
As a follow-up to this core result, we also show that every attractor must contain the \emph{full} mixed region spanned by each connected set of pure profiles it contains.
Thus, if the preference graph of the game is strongly connected, the dynamics of \ac{FTRL} do not admit a proper attractor.
\item
For \emph{subgames}, we show that preferences do suffice:
a subgame is asymptotically stable if and only if the subgame itself is \acl{club}, complementing in this way previous work by \citet{RW95} on the replicator dynamics, and more recent results by \cite{BS23,CLM25} for stochastic variants of \ac{FTRL}.
As a special case, this also yields a sharp characterization of attractors in weakly acyclic games, extending results previously known only for potential games \citep{MerSan18}.
\item
Beyond subgames, we construct an example where a set of pure profiles is \acl{club} but its span is nevertheless unstable, even under full-support initialization, thus establishing that preferential and dynamical stability need not coincide under regularized learning.
\item
Finally, we introduce the notion of sets that are \acdef{rad}, a payoff-based condition which guarantees asymptotic stability.
This is a setwise generalization of pure \acl{NE} and a cardinal refinement of \acl{clubness}, and, to the best of our knowledge, it is the first such condition that applies to arbitrary games and arbitrary sets of pure strategies.
\acused{radness}
\end{enumerate}

All in all, the above can be seen as a preference-based toolbox for predicting the long-run behavior of regularized learning in games, clarifying in particular how the preference graph \emph{constrains}\textemdash but does not always \emph{determine}\textemdash the outcome of the dynamics.
In this regard, our results should be compared and contrasted to the conjectured correspondence \citep{PP19} between minimal \ac{club} sets\textemdash also known as \emph{sink equilibria} \citep{GMV05}\textemdash and minimal attractors of the replicator dynamics.
\citet{BP25} recently disproved this conjecture by showing that a sink equilibrium containing a ``local source'' cannot be an attractor of the replicator dynamics;
they then conjectured that this condition is also necessary, \ie \emph{any} sink equilibrium which is not an attractor of the replicator dynamics must contain a local source.
Our counterexample disproves a more stringent version of this conjecture to the effect that the orbits of \ac{FTRL} may drift away from a sink equilibrium which does not contain a \emph{pure} local source;
however, our example still contains a mixed local source, leaving the original conjecture open.%
\footnote{We are grateful to Oliver Biggar for pointing out this discrepancy in an earlier version of our paper.}

\citet{BP25} also introduced a payoff-based sufficient condition for asymptotic stability of sink equilibria in two-player games under the replicator dynamics.
This condition does not directly extend to arbitrary numbers of players, in contrast to \ac{radness}, which concerns arbitrary games and arbitrary sets of pure strategies.
Finally, in concurrent work, \citet{BP26} showed that the span of a \ac{club} cycle in the preference graph
(a cycle for which each vertex has a unique outgoing edge)
is an attractor of the replicator dynamics.%
\footnote{The first version of \cite{BP26} appeared when the current paper was under review;
it is also worth noting that the extension of the results of \cite{BP26} to more general dynamics remains an open question.}
This result intersects non-trivially with our own result on the asymptotic stability of \ac{rad} sets:
neither implies the other, and they can both be used to show the asymptotic stability of specific sink equilibria like the six-cycle of best responses in Jordan's Matching Pennies game\textemdash which had only been established using fairly ad hoc techniques in the past \cite{GH95}.
This gives another instance where preferences are enough for dynamic stability, but also shows that
a complete understanding of the attractors of learning in games remains elusive.

\section{Preliminaries}
\label{sec:prelims}

\para{Game and strategies}

We consider a finite normal-form game $\mathcal G$, specified by a player set $\N=\{1,\dots,n\}$, finite action sets $\A_i$ and payoff functions $u_i:\A\to\R$ for each $i\in\N$, where $\A:=\prod_{i\in\N}\A_i$ is the set of pure profiles.
A pure action profile is denoted by $\alpha=(\alpha_1,\dots,\alpha_n)\in\A$, and for a given player $i$ we write $\alpha_{-i}=(\alpha_j)_{j\neq i}$ for the opponents' component of $\alpha$.
In addition to pure actions, players may randomize over $\A_i$.
Accordingly, player $i$'s mixed strategy space is the simplex $\X_i=\Delta(\A_i)$, the set of probability distributions over $\A_i$, and we write $\X:=\prod_{i\in\N}\X_i$ for the set of mixed profiles $x=(x_1,\dots,x_n)$, which we call the game's \emph{strategy space}.
For $\alpha_i\in\A_i$, the coordinate $x_{i\alpha}$ denotes the probability assigned to $\alpha_i$ by $x_i$,%
and we set $x_{-i}=(x_j)_{j\neq i}$.
In addition, we write $x_\alpha \defeq \prod_{i\in\N}x_{i\alpha}$ for the probability of pure profile $\alpha$ being played.
We also use the notation $\X_i^\circ$ for the relative interior of $\X_i$ (i.e., the set of mixed actions with full support) and similarly, $\X^\circ:=\prod_i \X_i^\circ$ denote the relative interior of $\X$.
By a slight abuse of notation, we identify each pure action $\alpha_i$ with the vertex $e_{\alpha_i}$ (\ie the mixed strategy assigning probability $1$ to $\alpha_i$), and likewise each pure profile $\alpha$ with the vertex $e_\alpha$.
Finally, for $x_i \in \X_i$ and $x \in \X$, we define the \emph{support}:
\[
\supp(x_i) := \{ \alpha_i \in \A_i : x_{i\alpha} \neq 0 \}
\]
and $\supp(x):=\prod_i\supp(x_i).$

\para{Payoff field and equilibria} The payoff functions extend multilinearly to mixed profiles via expectation under the product distribution induced by $x$:
\[
u_i(x)=\sum_{\alpha\in\A}x_\alpha u_i(\alpha).
\]
Denoting $\YY_i:=\R^{\A_i}$ and $\YY:=\prod_{i\in\N}\YY_i$ for the game's \emph{payoff space}, we consider the \emph{payoff field}
$v:\X\to\YY$, defined componentwise by
\[
v_{i\alpha}(x)=u_i(\alpha_i,x_{-i}),
\quad\text{so that}\quad
\langle v_i(x),x_i\rangle=u_i(x).
\]
Equivalently, since $u_i(\,\cdot\,,x_{-i})$ is linear in $x_i$, the vector $v_i(x)$ coincides with the gradient of $u_i$ with respect to $x_i$:
\[
v_i(x)=\nabla_{x_i}u_i(x).
\]
A profile $x^\ast\in\X$ is a \acdef{NE} if no player can improve their payoff by a unilateral deviation from $x^\ast$, i.e.,
$u_i(x^\ast)\ge u_i(x_i,x_{-i}^\ast)$ for all $x_i\in\X_i$ and all $i\in\N$.
A \emph{strict} \acl{NE} is a pure profile $\alpha^\ast\in\A$ such that each player strictly prefers $\alpha_i^\ast$ against $\alpha_{-i}^\ast$ to any other pure action, namely
$u_i(\alpha^\ast)>u_i(\alpha_i,\alpha_{-i}^\ast)$ for all $\alpha_i\neq \alpha_i^\ast$.

\section{Preferences, learning and stability}
\label{sec:definitions}

In this section, we define the basics of preference-based solution concepts and regularized learning in games.

\subsection{Preference graphs and ordinal stability}

We begin with some basics concerning the players' preferences and related concepts.

\para{Preference graph}

We encode ordinal incentives via the game's \emph{preference graph}.
Its vertex set is $\A$, and it contains an arc $\alpha\to\alpha'$ whenever $\alpha$ and $\alpha'$ differ only in a single player's action (say player $i$, in which case $\alpha$ and $\alpha'$ are \emph{$i$-comparable}) and the deviation is weakly profitable for $i$, i.e., $u_i(\alpha')\ge u_i(\alpha)$.
Thus, for each pure profile, the outgoing edges indicate which comparable profiles are preferred by each player, hence inducing a (partial) order on the pure profiles.
If $\alpha \to \alpha'$ and $\alpha' \to \alpha$ we say there is a \emph{tie} and we denote it $\alpha \leftrightarrow \alpha'$.
It is clear that a game having no ties is a \emph{generic} property, meaning it holds for an open dense subset of games.
Put differently, one can always find an arbitrarily small perturbation of the game which resolves ties.

\begin{figure}[t]
\centering
\begin{tabular}{c|ccc}
 & $\mathrm A$ & $\mathrm B$ & $\mathrm C$ \\ \hline
$\mathrm A$ & $(0,0)$ & $(2,1)$ & $(1,2)$ \\
$\mathrm B$ & $(1,2)$ & $(0,0)$ & $(2,1)$ \\
$\mathrm C$ & $(2,1)$ & $(1,2)$ & $(0,0)$
\end{tabular}
\\
%
%
\begin{tikzpicture}[scale=1, transform shape, >=Stealth, line cap=round,
  node/.style={circle,draw=black,minimum size=6.2mm,inner sep=0pt,font=\small},
  sink/.style={node,draw=black},
  edge/.style={->,draw=black,line width=0.9pt},
  sinkedge/.style={->,draw=black,line width=0.9pt}
]
\node[sink] (n12) at (0,1.85)      {$(\mathrm A,\mathrm B)$};
\node[sink] (n13) at (1.65,0.93)   {$(\mathrm A,\mathrm C)$};
\node[sink] (n23) at (1.65,-0.93)  {$(\mathrm B,\mathrm C)$};
\node[sink] (n21) at (0,-1.85)     {$(\mathrm B, \mathrm A)$};
\node[sink] (n31) at (-1.65,-0.93) {$(\mathrm C,\mathrm A)$};
\node[sink] (n32) at (-1.65,0.93)  {$(\mathrm C,\mathrm B)$};

\draw[sinkedge] (n12) -- (n13);
\draw[sinkedge] (n13) -- (n23);
\draw[sinkedge] (n23) -- (n21);
\draw[sinkedge] (n21) -- (n31);
\draw[sinkedge] (n31) -- (n32);
\draw[sinkedge] (n32) -- (n12);

\node[node] (d33) at (0,3.25)      {$(\mathrm C,\mathrm C)$};
\node[node] (d11) at (-3.25,-2.25) {$(\mathrm A,\mathrm A)$};
\node[node] (d22) at (3.25,-2.25)  {$(\mathrm B,\mathrm B)$};

\draw[edge,bend left=8]  (d11) to (n31);
\draw[edge,bend left=2]  (d11) to (n21);
\draw[edge,bend right=38] (d11) to (n12);
\draw[edge,bend right=20] (d11) to (n13);

\draw[edge,bend right=2] (d22) to (n21);
\draw[edge,bend right=8] (d22) to (n23);
\draw[edge,bend left=38] (d22) to (n12);
\draw[edge,bend left=20] (d22) to (n32);

\draw[edge] (d33) .. controls (-1.20,2.90) and (-2.55,2.05) .. (n32);
\draw[edge] (d33) .. controls ( 1.20,2.90) and ( 2.55,2.05) .. (n13);

\draw[edge] (d33) .. controls (-3.20,2.85) and (-3.20,-0.10) .. (n31);
\draw[edge] (d33) .. controls ( 3.20,2.85) and ( 3.20,-0.10) .. (n23);

\end{tikzpicture}
\caption{Payoffs and preferences in \emph{Shapley's game} \cite{Sha64}.
The cycle of six profiles in the center\textemdash the \enquote{hexagon}\textemdash is the game's unique proper \ac{club} set;
since it is minimal, it is a sink equilibrium.
See \cref{app:preferences} for more examples and details.}
\label{fig:shapley-main}
\end{figure}
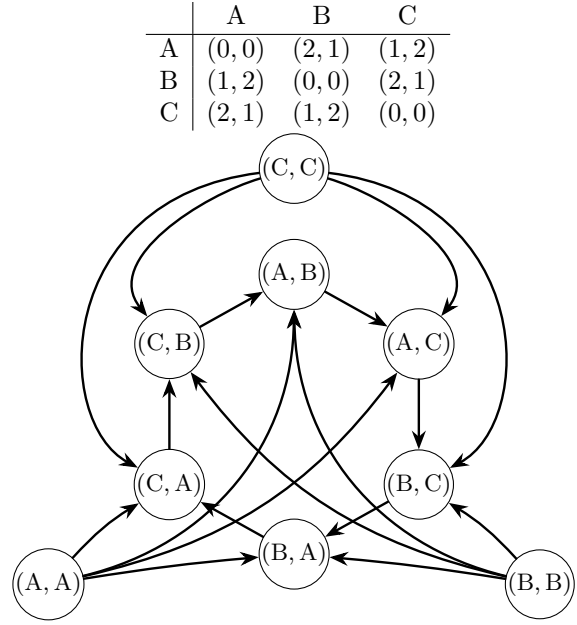

\para{Ordinal stability and connectedness}

We now define the notion of ``stability under unilateral improvements''.
\begin{definition}
We say $\HH \subseteq \A$ is \acdef{club} if it has no outgoing arc to $\A\setminus\HH$.
Similarly, we say that $\HH$ is \acdef{sclub} if it has no outgoing arc unless it's a tie.
\end{definition}
\acused{clubness}
\acused{sclubness}

Equivalently, a \ac{club} set is stable under \emph{weakly profitable} unilateral deviation:
if $\alpha\in \HH$ and $\alpha \to \beta$ then $\beta \in \HH$ as well.
Likewise, an \ac{sclub} set is stable under \emph{strictly profitable} unilateral deviations:
if $\alpha\in \HH$ and $\alpha \to \beta$ then either $\beta \in \HH$ or $\beta \to \alpha$ as well.
For concision, we will say a \ac{club}\,/\,\ac{sclub} set is \emph{proper} if it is a nonempty proper subset of $\A$.

We next turn from ordinal stability to a complementary structural feature: \emph{connectedness}.

\begin{definition}
A set $\HH\subseteq\A$ is \emph{strongly connected} if any two vertices in $\HH$ can be joined by a directed path (a unilateral improvement path) that remains in $\HH$.
\end{definition}

A nonempty set of pure profiles which is both \ac{club} and strongly connected is called a \emph{sink equilibrium}, it is then also a \emph{minimal} \ac{club} set, that is, a nonempty \ac{club} set which contains no proper \ac{club} set.

\para{Spans, subgames and skeletons}

Given any set of pure profiles $\HH\subseteq\A$, we define its \emph{span} as the region of mixed profiles supported on $\HH$, \viz%
\footnote{This has also been called the \emph{content} of the set \citep{BS23,BS24,BP25}.}
\[
\cont(\HH)
	\defeq \setdef{\strat\in\strats}{\strat_{\pure} = 0 \text{ for all } \alpha\not\in\HH}
\]
A set $\B\subseteq\A$ is a \emph{subgame} if it has a product structure $\B=\prod_{i\in\N}\B_i$ with nonempty $\B_i\subseteq\A_i$.
Its span is then the corresponding \emph{face} $\F= \cont(\B)=\prod_{i\in\N}\Delta(\B_i)$.
More generally, for any $\HH \subseteq \A$, its span corresponds to the union of faces of the strategy space for which all vertices are in $\HH$.
Finally, given a set $\set \subseteq \X$ we define its \emph{skeleton} $\skel(\set)$ as the set of pure profiles it contains, \ie
\[
\skel(\set)
	\defeq \set \cap \A.
\]
For an illustration, \cf \cref{fig:span}.

\begin{remark}
\label{rem:skel-span}
Note that $\skel$ is a left inverse of $\cont$ in the sense that $\skel(\cont(\HH)) = \HH$ for all $\HH\subseteq\pures$.
The converse, however, is not true:
for example, $\cont(\skel(\set)) = \varnothing$ if $\set\subseteq\relint\strats$ is a set of fully mixed strategies.
\endenv
\end{remark}

\begin{remark}
\label{rmrk:singletons}
Importantly, a pure strategy profile $\pure\in\pures$ is a \emph{pure} \textpar{resp.~\emph{strict}} \acl{NE} \emph{if and only if} $\braces{\pure}$ is \ac{sclub} \textpar{resp.~\ac{club}}.
Note in particular the inversion of adjectives due to the ``closedness'' proviso in the definition of \ac{club}\,/\,\ac{sclub} sets.
\endenv
\end{remark}

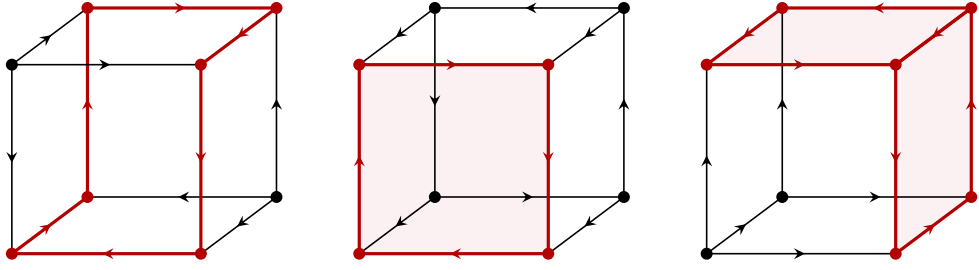
\begin{figure}[tbp]
\renewcommand{\thisfigscale}{2.5}
\centering
%
%
\begin{tikzpicture}[
	scale=\thisfigscale,
	line join=round,
	line cap=round,
	edge/.style={semithick,draw=black},
	sinkspan/.style={draw=MyDarkRed,very thick},
	midorient/.style={-{stealth},thick,draw=black},
	sinkorient/.style={-{stealth},thick,draw=MyDarkRed}
]

\coordinate (FLB) at (0,0);
\coordinate (FRB) at (1,0);
\coordinate (FRT) at (1,1);
\coordinate (FLT) at (0,1);

\coordinate (BLB) at (0.40,0.30);
\coordinate (BRB) at (1.40,0.30);
\coordinate (BRT) at (1.40,1.30);
\coordinate (BLT) at (0.40,1.30);

\draw[edge] (BLT)--(BRT)--(BRB)--(BLB)--cycle;
\draw[sinkspan] (BLB)--(BLT)--(BRT);

\draw[edge] (FLT)--(BLT);
\draw[edge] (FRT)--(BRT);
\draw[edge] (FRB)--(BRB);
\draw[edge] (FLB)--(BLB);
\draw[sinkspan] (BRT)--(FRT);
\draw[sinkspan] (FLB)--(BLB);

\draw[edge] (FLT)--(FRT)--(FRB)--(FLB)--cycle;
\draw[sinkspan] (FRT)--(FRB)--(FLB);

\draw[midorient] ($(FLT)!0.48!(FRT)$) -- ($(FLT)!0.52!(FRT)$);
\draw[sinkorient] ($(BLT)!0.48!(BRT)$) -- ($(BLT)!0.52!(BRT)$);
\draw[midorient] ($(BRB)!0.48!(BLB)$) -- ($(BRB)!0.52!(BLB)$);
\draw[sinkorient] ($(FRB)!0.48!(FLB)$) -- ($(FRB)!0.52!(FLB)$);

\draw[midorient] ($(FLT)!0.48!(FLB)$) -- ($(FLT)!0.52!(FLB)$);
\draw[sinkorient] ($(FRT)!0.48!(FRB)$) -- ($(FRT)!0.52!(FRB)$);
\draw[sinkorient] ($(BLB)!0.48!(BLT)$) -- ($(BLB)!0.52!(BLT)$);
\draw[midorient] ($(BRB)!0.48!(BRT)$) -- ($(BRB)!0.52!(BRT)$);

\draw[midorient] ($(FLT)!0.48!(BLT)$) -- ($(FLT)!0.52!(BLT)$);
\draw[sinkorient] ($(BRT)!0.48!(FRT)$) -- ($(BRT)!0.52!(FRT)$);
\draw[sinkorient] ($(FLB)!0.48!(BLB)$) -- ($(FLB)!0.52!(BLB)$);
\draw[midorient] ($(BRB)!0.48!(FRB)$) -- ($(BRB)!0.52!(FRB)$);

\foreach \V in {FLB,FRB,FRT,BLB,BRT,BLT}{
\fill[MyDarkRed] (\V) circle (0.9pt);
}
\foreach \V in {FLT,BRB}{
\fill[black] (\V) circle (0.9pt);
}

\end{tikzpicture}
\qquad
%
%
\begin{tikzpicture}[
	scale=\thisfigscale,
	line join=round,
	line cap=round,
	edge/.style={semithick,draw=black},
	sinkspan/.style={draw=MyDarkRed,very thick},
	midorient/.style={-{stealth},thick,draw=black},
	sinkorient/.style={-{stealth},thick,draw=MyDarkRed}
]

\coordinate (FLB) at (0,0);
\coordinate (FRB) at (1,0);
\coordinate (FRT) at (1,1);
\coordinate (FLT) at (0,1);

\coordinate (BLB) at (0.40,0.30);
\coordinate (BRB) at (1.40,0.30);
\coordinate (BRT) at (1.40,1.30);
\coordinate (BLT) at (0.40,1.30);

\fill[MyLightRed,opacity=0.25] (FLT)--(FRT)--(FRB)--(FLB)--cycle;

\draw[edge] (FLT)--(FRT)--(FRB)--(FLB)--cycle;
\draw[edge] (BLT)--(BRT)--(BRB)--(BLB)--cycle;
\draw[edge] (FLT)--(BLT);
\draw[edge] (FRT)--(BRT);
\draw[edge] (FRB)--(BRB);
\draw[edge] (FLB)--(BLB);

\draw[sinkspan] (FLT)--(FRT)--(FRB)--(FLB)--cycle;

\draw[sinkorient] ($(FLT)!0.48!(FRT)$) -- ($(FLT)!0.52!(FRT)$);
\draw[sinkorient] ($(FRB)!0.48!(FLB)$) -- ($(FRB)!0.52!(FLB)$);
\draw[midorient] ($(BRT)!0.48!(BLT)$) -- ($(BRT)!0.52!(BLT)$);
\draw[midorient] ($(BLB)!0.48!(BRB)$) -- ($(BLB)!0.52!(BRB)$);

\draw[sinkorient] ($(FLB)!0.48!(FLT)$) -- ($(FLB)!0.52!(FLT)$);
\draw[sinkorient] ($(FRT)!0.48!(FRB)$) -- ($(FRT)!0.52!(FRB)$);
\draw[midorient] ($(BLT)!0.48!(BLB)$) -- ($(BLT)!0.52!(BLB)$);
\draw[midorient] ($(BRB)!0.48!(BRT)$) -- ($(BRB)!0.52!(BRT)$);

\draw[midorient] ($(BLT)!0.48!(FLT)$) -- ($(BLT)!0.52!(FLT)$);
\draw[midorient] ($(BRT)!0.48!(FRT)$) -- ($(BRT)!0.52!(FRT)$);
\draw[midorient] ($(BLB)!0.48!(FLB)$) -- ($(BLB)!0.52!(FLB)$);
\draw[midorient] ($(BRB)!0.48!(FRB)$) -- ($(BRB)!0.52!(FRB)$);

\foreach \V in {FLB,FRB,FRT,FLT}{
\fill[MyDarkRed] (\V) circle (0.9pt);
}
\foreach \V in {BLB,BRB,BRT,BLT}{
\fill[black] (\V) circle (0.9pt);
}

\end{tikzpicture}
\qquad
%
%
\begin{tikzpicture}[
	scale=\thisfigscale,
	line join=round,
	line cap=round,
	edge/.style={semithick,draw=black},
	sinkspan/.style={draw=MyDarkRed,very thick},
	midorient/.style={-{stealth},thick,draw=black},
	sinkorient/.style={-{stealth},thick,draw=MyDarkRed}
]

\coordinate (FLB) at (0,0);
\coordinate (FRB) at (1,0);
\coordinate (FRT) at (1,1);
\coordinate (FLT) at (0,1);

\coordinate (BLB) at (0.40,0.30);
\coordinate (BRB) at (1.40,0.30);
\coordinate (BRT) at (1.40,1.30);
\coordinate (BLT) at (0.40,1.30);

\fill[MyLightRed,opacity=0.25] (FLT)--(FRT)--(BRT)--(BLT)--cycle;
\fill[MyLightRed,opacity=0.25] (FRT)--(FRB)--(BRB)--(BRT)--cycle;

\draw[edge] (FLT)--(FRT)--(FRB)--(FLB)--cycle;
\draw[edge] (BLT)--(BRT)--(BRB)--(BLB)--cycle;
\draw[edge] (FLT)--(BLT);
\draw[edge] (FRT)--(BRT);
\draw[edge] (FRB)--(BRB);
\draw[edge] (FLB)--(BLB);

\draw[sinkspan] (FLT)--(FRT)--(BRT)--(BLT)--cycle;
\draw[sinkspan] (FRT)--(FRB)--(BRB)--(BRT)--cycle;

\draw[sinkorient] ($(FLT)!0.48!(FRT)$) -- ($(FLT)!0.52!(FRT)$);
\draw[sinkorient] ($(FRT)!0.48!(FRB)$) -- ($(FRT)!0.52!(FRB)$);
\draw[sinkorient] ($(FRB)!0.48!(BRB)$) -- ($(FRB)!0.52!(BRB)$);
\draw[sinkorient] ($(BRB)!0.48!(BRT)$) -- ($(BRB)!0.52!(BRT)$);
\draw[sinkorient] ($(BRT)!0.48!(BLT)$) -- ($(BRT)!0.52!(BLT)$);
\draw[sinkorient] ($(BLT)!0.48!(FLT)$) -- ($(BLT)!0.52!(FLT)$);
\draw[sinkorient] ($(BRT)!0.48!(FRT)$) -- ($(BRT)!0.52!(FRT)$);

\draw[midorient] ($(FLB)!0.48!(FLT)$) -- ($(FLB)!0.52!(FLT)$);
\draw[midorient] ($(FLB)!0.48!(FRB)$) -- ($(FLB)!0.52!(FRB)$);
\draw[midorient] ($(FLB)!0.48!(BLB)$) -- ($(FLB)!0.52!(BLB)$);
\draw[midorient] ($(BLB)!0.48!(BLT)$) -- ($(BLB)!0.52!(BLT)$);
\draw[midorient] ($(BLB)!0.48!(BRB)$) -- ($(BLB)!0.52!(BRB)$);

\foreach \V in {FRB,FRT,FLT,BRB,BRT,BLT}{
\fill[MyDarkRed] (\V) circle (0.9pt);
}
\foreach \V in {FLB,BLB}{
\fill[black] (\V) circle (0.9pt);
}

\end{tikzpicture}
\caption{Preference graphs of three $2 \times 2 \times 2$ games.
In each one, the red vertex set is a minimal \ac{club} set\,/\,sink equilibrium, and the red region represent its span (an edge cycle on the left, a square face in the middle, and a union of two square faces on the right).}
\label{fig:span}
\end{figure}

\subsection{Regularized learning and dynamical stability}

We continue with our basic learning model and related concepts.

\para{Regret and \ac{FTRL}}

A central benchmark for adaptive learning is the notion of \emph{regret}.
In our setting, given a trajectory of play $\strat(\time)\in\strats$, $\time\geq0$, we define the regret of player $\play\in\players$ up to time $\horizon\geq0$ as
\[
\reg_{\play}(\horizon)
	\defeq \max_{\pure_{\play}\in\pures_{\play}}
		\int_{\tstart}^{\horizon}
			\bracks{\pay_{\play}(\pure_{\play};\strat_{-\play}(\time)) - \pay_{\play}(\strat(\time))}
		\dd\time
\]
\ie as the difference between the player's cumulative payoff up to time $\horizon$ and that of their best strategy in hindsight.
We then say that the player incurs \emph{no regret} if $\reg_{\play}(\horizon) = o(\horizon)$.

A broad and widely used family of no-regret dynamics is \acdef{FTRL}.
The guiding principle of \ac{FTRL} is that each player $i$ maintains a vector of \emph{scores}, i.e., aggregate payoffs, and selects a mixed action by trading off these scores against a regularization term that smoothens the dynamics and incentivizes exploration.
Concretely, fix for each player $i$ a \emph{regularizer} $h_i:\X_i\to\R\cup\{+\infty\}$ and define the associated \emph{choice map} $Q_i : \YY_i \to \X_i$
\[
Q_i(y_i) \defeq \arg\max_{x_i\in\X_i}\{\langle y_i,x_i\rangle-h_i(x_i)\}.
\]
With this notation, the \ac{FTRL} dynamics can be written for each player $i$ as:
\begin{equation}\label{eq:FTRL}
\underbrace{\dot y_i(t) = v_i(x(t))}_{\text{aggregate payoffs}}
\qquad
\underbrace{x_i(t) = Q_i(y_i(t))}_{\text{strategy update}}.
\tag{FTRL}
\end{equation}
For the regularizer, we will assume for our purposes that it is \emph{decomposable}, meaning that
\[
h_i(x_i)=\sum_{\alpha_i\in\A_i}\hker_{\play}(x_{i\alpha}),
\quad\hker_{\play}: [0,1] \to \R,
\]
where we call $\hker_{\play}$ the \emph{kernel} of $h_i$.
To streamline our presentation, we will assume throughout that $\hker_{\play}$ is
\begin{enumerate*}
[\upshape(\itshape a\upshape)]
\item
continuous on $[0,1]$;
\item
$C^{2}$-smooth on $(0,1]$;
and
\item
\emph{strongly convex}, \ie $\hker''_{\play}(z) \geq \hstr_{\play}$ for some $\hstr_{\play}>0$ and for all $z\in(0,1]$.
\end{enumerate*}
For consision, we will also write $h \defeq \sum_{i \in \N} h_i$ for the aggregate regularizer  and $Q \defeq (Q_i)_{i \in \N} : \YY \to \X$ for the induced choice map.
This regularization viewpoint yields clean regret guarantees.
Indeed \cite{KM17} proved that in continuous time, along any solution of \eqref{eq:FTRL}, player $i$'s regret is uniformly bounded by the range of the regularizer:
\[
\forall T \ge 0, \quad\reg_i(T)\le \max h_{i}- \min h_{i},
\]
In this sense, the regularizer influences performance, but it also plays a second, equally important role: it determines the geometry of the induced learning dynamics via $\choice$.

\para{Strategy dynamics} We call $h_i$ \emph{steep} if $\hker_{\play}'(p)\to -\infty$ as $p\downarrow 0$.
In that case the regularization cost becomes prohibitive near the boundary, so the regularized maximizer is always interior, that is, its image satisfies $\ImQ_i=\X_i^\circ$.
By contrast, if $h_i$ is not steep, then boundary points may in general be selected and typically reached in finite time.
For conciseness, we will say $h$ is steep if $h_i$ is steep for every player.
Two canonical regularizers illustrate this dichotomy:

\begin{example}
For the \emph{entropic} regularizer, with steep kernel $\hker_{\play}(p)=p\log p$,
the choice map is the softmax:
\[
\Lambda_i(y_i)
 \defeq \left(\frac{\exp(y_{i\alpha})}{\sum_{\beta_i\in\A_i}\exp(y_{i\beta})}\right)_{\alpha_i\in\A_i}.
\]
which induces, in strategy space, the \emph{replicator dynamics}
\begin{equation}\label{eq:RD}
\dot x_{i\alpha}=x_{i\alpha}\big(v_{i\alpha}(x)- u_i(x)\big).
\tag{RD}
\end{equation}
\end{example}

\begin{example}
    At the other extreme, for the \emph{Euclidean} regularizer, whose  kernel $\hker_{\play}(p)=\tfrac12p^2$ is not steep, the choice map is the Euclidean projection into $\X_i$: \[\Pi_i(y_i) \defeq \arg\min_{x_i\in\X_i}\|x_i-y_i\|_2^2.\] Hence, the induced dynamics are the piecewise-smooth \emph{projection dynamics} \cite{MerSan18}, and are most transparently expressed on each fixed-support face: if $x_i(t)\in\Delta(\B_i)^\circ$ for $t$ in some interval, $x(t)$ follows the \emph{projection dynamics}
\begin{equation}\label{eq:PD}
\dot x_{i\alpha} = v_{i\alpha}(x)-\dfrac{1}{|\B_i|}\sum_{\beta_i\in \B_i} v_{i\beta}(x), \quad \alpha_i\in \B_i.
\tag{PD}
\end{equation}
\end{example}

More generally, it is useful to eliminate the score variables and write the induced evolution directly in strategy space.
On any time interval where player $i$'s support is $\B_{i}$, let
\[
s_i
	\defeq \frac{1}{\hker_{\play}''}
	\;
	\text{and}
	\;
\pi_{i\alpha}(x_i)
	\defeq \begin{dcases*}
		\frac{s_i(x_{i\alpha})}{\sum_{\gamma_i\in \B_i}s_i(x_{i\gamma})}
			&if $\alpha_i \in \B_i$,
			\\
		0
			&otherwise.
		\end{dcases*}
\]
Then $\strat(\time)$ follows the \emph{strategy dynamics}
\begin{equation}\label{eq:SD}
\dot x_{i\alpha}
=s_i(x_{i\alpha})\big(v_{i\alpha}(x)-\langle \pi_i(x_i),v_i(x)\rangle\big).
\tag{SD}
\end{equation}
Finally, fix the initial support $\B=\supp(x(0))$ and let $\F=\cont(\B)$ be the corresponding face.
If the regularizer is steep, the support never changes and trajectories remain in $\F^\circ$ for all time.
Thus, under some additional mild regularity assumptions on $s_i$, the facewise vector fields coincide on intersections of adjacent faces and therefore glue into a globally Lipschitz vector field on $\X$.
Hence the induced strategy dynamics are globally well-posed
for all $t\in\R$ and are \emph{face-invariant}, meaning the support never changes along the dynamics.
We call the resulting flow the \emph{strategy flow} and denote it by $(\Theta_t)_{t\in\R}$ (see \cref{app:dynamics} for more details).

\para{Dynamic stability and attraction}

Informally, stability means that orbits starting sufficiently close to a set remain close to it for all time, while attraction means that trajectories starting nearby converge to said set.
Formally:

\begin{definition}
Let $x(t) = Q(y(t))$ be an orbit of \eqref{eq:FTRL}, and let $\set\subseteq\X$ be a nonempty closed subset of $\strats$.
Then:
\begin{subequations}
\begin{enumerate}
\item
$\set$ is \emph{\textpar{Lyapunov} stable} if for every neighborhood $\nhd$ of $\set$ there exists a neighborhood $\nhdalt$ of $\set$ such that
\[
x(0)\in \nhdalt \cap \ImQ
	\implies
x(t)\in \nhd\;\text{for all }t\ge 0
	\eqstop
\]
\item
$\set$ is \emph{attracting} if it admits a neighborhood $\nhd$ such that
\[
x(0)\in \nhd \cap \ImQ
	\implies
\lim\nolimits_{t\to\infty}\dist(x(t),\set)= 0
	\eqstop\!
\]
\item
$\set$ is \emph{asymptotically stable} if it is stable and attracting.
\end{enumerate}
\end{subequations}
\end{definition}

As for the induced strategy flow $(\Theta_t)_{t\in \R}$ of \eqref{eq:SD}, we shall work with the notion of \emph{attractor}, which is an invariant\footnote{For the strategy flow, invariant means that $S\subseteq\X$ satisfies $\Theta_t(S)=S$ for all $t\in\R$.} asymptotically stable set (see \cref{app:attractors} for precise definitions and related discussions concerning attractors).
Importantly, for this strategy flow notion of attractor, we allow nearby initializations on \emph{all of $\X$}, including all adjacent faces, whereas for asymptotic stability under \eqref{eq:FTRL}, the only allowed initializations are on \emph{the image of the choice map} (in particular, $\X^\circ$ for steep regularizers).

\section{Implications of dynamic stability}
\label{sec:dynstable}

This section develops a systematic link between two levels of description:
(1) combinatorial properties of subsets of the preference graph (ordinal stability notions), and
(2) dynamical properties of regions of strategy space (dynamic stability notions).
The guiding theme is that ordinal incentives, encoded by the preference graph, impose constraints on which regions of~$\X$ can be long-run stable outcomes under regularized learning.
All proofs are deferred to \cref{app:dynstable}.

We start with the first connection, which is that attraction for the strategy flow forces no outgoing better replies:

\begin{proposition}
\label{prop:attract-club}
Suppose that $S \subseteq \X$ is attracting under the strategy dynamics \eqref{eq:SD}.
Then its skeleton is \ac{club}.
\end{proposition}

The intuition behind this is that, under steep regularization, the dynamics \eqref{eq:SD} are face-invariant, so starting on the $1$-dimensional face spanned by two comparable vertices, the orbit stays on that segment, and the flow always follows the direction of the better reply.
Outside the strategy flow regime, face-invariance is no longer available and we can no longer exploit initializations on proper subfaces of the strategy space.
Nonetheless we are still able to prove, via new techniques utilizing the dual score representation of the dynamics, a strict version of the same result for stable sets, thereby showing that dynamical stability implies preferential stability.

\begin{theorem}
\label{thm:stable-sclub}
Suppose that $S \subseteq \X$ is stable under \eqref{eq:FTRL}.
Then its skeleton is \ac{sclub}.
\end{theorem}

Indeed, even without face invariance, stability near a pure profile $\alpha$ rules out the existence of a \emph{strict} better-reply $\alpha\to\alpha'$ because one could initialize scores so that the opponents remain close to $\alpha$ long enough for player $i$'s score difference $y_{i\alpha'}-y_{i\alpha}$ to grow linearly thanks to the strict payoff gap.
Once this gap exceeds a certain carefully chosen regularizer-dependent threshold, the choice map shifts player $i$'s mixed action noticeably toward $\alpha'_i$, pushing $x(t)$ outside any prescribed small neighborhood of $\alpha$, contradicting stability.
Thus, the whole difficulty lies in choosing the right initializations so that player $i$, the deviating player, has enough time to move outside the prescribed neighborhood while the other players remain (almost) frozen during that time, exploiting in particular the fact that the dynamics have \emph{bounded speed} in the space of scores $\scores$.

The next result shows that strong connectivity in the preference graph imposes a constraint on the shape of attractors.

\begin{theorem}
\label{thm:sc-span}
If $A\subseteq \X$ is an attractor of the strategy dynamics \eqref{eq:SD} and contains a strongly connected set $\HH \subseteq \A$, then $\cont(\HH) \subseteq A$.
Hence, if the preference graph is strongly connected, the flow admits no proper attractor.%
\footnote{Proper here means in the sense of being a strict subset, that is, the attractor is not all of $\X$.}
\end{theorem}

This theorem is a consequence of a general principle: connectedness in the preference graph percolates to a corresponding notion of connectedness for the dynamics on the associated mixed region, which generalizes the second theorem of \cite{BS23}.
Formally, this is captured by \emph{chain transitivity}, which intuitively means that one can move from any point to any other by following the flow for arbitrarily long stretches, allowing only arbitrarily small perturbations in between.
We refer to \cref{app:dynamics} for precise definitions.
In order to establish this principle, we proceed by induction, proving that if $A$ contains all subfaces of a face, it must then contain the whole face as well, exploiting the fact that under the strategy flow, no asymptotically stable set can be fully contained in the interior \cite{FVGL+20}.

\section{Implications of preferential stability}
\label{sec:prefstable}

The results of \cref{sec:dynstable} show that asymptotic stability forces a range of closedness and stability properties on the preference graph of the game.
We now ask when the converse holds for the natural geometric candidates arising from the graph, namely \emph{spans of pure profiles}, that is, unions of faces of $\X$.
In words:
\begin{center}
\itshape
When does preferential stability imply dynamic stability\\
of the region that it spans, and when does it not?
\end{center}

\subsection{When preferences are enough}
\label{sec:prefs-enough}

We begin with the simplest class of spans:
those corresponding to a single face, \ie spans of subgames, which can be written as $\B=\prod_{i\in\N}\B_i$, with $\B_i\subseteq\A_i$.
In this case, preferential stability turns out to imply asymptotic stability:

\begin{theorem}\label{thm:subgame-stab}
    If $\B \subseteq \A$ is a \ac{club} subgame, then $\cont(\B)$ is asymptotically stable under \eqref{eq:FTRL}.
\end{theorem}

The proof (in \cref{app:prefstable}) relies on a recurring device in this section: constructing a suitable \emph{energy function} which converts dissipativity in score space to asymptotic stability in strategy space.
In the present case, the relevant energy is the \emph{Fenchel gap}
\[
F_\B(y):=h^*(y)-h^*_{\B}(y),
\]
where $h^*$ is the convex conjugate of $h$ and $h_\B^*(y)$
is its restriction to the face spanned by $\B$.
The Fenchel gap may be viewed as a ``dual'' measure of distance to $\cont(\B)$, since it is always nonnegative and vanishes exactly when $Q(y) \in \cont(\B)$.
If $\B$ is a \ac{club} set, then in a neighborhood of $\cont(\B)$ every action outside $\B$ is uniformly worse than some action inside $\B$, yielding a differential inequality showing that $F_\B$ decreases at a rate proportional to the total probability mass of playing outside $\B$, or, equivalently, to the $\ell_1$ distance to its span.
Therefore, the dynamics exhibit a sharp inward drift toward $\cont(\B)$, which implies attraction and stability.
Importantly, this energy function works even for regularizers that may be not steep, which makes it a substantial improvement over recent prior approaches \cite{BM23,CLM25}.

Now, combining the two theorems above with the results of \cref{sec:dynstable}, we obtain the following strong equivalences:

\begin{corollary}
\label{cor:subgame-equiv}
Let $\B \subseteq \A$ be a subgame and assume the game has no ties.
Then $\B$ is \ac{club} if and only if $\cont(\B)$ is asymptotically stable under \eqref{eq:FTRL}.
\end{corollary}

\begin{corollary}
\label{cor:subgame-attract}
Let $\B \subseteq \A$ be a subgame.
Then $\B$ is \ac{club} if and only if $\cont(\B)$ is an attractor of the strategy dynamics \eqref{eq:SD}.
\end{corollary}

Note that these results are setwise generalizations of the pointwise version, which has been sometimes called \enquote{the folk theorem of evolutionary game theory} in the literature \cite{MS16,San10,Mer26}, and says that a point is asymptotically stable if and only if it is a strict \acl{NE} \citep{FVGL+20}.
Indeed, any strict equilibrium is a vertex of the preference graph with no outgoing edge, and therefore a singleton \ac{club} subgame.

We next illustrate how these equivalences specialize in games where better-reply paths always lead to equilibrium.
This is the reachability notion of \emph{weak acyclicity}, introduced by \citet{You93}.

\begin{definition}\label{def:weak-acyclic}
A game is \emph{weakly acyclic} if for every $\alpha\in\A$ there exists a finite path of better replies
\begin{equation*}
\alpha=\alpha^0\to\alpha^1\to\cdots\to\alpha^k
\end{equation*}
ending at a pure \acl{NE} $\alpha^k$.
\end{definition}

\begin{remark}\label{rmrk:potential}
Note that weakly acyclic games are a generalization of the widely studied class of \emph{potential games}, and more generally \emph{ordinal potential games}, introduced by \citet{MS96}, which are equivalent to \emph{acyclic} games, games whose preference graph admits no directed cycles.
\endenv
\end{remark}

Applying our results to the class of weakly acyclic games, we obtain a sharp characterization of \emph{minimal attractors}.
\footnote{Here, minimal is in the sense of set inclusion, that is, attractors which contain no proper attractor.}

\begin{corollary}\label{thm:wa-attract}
If the game is weakly acyclic and has no ties, then $M \subseteq \X$ is a minimal attractor of the strategy dynamics \eqref{eq:SD} if and only if $M$ is a strict \ac{NE}.
\end{corollary}

An important application of this result follows from \citet{JSST23}, who show that, in the regime of many players relative to the maximal number of actions per player, weak acyclicity becomes overwhelmingly likely: among games admitting at least one pure \acl{NE}, the probability of being weakly acyclic converges to 1 at an exponential rate in the number of players.
Consequently, for such ``typical'' large games with a pure \acl{NE}, \cref{thm:wa-attract} guarantees, under the strategy flow, that all minimal attractors coincide with strict \aclp{NE}.
This gives theoretical insight as to why \ac{FTRL} seems in practice to work well with games with many players but (relatively) few actions.

\subsection{When preferences are \emph{not} enough}
\label{sec:prefs-not-enough}

The preceding results rely crucially on the product structure of subgames.
When $\HH\subseteq\A$ is \emph{not} a subgame, preferential stability of $\HH$ may fail to control the dynamic stability of its span.
The next proposition makes this separation explicit.

\begin{proposition}
\label{prop:span-counter}
There is a $2 \times 2 \times 2$ game with a unique proper \ac{club} set $\HH$ for which $\cont(\HH)$ is not stable under \eqref{eq:FTRL}.
\end{proposition}

The game in question is represented in \cref{fig:misalign}, where $\HH$ is the set of red vertices, and its span $\cS = \cont(\HH)$ is the union of the top, right and back faces.
One can observe that, near the center of the top face, the third player has an incentive to deviate toward the bottom face, which is \emph{not} included in $\cS$;
this occurs ``before'' the other two players can move toward another face of $\cS$, thus rendering it unstable.
In fact, in this example, the minimal attractor containing $\HH$ is not even a span of pure profiles (\cf \cref{rmrk:not-span}).
The center of the top face is also a (mixed) local source in the sense of \citet{BP25}, so this instability result is consistent with the corresponding higher-dimensional example of \cite{BP25}.
In our example though, the escape happens from full-support initializations and under every instance of \eqref{eq:FTRL}, whereas the analysis of \cite{BP25} only concerns \eqref{eq:RD} (see \cref{rmrk:biggar} for more details).

\begin{remark}
\label{rem:BP25}
\citet{BP25} conjectured that a sink equilibrium which is unstable under the replicator dynamics must contain a local source.
The counterexample of \cref{prop:span-counter} exhibits a dynamically unstable sink equilibrium without a \emph{pure} local source, showing that the conjecture of \citet{BP25} cannot be strengthened to positing the existence of pure local sources.
\endenv
\end{remark}

\section{Recovering dynamic stability}
\label{sec:rad}

\Cref{prop:span-counter} shows that preferential information alone may be too coarse to determine dynamic stability: recovering it requires \emph{a fortiori} additional information, such as payoff values.
For this reason, we introduce a natural and simple payoff-dependent condition which restores stability.

\begin{definition}
Given $\alpha, \beta \in\A$, we define the \emph{payoff flux from $\beta$ to $\alpha$} as
\begin{equation}\label{eq:flux}
\Phi(\beta,\alpha)\coloneqq \sum_{i\in \N}\left(u_i(\alpha_i,\beta_{-i})-u_i(\beta)\right).
\end{equation}
Given $\HH\subseteq\A$, we say that $\HH$ is \acdef{rad} if
\[
\Phi(\beta,\alpha) \leq 0
    \quad\text{for all $\beta\in\HH$ and all $\alpha\notin\HH$},
\]
and we say that it is \acdef{srad} if the above inequalities are strict.
\end{definition}
\acused{radness}
\acused{sradness}

We similarly extend the notion of the payoff flux to mixed strategies by setting, for $x, x' \in \X$,
\[
\Phi(x,x')\coloneqq \sum_{i\in \N}\left(u_i(x'_i,x_{-i})-u_i(x)\right).
\]
The summand $u_i(\alpha_i,\beta_{-i})-u_i(\beta)$ is the instantaneous gain to player $i$ from deviating unilaterally to $\alpha_i$ at $\beta$, so $\Phi(\beta,\alpha)$ aggregates these gains along the target profile $\alpha$.
The notion of \ac{sradness} therefore means that every outside profile $\alpha\notin\HH$ is unappealing in aggregate against every $\beta\in\HH$.
This property depends on payoff magnitudes, which correspond naturally to the weights of the preference graph: if $\alpha\to\beta$ are $i$-comparable, the weight of that edge is $u_i(\beta)-u_i(\alpha)$.
Moreover, from a graph-theoretic perspective, $\Phi(\beta,\alpha)$ is obtained by summing the signed weights of edges moving one step (\ie a unilateral deviation) from $\beta$ toward $\alpha$.
Importantly, this notion of \ac{radness} provides a cardinal strengthening of ordinal closure.

\begin{proposition}\label{prop:resilient-club}
Let $\HH$ be a subset of $\A$.
If $\HH$ is \ac{rad} \textpar{resp.~\ac{srad}}, then $\HH$ is \ac{sclub} \textpar{resp.~\ac{club}}.
\end{proposition}

The diagram below summarizes the chain of implications chain between \ac{club}\,/\,\ac{sclub} and \ac{rad}\,/\,\ac{srad} sets:
\begin{equation*}
\begin{array}{ccc}
\text{\ac{rad}} & \Longrightarrow & \text{\ac{sclub}} \\[\smallskipamount]
\Big\Uparrow      &                  & \Big\Uparrow   \\[\smallskipamount]
\text{\ac{srad}} & \Longrightarrow & \text{\ac{club}}
\end{array}
\end{equation*}
Specifically, in games without ties, \ac{club} and \ac{sclub} sets coincide, whereas \ac{rad} and \ac{srad} sets need not.
Moreover, by \cref{rmrk:singletons}, \ac{rad}\,/\,\ac{sclub} singletons and pure \aclp{NE} \emph{all coincide}, and likewise for \ac{srad}\,/\,\ac{club} singletons and strict \aclp{NE}.
Thus, when there are no ties, all of them are equivalent for singletons.
Hence, \ac{radness} provides a setwise \emph{cardinal} generalization of the notion of pure \aclp{NE}, just as \acl{clubness} provides its setwise \emph{ordinal} counterpart.

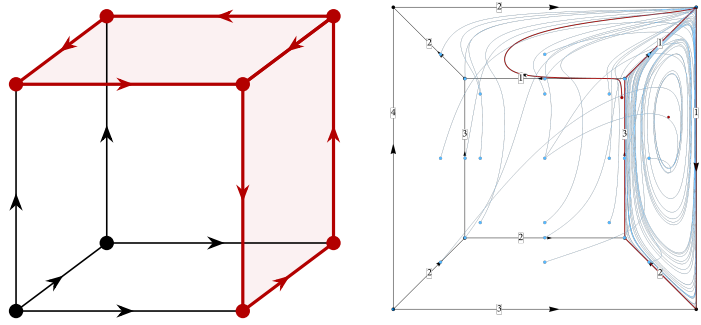
\begin{figure}[t]
\renewcommand{\thisfigscale}{3}
\renewcommand{\thisfigheight}{28ex}
\centering
%
%
\begin{tikzpicture}[
	scale=\thisfigscale,
	line join=round,
	line cap=round,
	edge/.style={semithick,draw=black},
	sinkspan/.style={draw=MyDarkRed,very thick},
	midorient/.style={-{Stealth},thick,draw=black},
	sinkorient/.style={-{Stealth},thick,draw=MyDarkRed}
]

\coordinate (BLB) at (0,0);
\coordinate (BRB) at (1,0);
\coordinate (BRT) at (1,1);
\coordinate (BLT) at (0,1);

\coordinate (FLB) at (0.40,0.30);
\coordinate (FRB) at (1.40,0.30);
\coordinate (FRT) at (1.40,1.30);
\coordinate (FLT) at (0.40,1.30);

\fill[MyLightRed,opacity=0.25] (BLT)--(BRT)--(FRT)--(FLT)--cycle;
\fill[MyLightRed,opacity=0.25] (BRB)--(FRB)--(FRT)--(BRT)--cycle;

\draw[edge] (FLB)--(FRB)--(FRT)--(FLT)--cycle;
\draw[edge] (BLB)--(BRB)--(BRT)--(BLT)--cycle;
\draw[edge] (FLB)--(BLB);
\draw[edge] (FRB)--(BRB);
\draw[edge] (FRT)--(BRT);
\draw[edge] (FLT)--(BLT);

\draw[sinkspan] (BLT)--(BRT)--(FRT)--(FLT)--cycle;
\draw[sinkspan] (BRB)--(FRB)--(FRT)--(BRT)--cycle;

\draw[sinkorient] ($(BLT)!0.48!(BRT)$) -- ($(BLT)!0.52!(BRT)$);
\draw[sinkorient] ($(BRT)!0.48!(BRB)$) -- ($(BRT)!0.52!(BRB)$);
\draw[sinkorient] ($(BRB)!0.48!(FRB)$) -- ($(BRB)!0.52!(FRB)$);
\draw[sinkorient] ($(FRB)!0.48!(FRT)$) -- ($(FRB)!0.52!(FRT)$);
\draw[sinkorient] ($(FRT)!0.48!(FLT)$) -- ($(FRT)!0.52!(FLT)$);
\draw[sinkorient] ($(FLT)!0.48!(BLT)$) -- ($(FLT)!0.52!(BLT)$);
\draw[sinkorient] ($(FRT)!0.48!(BRT)$) -- ($(FRT)!0.52!(BRT)$);

\draw[midorient] ($(BLB)!0.48!(BLT)$) -- ($(BLB)!0.52!(BLT)$);
\draw[midorient] ($(BLB)!0.48!(BRB)$) -- ($(BLB)!0.52!(BRB)$);
\draw[midorient] ($(BLB)!0.48!(FLB)$) -- ($(BLB)!0.52!(FLB)$);
\draw[midorient] ($(FLB)!0.48!(FLT)$) -- ($(FLB)!0.52!(FLT)$);
\draw[midorient] ($(FLB)!0.48!(FRB)$) -- ($(FLB)!0.52!(FRB)$);

\foreach \V in {FRB,FRT,FLT,BRB,BRT,BLT}{
\fill[MyDarkRed] (\V) circle (0.9pt);
}
\foreach \V in {FLB,BLB}{
\fill[black] (\V) circle (0.9pt);
}

\end{tikzpicture}
\quad
\includegraphics[height=\thisfigheight]{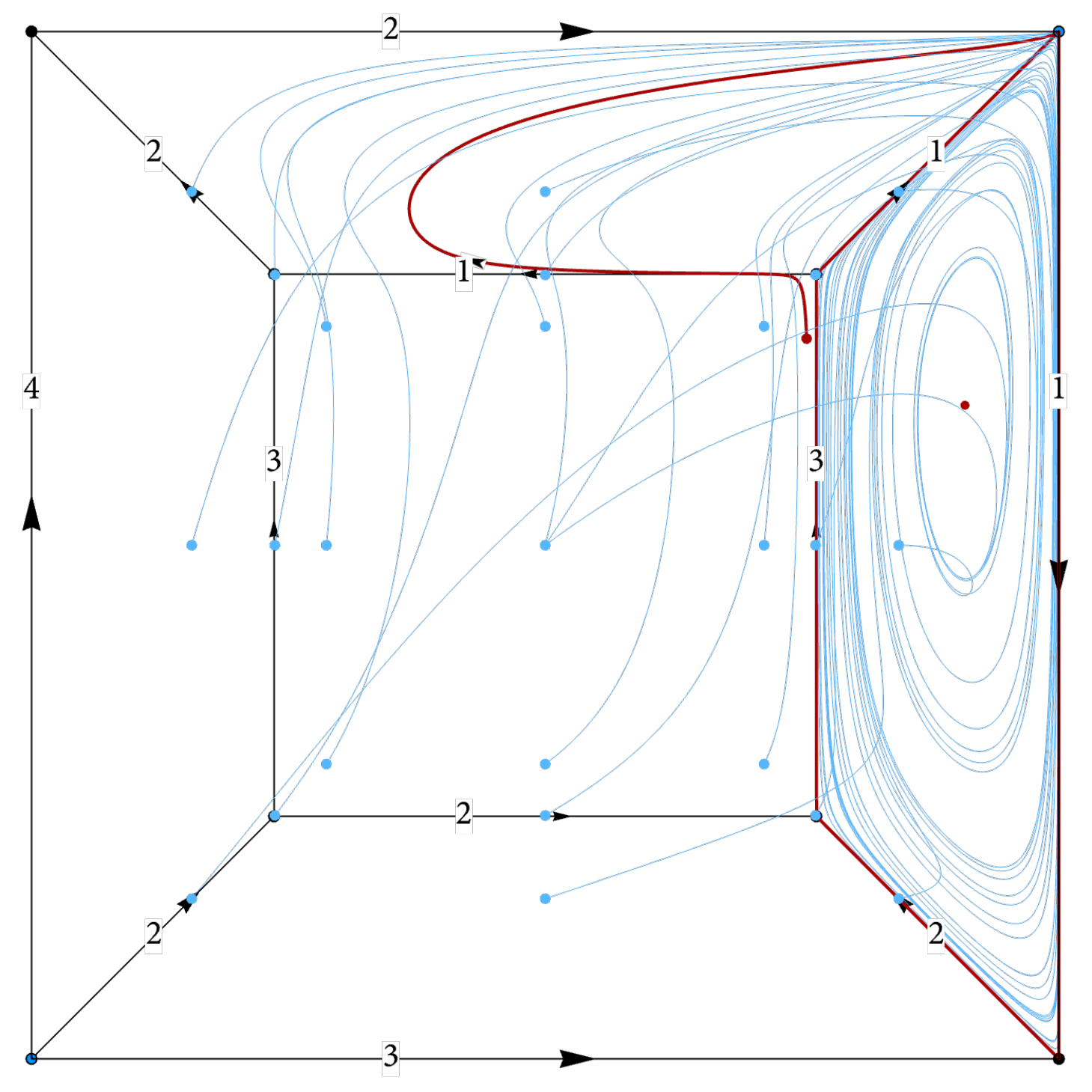}
\caption{A game with an \ac{srad} set (for the game's payoffs, \cf \cref{app:preferences}).
On the left, we plot the game's preference graph, highlighting the \ac{srad} set (and its span) in red;
on the right, we plot the evolution of the replicator dynamics in said game.
We note that the (span of) the \ac{srad} set is asymptotically stable, as per \cref{thm:strict-resilience}.
At the same time, we also note that all replicator orbits eventually converge to the rightmost face, a face which is neither \ac{rad} (even though it is contained in the original \ac{srad} set), nor asymptotically stable (an orbit may start arbitrarily close to said face and venture arbitrarily far before eventually returning to it, \cf the red orbit in the plot to the right).}
\label{fig:resilience}
\end{figure}

As for complexity, it is also clear that \ac{radness} can be checked  in $O\!\left(|\N|\,|\HH|\,|\A\setminus\HH|\right)$ steps given $\HH$.
Moreover, one can \emph{find} a \ac{rad} set in $O(|\N|\,|\A|^2)$ time as follows: define the directed graph $G_\Phi$ on vertex set $\A$ with an arc $\beta\to\alpha$ if and only if $\Phi(\beta,\alpha)>0$ (resp.\ $\Phi(\beta,\alpha)\ge 0$ for strict \ac{rad}).
Then $\HH$ is \ac{rad} if and only if it has no outgoing arcs in $G_\Phi$, i.e., if and only if there is no edge $\beta\to\alpha$ with $\beta\in\HH$ and $\alpha\notin\HH$.
Constructing $G_\Phi$ requires computing all payoff fluxes between pure profiles, which takes $O(|\N|\,|\A|^2)$ time, and finding a subset with no outgoing edge is $O(|\A|^2)$ (via standard graph decomposition algorithms, \eg \cite{Tar72}), hence the total complexity is $O(|\N|\,|\A|^2)$.
Thus, given access to the payoff table, \ac{radness} can be tested and
found with a polynomial number of payoff queries in
the size of the explicit normal-form representation.

This stands in sharp contrast with the convex relaxation of the notion of a pure \aclp{NE}, namely \emph{mixed} \aclp{NE}.
Despite being the most widely studied solution concept in finite games, non-pure, mixed \aclp{NE} are computationally hard:
in the same normal-form payoff-table model, computing one is well known to be \texttt{PPAD}-hard \cite{DGP09}.
They are moreover dynamically unappealing:
they are \emph{never} asymptotically stable under regularized learning \cite{FVGL+20}.
In comparison, the setwise relaxation of \ac{sradness} recovers asymptotic stability of the entire span:

\begin{theorem}
\label{thm:strict-resilience}
Suppose that $\HH \subseteq \A$ is \ac{srad}.
Then $\cont(\HH)$ is an attractor of the strategy dynamics \eqref{eq:SD}.
\end{theorem}

The proof (which we present in \cref{app:rad}) is based on the construction of an energy function adapted to general spans, defined by
\begin{equation}
\label{eq:res-energy}
\bar F_\HH(y)
	\defeq \insum_{\alpha\notin\HH} e^{-F_h(\alpha,y)}
\end{equation}
with the convention $e^{-\infty}=0$, where $F_h(x,y):=h(x)+h^*(y)-\langle y,x \rangle$ is the \emph{Fenchel coupling}, a primal-dual divergence between strategies and scores.
In the steep case, along any \eqref{eq:FTRL} orbit $x(t)=Q(y(t))$, we have $F_h(\alpha,y(t)) \to +\infty$ if and only if $x_\alpha(t) \to 0$, hence the energy vanishes exactly on $\cont(\HH)$.
Additionally, the time-derivative of this energy along the same orbit is given by 
\[\frac{d}{dt} \bar F_\HH(y(t))=\sum_{\alpha\notin\HH} e^{-F_h(\alpha,\,y(t))}\,\Phi(x(t),\alpha),\] 
so the dynamics are dissipative near the \ac{srad} set. We conclude with a refinement of this result for the replicator dynamics:

\begin{theorem}\label{thm:resilience-RD}
Let $\HH$ be a \ac{rad} \ac{club} subset of $\A$.
Then $\cont(\HH)$ is an attractor of the replicator dynamics \eqref{eq:RD}.
\end{theorem}

This result implies in particular, when combined with \cref{prop:resilient-club}, that in any game with no ties, the span of every \ac{rad} set is asymptotically stable under the replicator dynamics.

The proof (in \cref{app:rad}) relies on a strikingly simple energy function: it is given by the total outside mass under the strategy profile $x$
\[
\bar W_\HH(x):=\sum_{\alpha\notin\HH}x_\alpha .
\]
It is nonnegative and vanishes exactly on $\cont(\HH)$.
Along the replicator
dynamics, the product rule gives
\[
\dot x_\alpha
=
x_\alpha\sum_{i\in\N}\bigl(v_{i\alpha}(x)-u_i(x)\bigr)
=
x_\alpha \Phi(x,\alpha),
\]
so, by the multilinearity of $\Phi(\cdot,\alpha)$, we get:
\begin{align*}
\frac{d}{dt} \bar W_\HH(x(t))
	&= \sum_{\alpha\notin\HH}\sum_{\beta\in\HH} x_\alpha x_\beta \Phi(\beta,\alpha)
	\notag\\
	&+ \sum_{\alpha\notin\HH}\sum_{\beta\notin\HH}
x_\alpha x_\beta \Phi(\beta,\alpha)
	\eqstop
\end{align*}
The first sum is nonpositive by \ac{radness}, while the second is
$O(\bar W_\HH(x)^2)$.
To get strict negativity near $\cont(\HH)$, fix
$x^\ast\in\cont(\HH)$ and let $\B=\supp(x^\ast)\subseteq\HH$.
For every
minimal profile $\alpha\notin\HH$, in the deviation order, that can appear when moving out of $\B$, \acs{clubness} gives a neighboring profile $\beta(\alpha)\in\HH$ with
\(
\Phi(\beta(\alpha),\alpha)<0,
\)
uniformly over such $\alpha$.
Hence the first sum contains a strictly negative contribution of the form
\(
\sum_{\alpha}x_\alpha x_{\beta(\alpha)}
\Phi(\beta(\alpha),\alpha),
\)
which dominates the second sum close to $\cont(\HH)$.
This ultimately gives
\(
\frac{d}{dt}\bar W_\HH(x(t))<0
\)
whenever $x(t)$ is sufficiently close to, but not in, $\cont(\HH)$, hence $\bar W_\HH$ is a local energy for the span.

Some intriguing examples of \ac{club} sets in the literature are
\begin{enumerate*}
[\upshape(\itshape a\upshape)]
\item
the $6$-cycle in Jordan's Matching Pennies \citep{Jor93} (\cref{fig:align});
and
\item
the hexagon in Shapley's game \citep{Sha64} (\cref{fig:shapley-main}).
\end{enumerate*}
These sets are both \ac{club} and \ac{rad}, but they are not \ac{srad} (\cf \cref{sub:examples}).
Importantly, existing proofs of asymptotic stability for these sets under the replicator dynamics \cite{GH95,KLPT11} rely on essentially ad hoc arguments tailored to those games.
It thus remains open to identify a general game-theoretic principle that could explain, beyond these special cases, the emergence of apparently coordinated outcomes under myopic, uncoupled adaptive behavior.
As far as we are aware, \ac{radness} provides the first general explanation for this form of stability under regularized learning.

\section{Concluding remarks}
\label{sec:discussion}

\begin{figure*}[t]
\centering
\includegraphics[width=.24\textwidth]{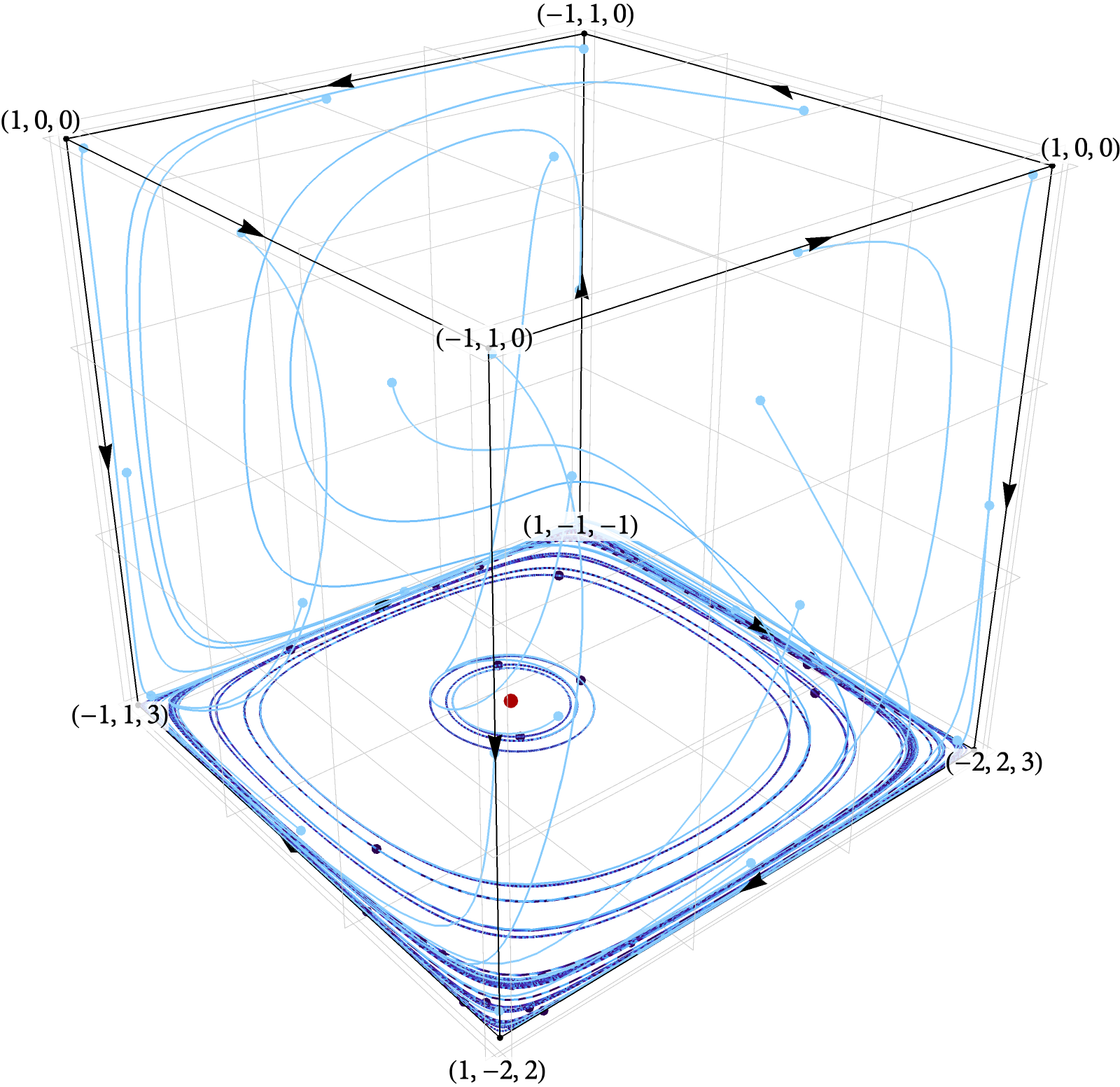}%
\hfill
\includegraphics[width=.24\textwidth]{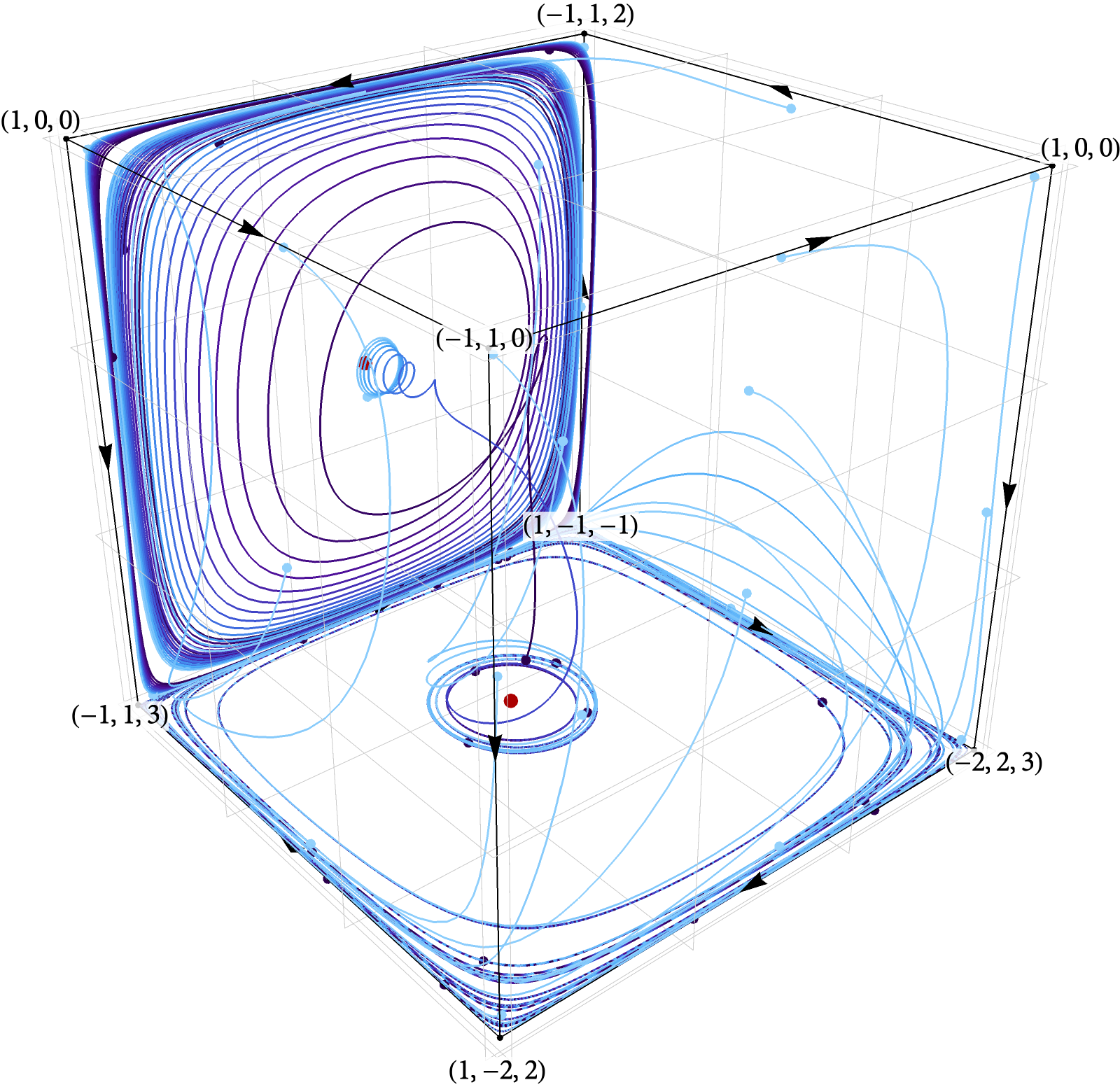}%
\hfill
\includegraphics[width=.24\textwidth]{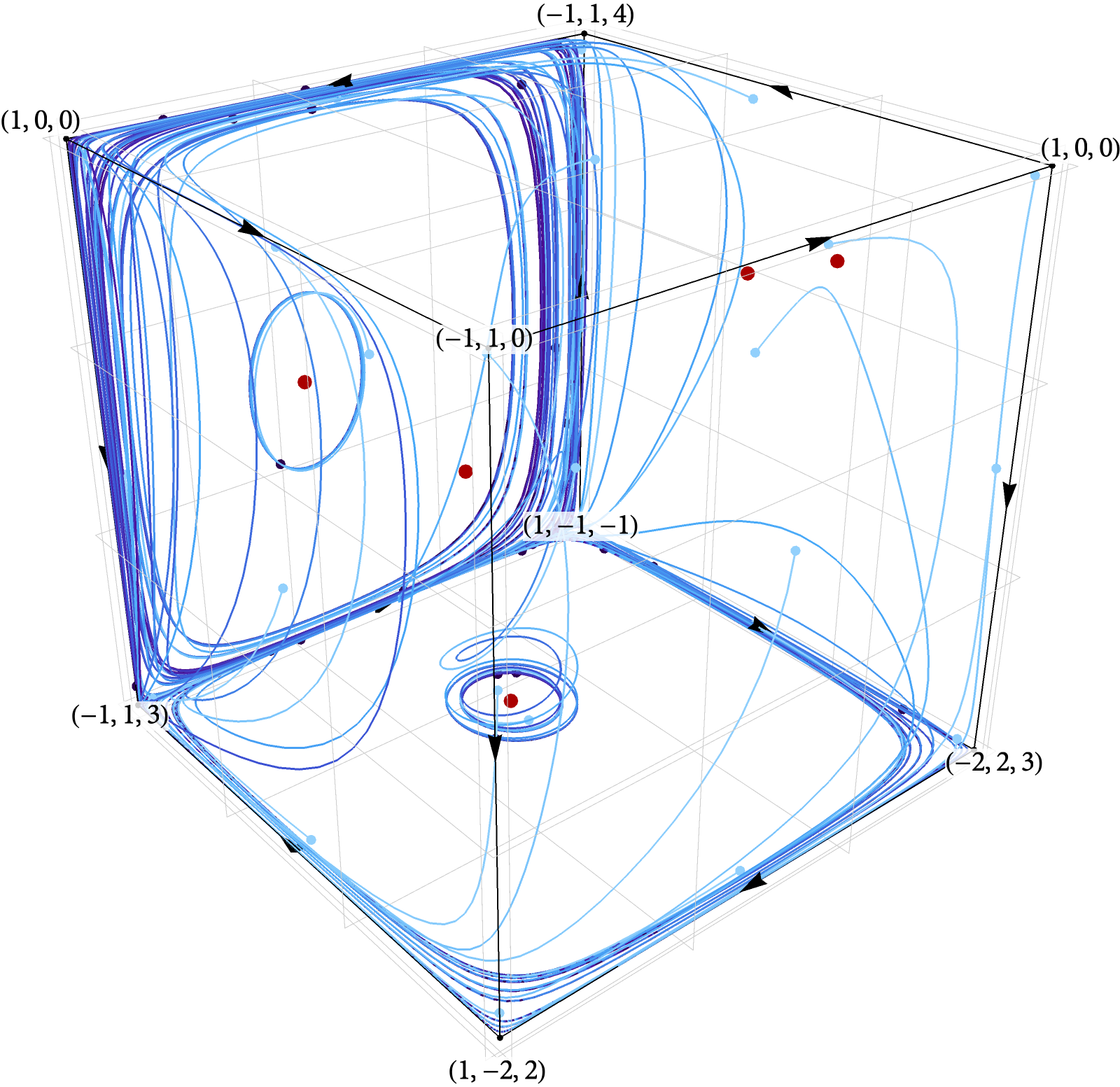}%
\hfill
\includegraphics[width=.24\textwidth]{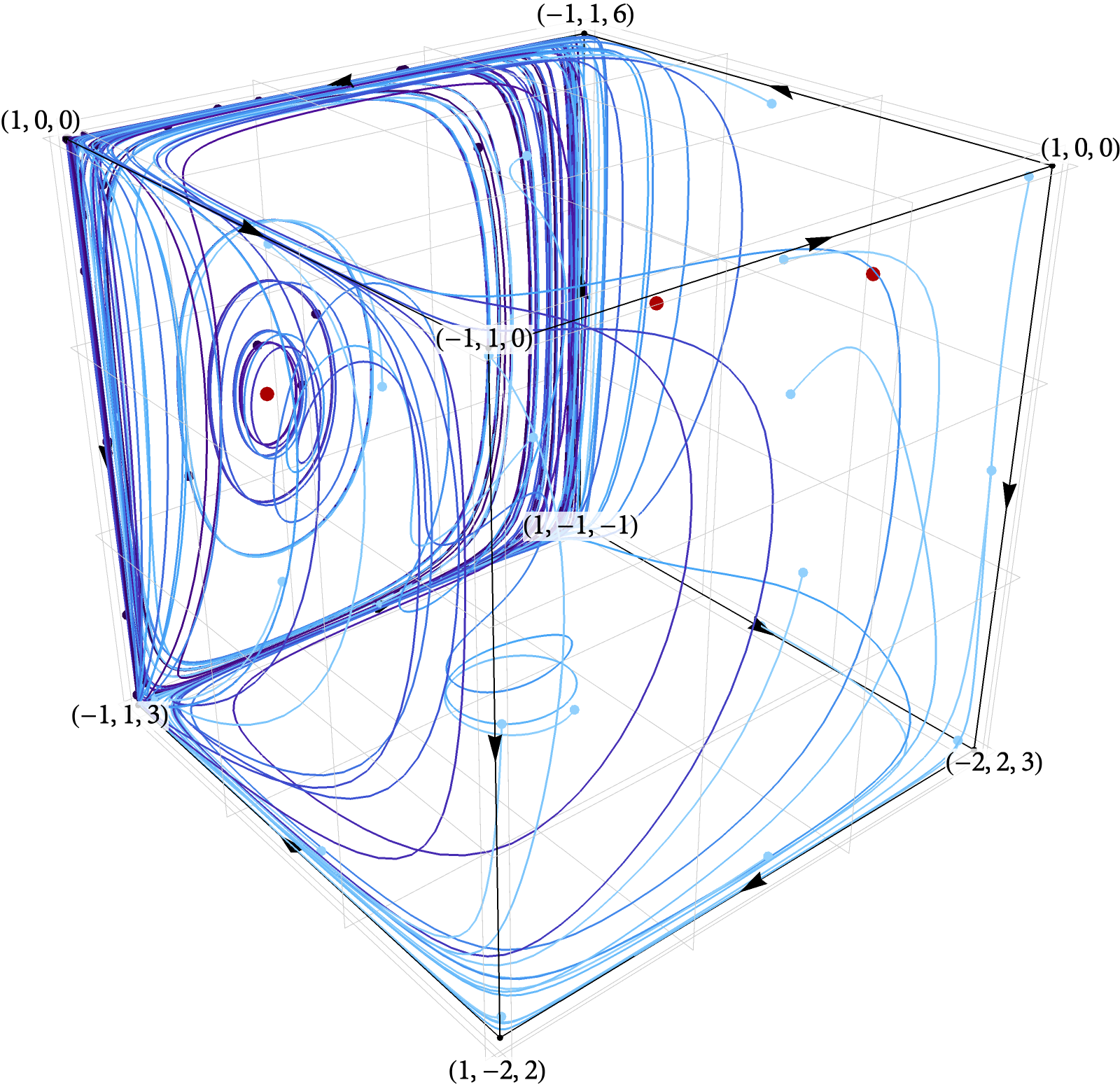}%
\caption{Four games with the same preference graph, differing at only a single payoff entry.
As we move from left to right, the dynamics become progressively more complicated and orbits converge to different faces, despite having the same preference graph.}
\label{fig:four-games}
\end{figure*}

Putting everything together, our analysis shows that the preference structure of a game provides a simple combinatorial blueprint which \emph{constrains} the stable outcomes of regularized learning, but is not enough to \emph{determine} it.
Moreover, although \ac{radness} provides a natural sufficient condition for asymptotic stability, it is \emph{not} necessary in general (\cf~the last example in \cref{app:preferences}).
It is therefore an open question whether asymptotically stable sets admit a straightforward characterization in terms of payoff conditions that can be checked at the level of pure strategies.

Finally, while our results concern attractors, a full account of the players' limiting behavior remains elusive:
for example, \cref{fig:four-games} shows four games with the same preference graph and no proper attractors, but markedly different limiting behaviors.
A plausible conjecture is that any asymptotically stable limit set of the replicator dynamics\textemdash or \ac{FTRL}\textemdash must be contained in the boundary of the game's strategy space \cite{Mer23};
however, this likely requires new ideas, which we intend to pursue in future work.

\section*{Acknowledgments}
\begingroup
\small
%
%
This research was supported in part by the French National Research Agency (ANR) in the framework of the PEPR IA FOUNDRY project (ANR-23-PEIA-0003).
PM is also a member of the Archimedes Research Unit/Athena RC, and was partially supported by project MIS 5154714 of the National Recovery and Resilience Plan Greece 2.0 funded by the European Union under the NextGenerationEU Program.
We would also like to thank the IRP FAIRGAME of the CNRS for supporting this project.
\endgroup

\appendix
\crefalias{section}{appendix}
\crefalias{subsection}{appendix}

\numberwithin{equation}{section}	
\numberwithin{lemma}{section}	
\numberwithin{proposition}{section}	
\numberwithin{theorem}{section}	
\numberwithin{corollary}{section}	
\numberwithin{definition}{section}	
\numberwithin{assumption}{section}	
\numberwithin{remark}{section}	

\section{Further related work}
\label{app:related}

The study of attractors of the replicator dynamics originates in evolutionary game theory, where the replicator equation was introduced as a model of evolutionary selection \citep{TJ78,HS98,San10}. 
Already in this literature, the long-run behavior of learning was understood to be set-valued, and a seminal contribution in this direction is due to \citet{RW95}, who studied the asymptotic stability of faces under evolutionary dynamics and showed that subgames closed under better replies induce attracting faces, which is \Cref{thm:subgame-stab} for the special case of the replicator.

These questions later became central for learning in games. 
The replicator dynamics are the continuous-time limit of exponential weights \citep{HSV09,Sor09}, while exponential weights is one of the few quintessential no-regret algorithm \citep{CBL06}. 
More generally, regularized dynamics, such as \ac{FTRL}, provide a common, unified framework for no-regret learning \citep{SS11,Haz16,MS16}.
Thus, the question of identifying attractors of the replicator naturally extends to the broader question of identifying the stable long-run outcomes of regularized learning dynamics.

For special classes of games, the picture is relatively well understood. 
In zero-sum and adversarial settings, regularized learning typically cycles rather than converges \citep{MPP18}. 
For zero-sum games, \cite{BS24} showed that the span of the sink equilibrium is the unique global attractor of the replicator dynamics. 
Moreover, \cite{LMP24,LMPP+24} showed that harmonic games, a generalization of zero-sum games with interior equilibria, are Poincar\'e recurrent in the interior under \ac{FTRL}, and therefore admit no proper attractor.
By contrast, in potential games, the potential acts as a Lyapunov function, yielding convergence to \aclp{NE} under broad classes of game dynamics \citep{MerSan18}. 
Outside such structured classes, the long-run behavior of learning in finite games remains poorly understood.

A recent line of work studies this problem through the preference graph of the game. 
\citet{PP19} conjectured that, for the replicator dynamics, minimal attractors should exist and should correspond to sink equilibria of the preference graph. 
\citet{BS23} proved that minimal attractors do exist and always contain sink equilibria. 
They also proved, for the replicator dynamics, a weaker form of \Cref{lem:ict-span}: any chain transitive set which contains a strongly connected set of pure profiles also contains its span. 
In the same paper, they also showed that \ac{club} subgames are attracting for the replicator dynamics, essentially recovering the same result of \citet{RW95}. 
This had already been extended prior to discrete-time \ac{FTRL} by \cite{BM23}.
However, as noted by \cite{CLM25}, their proof contains a gap: the energy function used there does not generally induce neighborhoods of \ac{club} faces, and hence does not by itself establish asymptotic stability. 
The energy of \cite{CLM25}, based on the Bregman divergence, fixes this for steep regularizers (where they work with a stochastic version of continuous-time \ac{FTRL}), while ours (the Fenchel gap) allows to prove the result for both steep and non-steep regularizers in continuous-time.

\section{Preference graphs}
\label{app:preferences}

We recall that the \emph{preference graph} has vertex set $\mathcal{A}$ and an arc $\alpha\to\alpha'$ whenever the two profiles are $i$-comparable for some player $i\in N$ (i.e.~$\alpha_{-i}=\alpha'_{-i}$) with $u_i(\alpha')\ge u_i(\alpha)$.

\subsection{Sink equilibria}
Note that, given any directed graph, we can decompose its set of vertices into a disjoint union of \emph{strongly connected components} (SCCs), which are strongly connected sets that admit no strongly connected superset. Given now such a decomposition, there will always exists among them some components which have no outgoing edge, known as \emph{sink strongly connected components} (sink SCCs).

When applied to the preference graph, these sink SCCs have been called the \emph{sink equilibria} of the game \cite{BS25}. They are therefore nonempty subsets of pure profiles which are both closed under better replies and strongly connected in the preference graph. Equivalently, they are also nonempty minimal \ac{club} sets, that is, \ac{club} sets that contain no proper \ac{club} set. Thus, once a better-reply path enters a sink equilibrium, it can keep moving within it, but it cannot leave it. In this sense, sink equilibrium represent the terminal ``basins'' of the ordinal deviation structure.

This interpretation is also reflected by a probabilistic viewpoint \cite{PP19, HMPP24, OPPT+19} : consider a random walk on $\A$ whose transitions $P$ follow the preference edges, i.e. $P(\alpha,\alpha') >0$ if and only if $\alpha\to\alpha'$. An interpretation of such dynamics is that, at each turn, given we are on a pure profile $\alpha$, we pick a random player who has one or many profitable deviations from $\alpha$, who in turn picks at random such a deviation.
Then, the classical ergodic theorem on Markov chains shows that the recurrent classes of this chain coincide with the sink SCCs of the preference graph, i.e. the sink equilibria, and every invariant probability measure $\pi$ is therefore supported on the union of sink equilibria.

\subsection{Examples}\label{sub:examples}
We now illustrate this with some example. In each case, the preference graph is drawn on the left, and the payoff table is
given on the right. The sink equilibria vertices are colored red, and their spans are red as well.

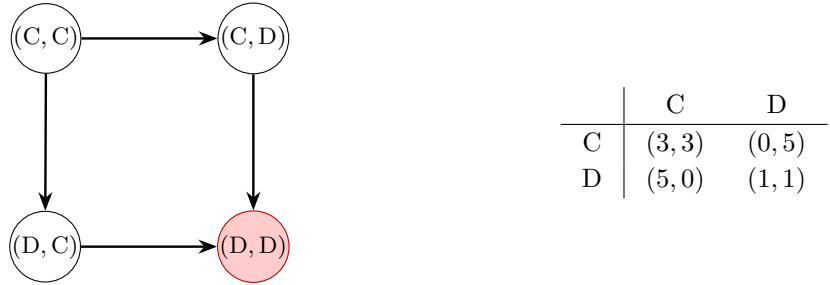
\begin{figure}[H]
\centering
\begin{minipage}{0.48\linewidth}
\centering

\begin{tikzpicture}[>=Stealth, node distance=18mm,
  node/.style={circle,draw,minimum size=5.5mm,inner sep=0pt,font=\small},
  sink/.style={node,draw=red!70!black,fill=red!20},
  edge/.style={->,draw=black,line width=0.9pt}
]
\node[node] (CC) at (0,1.6) {$(\mathrm C,\mathrm C)$};
\node[node] (CD) [right=of CC] {$(\mathrm C,\mathrm D)$};
\node[sink] (DD) [below=of CD] {$(\mathrm D,\mathrm D)$};
\node[node] (DC) [left=of DD] {$(\mathrm D,\mathrm C)$};
\draw[edge] (CC) -- (DC);
\draw[edge] (CC) -- (CD);
\draw[edge] (DC) -- (DD);
\draw[edge] (CD) -- (DD);
\end{tikzpicture}

\label{fig:pd-graph}
\end{minipage}\hfill
\begin{minipage}{0.48\linewidth}
\centering
\renewcommand{\arraystretch}{1.15}
\setlength{\tabcolsep}{8pt}
\begin{tabular}{c|cc}
 & $\mathrm C$ & $\mathrm D$ \\ \hline
$\mathrm C$ & $(3,3)$ & $(0,5)$ \\
$\mathrm D$ & $(5,0)$ & $(1,1)$
\end{tabular}
\label{tab:pd}
\end{minipage}
\caption{\emph{Prisoner's Dilemma}}
\end{figure}
This is the well known \emph{Prisoner's Dilemma}. The preference graph is acyclic: all paths flow to $(\mathrm D,\mathrm D)$, which is a sink vertex, and therefore the unique strict \acl{NE}.

\begin{figure}[H]
\centering
\begin{minipage}{0.48\linewidth}
\centering

\begin{tikzpicture}[>=Stealth, node distance=18mm,
  node/.style={circle,draw,minimum size=5.5mm,inner sep=0pt,font=\small},
  sink/.style={node,draw=red!70!black,fill=red!20},
  sinkedge/.style={->,draw=red!70!black,line width=0.9pt}
]
\node[sink] (HH) at (0,1.6) {$(\mathrm H, \mathrm H)$};
\node[sink] (HT) [right=of HH] {$(\mathrm H,\mathrm T)$};
\node[sink] (TT) [below=of HT] {$(\mathrm T,\mathrm T)$};
\node[sink] (TH) [left=of TT] {$(\mathrm T,\mathrm H)$};
\draw[sinkedge] (HH) -- (HT);
\draw[sinkedge] (HT) -- (TT);
\draw[sinkedge] (TT) -- (TH);
\draw[sinkedge] (TH) -- (HH);
\end{tikzpicture}

\label{fig:mp-graph}
\end{minipage}\hfill
\begin{minipage}{0.48\linewidth}
\centering
\renewcommand{\arraystretch}{1.15}
\setlength{\tabcolsep}{8pt}
\begin{tabular}{c|cc}
 & $\mathrm H$ & $\mathrm T$ \\ \hline
$\mathrm H$ & $(1,-1)$ & $(-1,1)$ \\
$\mathrm T$ & $(-1,1)$ & $(1,-1)$
\end{tabular}
\label{tab:mp}
\end{minipage}
\caption{Matching Pennies}
\end{figure}
This is the \emph{Matching Pennies} game. It is a zero-sum game, and its preference graph forms a directed 4--cycle, hence the graph is strongly connected and the unique \ac{club} set is the whole graph.

\begin{figure}[H]
\centering
\begin{minipage}{0.48\linewidth}
\centering
\begin{tikzpicture}[>=Stealth, node distance=18mm,
  node/.style={circle,draw,minimum size=5.5mm,inner sep=0pt,font=\small},
  sink/.style={node,draw=red!70!black,fill=red!20},
  edge/.style={->,draw=black,line width=0.9pt}
]
\node[sink] (AA) at (0,1.6) {$(\mathrm A,\mathrm A)$};
\node[node] (AB) [right=of AA] {$(\mathrm A,\mathrm B)$};
\node[sink] (BB) [below=of AB] {$(\mathrm B,\mathrm B)$};
\node[node] (BA) [left=of BB] {$(\mathrm B,\mathrm A)$};

\draw[edge] (AB) -- (AA);
\draw[edge] (AB) -- (BB);
\draw[edge] (BA) -- (AA);
\draw[edge] (BA) -- (BB);
\end{tikzpicture}
\label{fig:coord-graph}
\end{minipage}\hfill
\begin{minipage}{0.48\linewidth}
\centering
\renewcommand{\arraystretch}{1.15}
\setlength{\tabcolsep}{8pt}
\begin{tabular}{c|cc}
 & $\mathrm A$ & $\mathrm B$ \\ \hline
$\mathrm A$ & $(1,1)$ & $(0,0)$ \\
$\mathrm B$ & $(0,0)$ & $(1,1)$
\end{tabular}
\label{tab:coord}
\end{minipage}
\caption{A coordination game.}
\end{figure}
The preference graph is acyclic and has two sink vertices, $(\mathrm A,\mathrm A)$ and $(\mathrm B,\mathrm B)$, so the game has two strict \aclp{NE} (and hence two singleton \ac{club} sets), corresponding to the two coordinated outcomes.

\begin{figure}[H]
\centering
\begin{minipage}{0.56\linewidth}
\centering
\resizebox{4.5cm}{!}{
\begin{tikzpicture}[
	scale=3,
	line join=round,
	line cap=round,
	edge/.style={semithick,draw=black},
	sinkspan/.style={draw=MyDarkRed,very thick},
	midorient/.style={-{Stealth},thick,draw=black},
	sinkorient/.style={-Stealth,thick,draw=MyDarkRed}
]

\coordinate (FLB) at (0,0);
\coordinate (FRB) at (1,0);
\coordinate (FRT) at (1,1);
\coordinate (FLT) at (0,1);

\coordinate (BLB) at (0.40,0.30);
\coordinate (BRB) at (1.40,0.30);
\coordinate (BRT) at (1.40,1.30);
\coordinate (BLT) at (0.40,1.30);

\fill[MyLightRed,opacity=0.25] (FLT)--(FRT)--(FRB)--(FLB)--cycle;

\draw[edge] (FLT)--(FRT)--(FRB)--(FLB)--cycle;
\draw[edge] (BLT)--(BRT)--(BRB)--(BLB)--cycle;
\draw[edge] (FLT)--(BLT);
\draw[edge] (FRT)--(BRT);
\draw[edge] (FRB)--(BRB);
\draw[edge] (FLB)--(BLB);

\draw[sinkspan] (FLT)--(FRT)--(FRB)--(FLB)--cycle;

\draw[sinkorient] ($(FLT)!0.48!(FRT)$) -- ($(FLT)!0.52!(FRT)$);
\draw[sinkorient] ($(FRB)!0.48!(FLB)$) -- ($(FRB)!0.52!(FLB)$);
\draw[midorient] ($(BRT)!0.48!(BLT)$) -- ($(BRT)!0.52!(BLT)$);
\draw[midorient] ($(BLB)!0.48!(BRB)$) -- ($(BLB)!0.52!(BRB)$);

\draw[sinkorient] ($(FLB)!0.48!(FLT)$) -- ($(FLB)!0.52!(FLT)$);
\draw[sinkorient] ($(FRT)!0.48!(FRB)$) -- ($(FRT)!0.52!(FRB)$);
\draw[midorient] ($(BLT)!0.48!(BLB)$) -- ($(BLT)!0.52!(BLB)$);
\draw[midorient] ($(BRB)!0.48!(BRT)$) -- ($(BRB)!0.52!(BRT)$);

\draw[midorient] ($(BLT)!0.48!(FLT)$) -- ($(BLT)!0.52!(FLT)$);
\draw[midorient] ($(BRT)!0.48!(FRT)$) -- ($(BRT)!0.52!(FRT)$);
\draw[midorient] ($(BLB)!0.48!(FLB)$) -- ($(BLB)!0.52!(FLB)$);
\draw[midorient] ($(BRB)!0.48!(FRB)$) -- ($(BRB)!0.52!(FRB)$);

\foreach \V in {FLB,FRB,FRT,FLT}{
	\fill[MyDarkRed] (\V) circle (0.9pt);
}
\foreach \V in {BLB,BRB,BRT,BLT}{
	\fill[black] (\V) circle (0.9pt);
}
\end{tikzpicture}
}
\label{fig:face-subgame}
\end{minipage}\hfill
\begin{minipage}{0.40\linewidth}
\centering
\renewcommand{\arraystretch}{1.15}
\setlength{\tabcolsep}{6pt}
\begin{tabular}{c}
\textbf{Front} \\[-2pt]
\begin{tabular}{c|cc}
   & $\mathrm L$ & $\mathrm R$ \\ \hline
$\mathrm T$ & $(1,0,1)$  & $(0,1,1)$ \\
$\mathrm B$ & $(0,1,1)$  & $(1,0,1)$
\end{tabular}
\\[8pt]
\textbf{Back} \\[-2pt]
\begin{tabular}{c|cc}
   & $\mathrm L$ & $\mathrm R$ \\ \hline
$\mathrm T$ & $(-1,1,0)$ & $(1,0,0)$ \\
$\mathrm B$ & $(1,0,0)$  & $(-1,1,0)$
\end{tabular}
\end{tabular}
\label{tab:face-subgame}
\end{minipage}
\caption{A cube with a \ac{club} square.}
\end{figure}
A $2\times2\times2$ game, for which we represent pure profiles as the vertices of a cube. Player 1 choses to play top or bottom, player 2 left or right and player 3 front or back. Edges carry arrows for profitable unilateral deviations. The $2 \times 2$ subgame where the third player plays front only is closed under better replies.

\begin{figure}[H]
\centering
\begin{minipage}{0.56\linewidth}
\centering
\resizebox{4.5cm}{!}{
\begin{tikzpicture}[
	scale=3,
	line join=round,
	line cap=round,
	edge/.style={semithick,draw=black},
	sinkspan/.style={draw=MyDarkRed,very thick},
	midorient/.style={-{Stealth},thick,draw=black},
	sinkorient/.style={-Stealth,thick,draw=MyDarkRed}
]

\coordinate (FLB) at (0,0);
\coordinate (FRB) at (1,0);
\coordinate (FRT) at (1,1);
\coordinate (FLT) at (0,1);

\coordinate (BLB) at (0.40,0.30);
\coordinate (BRB) at (1.40,0.30);
\coordinate (BRT) at (1.40,1.30);
\coordinate (BLT) at (0.40,1.30);

\fill[MyLightRed,opacity=0.25] (FLT)--(FRT)--(BRT)--(BLT)--cycle;
\fill[MyLightRed,opacity=0.25] (FRT)--(FRB)--(BRB)--(BRT)--cycle;

\draw[edge] (FLT)--(FRT)--(FRB)--(FLB)--cycle;
\draw[edge] (BLT)--(BRT)--(BRB)--(BLB)--cycle;
\draw[edge] (FLT)--(BLT);
\draw[edge] (FRT)--(BRT);
\draw[edge] (FRB)--(BRB);
\draw[edge] (FLB)--(BLB);

\draw[sinkspan] (FLT)--(FRT)--(BRT)--(BLT)--cycle;
\draw[sinkspan] (FRT)--(FRB)--(BRB)--(BRT)--cycle;

\draw[sinkorient] ($(FLT)!0.48!(FRT)$) -- ($(FLT)!0.52!(FRT)$);
\draw[sinkorient] ($(FRT)!0.48!(FRB)$) -- ($(FRT)!0.52!(FRB)$);
\draw[sinkorient] ($(FRB)!0.48!(BRB)$) -- ($(FRB)!0.52!(BRB)$);
\draw[sinkorient] ($(BRB)!0.48!(BRT)$) -- ($(BRB)!0.52!(BRT)$);
\draw[sinkorient] ($(BRT)!0.48!(BLT)$) -- ($(BRT)!0.52!(BLT)$);
\draw[sinkorient] ($(BLT)!0.48!(FLT)$) -- ($(BLT)!0.52!(FLT)$);
\draw[sinkorient] ($(BRT)!0.48!(FRT)$) -- ($(BRT)!0.52!(FRT)$);

\draw[midorient] ($(FLB)!0.48!(FLT)$) -- ($(FLB)!0.52!(FLT)$);
\draw[midorient] ($(FLB)!0.48!(FRB)$) -- ($(FLB)!0.52!(FRB)$);
\draw[midorient] ($(FLB)!0.48!(BLB)$) -- ($(FLB)!0.52!(BLB)$);
\draw[midorient] ($(BLB)!0.48!(BLT)$) -- ($(BLB)!0.52!(BLT)$);
\draw[midorient] ($(BLB)!0.48!(BRB)$) -- ($(BLB)!0.52!(BRB)$);

\foreach \V in {FRB,FRT,FLT,BRB,BRT,BLT}{
	\fill[MyDarkRed] (\V) circle (0.9pt);
}
\foreach \V in {FLB,BLB}{
	\fill[black] (\V) circle (0.9pt);
}

\end{tikzpicture}
}
\label{fig:union-faces}
\end{minipage}\hfill
\begin{minipage}{0.40\linewidth}
\centering
\renewcommand{\arraystretch}{1.15}
\setlength{\tabcolsep}{6pt}
\begin{tabular}{c}
\textbf{Front} \\[-2pt]
\begin{tabular}{c|cc}
   & $\mathrm L$ & $\mathrm R$ \\ \hline
$\mathrm T$ & $(2,-1,1)$  & $(0,1,1)$ \\
$\mathrm B$ & $(-2,-2,-2)$ & $(1,1,-1)$
\end{tabular}
\\[8pt]
\textbf{Back} \\[-2pt]
\begin{tabular}{c|cc}
   & $\mathrm L$ & $\mathrm R$ \\ \hline
$\mathrm T$ & $(1,0,-1)$  & $(1,-1,0)$ \\
$\mathrm B$ & $(-2,-2,0)$ & $(-2,0,1)$
\end{tabular}
\end{tabular}
\label{tab:union-faces}
\end{minipage}
\caption{A game with a \ac{srad} set.}
\end{figure}
A $2 \times 2 \times 2$ game with a \ac{srad} sink equilibrium.
Indeed let $\HH$ be that sink, a direct computation yields, $\Phi(\beta,\alpha)\le -1<0$ for all $\alpha\notin\HH$ and all $\beta\in\HH$.

    \begin{figure}[H]
\centering
\begin{minipage}{0.56\linewidth}
\centering
\resizebox{4.5cm}{!}{
\begin{tikzpicture}[
scale=3,
line join=round,
line cap=round,
edge/.style={semithick,draw=black},
rededge/.style={very thick,draw=MyDarkRed},
midorient/.style={-{Stealth},thick,draw=black},
sinkorient/.style={-Stealth,thick,draw=MyDarkRed}
]
\coordinate (FLB) at (0,0);
\coordinate (FRB) at (1,0);
\coordinate (FRT) at (1,1);
\coordinate (FLT) at (0,1);
\coordinate (BLB) at (0.40,0.30);
\coordinate (BRB) at (1.40,0.30);
\coordinate (BRT) at (1.40,1.30);
\coordinate (BLT) at (0.40,1.30);

\draw[edge]    (BLB)--(BRB)--(BRT);
\draw[rededge] (BRT)--(BLT)--(BLB);

\draw[rededge] (FLB)--(BLB);
\draw[edge]    (FRB)--(BRB);
\draw[rededge] (FRT)--(BRT);
\draw[edge]    (FLT)--(BLT);

\draw[rededge] (FLB)--(FRB)--(FRT);
\draw[edge]    (FRT)--(FLT)--(FLB);

\draw[midorient]
  ($(FLT)!0.48!(FRT)$) -- ($(FLT)!0.52!(FRT)$);
\draw[sinkorient]
  ($(BLT)!0.48!(BRT)$) -- ($(BLT)!0.52!(BRT)$);

\draw[midorient]
  ($(BRB)!0.48!(BLB)$) -- ($(BRB)!0.52!(BLB)$);
\draw[sinkorient]
  ($(FRB)!0.48!(FLB)$) -- ($(FRB)!0.52!(FLB)$);

\draw[midorient]
  ($(FLT)!0.48!(FLB)$) -- ($(FLT)!0.52!(FLB)$);
\draw[sinkorient]
  ($(FRT)!0.48!(FRB)$) -- ($(FRT)!0.52!(FRB)$);

\draw[sinkorient]
  ($(BLB)!0.48!(BLT)$) -- ($(BLB)!0.52!(BLT)$);
\draw[midorient]
  ($(BRB)!0.48!(BRT)$) -- ($(BRB)!0.52!(BRT)$);

\draw[midorient]
  ($(FLT)!0.48!(BLT)$) -- ($(FLT)!0.52!(BLT)$);
\draw[sinkorient]
  ($(BRT)!0.48!(FRT)$) -- ($(BRT)!0.52!(FRT)$);

\draw[sinkorient]
  ($(FLB)!0.48!(BLB)$) -- ($(FLB)!0.52!(BLB)$);
\draw[midorient]
  ($(BRB)!0.48!(FRB)$) -- ($(BRB)!0.52!(FRB)$);

\foreach \V in {FLB,FRB,FRT,BLB,BRT,BLT}{
  \fill[MyDarkRed] (\V) circle (0.9pt);
}
\foreach \V in {FLT,BRB}{
  \fill[black] (\V) circle (0.9pt);
}

\end{tikzpicture}
}
\label{fig:jordan-cube}
\end{minipage}\hfill
\begin{minipage}{0.4\linewidth}
\centering
\renewcommand{\arraystretch}{1.15}
\setlength{\tabcolsep}{6pt}
\begin{tabular}{c}
\textbf{Front} \\[-2pt]
\begin{tabular}{c|cc}
   & \(\mathrm L\) & \(\mathrm R\) \\ \hline
\(\mathrm T\) & \((-1,-1,-1)\) & \((-1,1,1)\) \\
\(\mathrm B\) & \((1,1,-1)\)   & \((1,-1,1)\)
\end{tabular}
\\[8pt]
\textbf{Back} \\[-2pt]
\begin{tabular}{c|cc}
   & \(\mathrm L\) & \(\mathrm R\) \\ \hline
\(\mathrm T\) & \((1,-1,1)\)  & \((1,1,-1)\) \\
\(\mathrm B\) & \((-1,1,1)\)  & \((-1,-1,-1)\)
\end{tabular}
\end{tabular}
\label{tab:jordan}
\end{minipage}
\caption{Jordan's Matching Pennies}
\end{figure}
This game, known as \emph{Jordan's Matching Pennies}, was introduced by \citet{Jor93} and studied in \cite{GH95} as a compact, explicit instance in which adaptive play can fail to converge to \acl{NE} and instead exhibits persistent cycling concentrated on the red $6$-cycle (the span of the sink equilibrium). It also provides a minimal example where long-run behavior is naturally described by sink equilibria rather than \aclp{NE} \cite{KLPT11}.

Now, let $\HH$ denote the sink equilibrium of Jordan's Matching Pennies (the six red vertices in the previous figure).
    Here $\A\setminus\HH=\{(\mathrm{T,L,F}),(\mathrm{B,R,B})\}$, and a direct computation gives $\Phi(\beta,\alpha) \in \{0, -2\}$ for all such $\alpha \notin \HH$ and $\beta \in \HH$. Hence $\HH$ is \ac{rad} but \emph{not} \ac{srad}.

\begin{figure}[H]
\centering
\begin{minipage}{0.50\linewidth}
\centering
\begin{tikzpicture}[>=Stealth, line cap=round,
  node/.style={circle,draw,minimum size=6.2mm,inner sep=0pt,font=\small},
  sink/.style={node,draw=red!70!black,fill=red!20},
  edge/.style={->,draw=black,line width=0.9pt},
  sinkedge/.style={->,draw=red!70!black,line width=0.9pt}
]
\node[sink] (n12) at (0,1.85)      {$(\mathrm A,\mathrm B)$};
\node[sink] (n13) at (1.65,0.93)   {$(\mathrm A,\mathrm C)$};
\node[sink] (n23) at (1.65,-0.93)  {$(\mathrm B,\mathrm C)$};
\node[sink] (n21) at (0,-1.85)     {$(\mathrm B, \mathrm A)$};
\node[sink] (n31) at (-1.65,-0.93) {$(\mathrm C,\mathrm A)$};
\node[sink] (n32) at (-1.65,0.93)  {$(\mathrm C,\mathrm B)$};

\draw[sinkedge] (n12) -- (n13);
\draw[sinkedge] (n13) -- (n23);
\draw[sinkedge] (n23) -- (n21);
\draw[sinkedge] (n21) -- (n31);
\draw[sinkedge] (n31) -- (n32);
\draw[sinkedge] (n32) -- (n12);

\node[node] (d33) at (0,3.25)      {$(\mathrm C,\mathrm C)$};
\node[node] (d11) at (-3.25,-2.25) {$(\mathrm A,\mathrm A)$};
\node[node] (d22) at (3.25,-2.25)  {$(\mathrm B,\mathrm B)$};

\draw[edge,bend left=8]  (d11) to (n31);
\draw[edge,bend left=2]  (d11) to (n21);
\draw[edge,bend right=38] (d11) to (n12);
\draw[edge,bend right=20] (d11) to (n13);

\draw[edge,bend right=2] (d22) to (n21);
\draw[edge,bend right=8] (d22) to (n23);
\draw[edge,bend left=38] (d22) to (n12);
\draw[edge,bend left=20] (d22) to (n32);

\draw[edge] (d33) .. controls (-1.20,2.90) and (-2.55,2.05) .. (n32);
\draw[edge] (d33) .. controls ( 1.20,2.90) and ( 2.55,2.05) .. (n13);

\draw[edge] (d33) .. controls (-3.20,2.85) and (-3.20,-0.10) .. (n31);
\draw[edge] (d33) .. controls ( 3.20,2.85) and ( 3.20,-0.10) .. (n23);

\end{tikzpicture}
\label{fig:shapley-graph}
\end{minipage}\hfill
\begin{minipage}{0.46\linewidth}
\centering
\renewcommand{\arraystretch}{1.15}
\setlength{\tabcolsep}{10pt}
\begin{tabular}{c|ccc}
 & $\mathrm A$ & $\mathrm B$ & $\mathrm C$ \\ \hline
$\mathrm A$ & $(0,0)$ & $(2,1)$ & $(1,2)$ \\
$\mathrm B$ & $(1,2)$ & $(0,0)$ & $(2,1)$ \\
$\mathrm C$ & $(2,1)$ & $(1,2)$ & $(0,0)$
\end{tabular}
\label{tab:shapley}
\end{minipage}
\caption{Shapley's game.}
\end{figure}
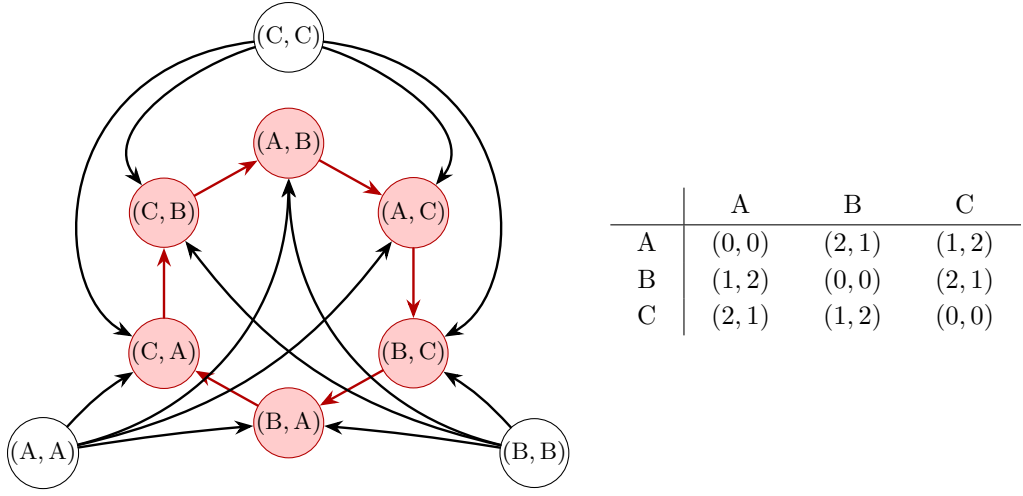
Shapley's classical $3\times 3$ example \cite{Sha64} is historically important as (one of) the first explicit demonstrations that learning can fail to converge to \acl{NE} in two-player non-zero-sum games, where it was shown specifically for the fictitious play scheme \cite{KriSjo98}.

Let $\HH:=\{(i,j): i\neq j\}$ be the sink equilibrium of Shapley's game (the red hexagon in the previous figure).  Its span is exactly this directed cycle (union of one-dimensional faces). For $\alpha=(k,k)\notin\HH$ and $\beta=(i,j)\in\HH$,
\[
\Phi(\beta,\alpha)
=\big(u_1(k,j)-u_1(i,j)\big)\;+\;\big(u_2(i,k)-u_2(i,j)\big).
\]
If $k=i$ then $u_1(k,j)=u_1(i,j)$ and $u_2(i,k)=u_2(i,i)=0$, so $\Phi(\beta,\alpha)=-u_2(i,j)\in\{-1,-2\}<0$.
If $k=j$ similarly $\Phi(\beta,\alpha)=-u_1(i,j)\in\{-1,-2\}<0$.
If $k\notin\{i,j\}$, then $u_1(i,j)+u_2(i,j)=3$ and also $u_1(k,j)+u_2(i,k)=3$, hence $\Phi(\beta,\alpha)=0$.
Therefore $\Phi(\beta,\alpha)\le 0$ for all $\alpha\notin\HH$, $\beta\in\HH$, so $\HH$ is \ac{rad}, but it is not \ac{srad} since equality occurs (e.g.\ $\alpha=(\mathrm A,\mathrm A)$ and $\beta\in\{(\mathrm B,\mathrm C),(\mathrm C,\mathrm B)\}$).

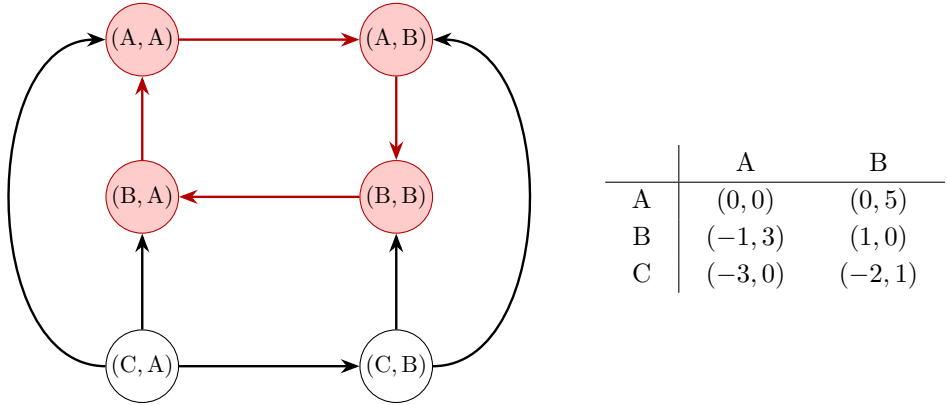
\begin{figure}[H]
\centering
\begin{minipage}{0.50\linewidth}
\centering
\begin{tikzpicture}[
  scale=0.8, transform shape, >=Stealth,
  node/.style={circle,draw,minimum size=12mm,inner sep=0pt,font=\large},
  clubnode/.style={node,draw=red!70!black,fill=red!20},
  edge/.style={->,draw=black,line width=0.9pt},
  clubedge/.style={edge,draw=red!70!black}
]
\node[clubnode] (aA) at (0,4.8) {$(\mathrm A,\mathrm A)$};
\node[clubnode] (aB) at (4.2,4.8) {$(\mathrm A,\mathrm B)$};

\node[clubnode] (bA) at (0,2.2) {$(\mathrm B,\mathrm A)$};
\node[clubnode] (bB) at (4.2,2.2) {$(\mathrm B,\mathrm B)$};

\node[node] (cA) at (0,-0.6) {$(\mathrm C,\mathrm A)$};
\node[node] (cB) at (4.2,-0.6) {$(\mathrm C,\mathrm B)$};

\draw[clubedge] (aA) -- (aB);
\draw[clubedge] (aB) -- (bB);
\draw[clubedge] (bB) -- (bA);
\draw[clubedge] (bA) -- (aA);

\draw[edge] (cA) -- (bA);
\draw[edge] (cA) -- (cB);
\draw[edge] (cB) -- (bB);

\draw[edge,out=180,in=180,looseness=1] (cA) to (aA);
\draw[edge,out=0,in=0,looseness=1]     (cB) to (aB);

\end{tikzpicture}
\label{fig:nonresilient}
\end{minipage}\hfill
\begin{minipage}{0.48\linewidth}
\centering
\renewcommand{\arraystretch}{1.15}
\setlength{\tabcolsep}{10pt}
\begin{tabular}{c|cc}
 & $\mathrm A$ & $\mathrm B$ \\ \hline
$\mathrm A$ & $(0,0)$   & $(0,5)$ \\
$\mathrm B$ & $(-1,3)$  & $(1,0)$ \\
$\mathrm C$ & $(-3,0)$  & $(-2,1)$
\end{tabular}
\label{tab:nonresilient}
\end{minipage}
\caption{A \ac{club} subgame which is not \ac{rad}.}
\end{figure}

The \ac{club} subgame is $\B=\{\mathrm A,\mathrm B\}\times\{\mathrm A,\mathrm B\}$. Nonetheless, $\B$ is not \ac{rad}: with $\beta=(\mathrm A,\mathrm A)\in\B$ and $\alpha=(\mathrm C,\mathrm B)\notin\B$,
\begin{equation}
\Phi(\beta,\alpha)=(u_1(\mathrm C,\mathrm A)-u_1(\mathrm A,\mathrm A))+(u_2(\mathrm A,\mathrm B)-u_2(\mathrm A,\mathrm A))=(-3-0)+(5-0)=2>0.
\end{equation}
Since \ac{club} subgames are asymptotically stable, this shows \ac{radness} is not necessary for asymptotic stability.

\section{Attractors and chain transitivity}
\label{app:attractors}

We collect here standard notions from dynamical systems used throughout this paper, see e.g. \cite{San10, AN07}.

Let $X$ be a compact subset of $\R^d$. We shall consider a \emph{flow}, which is a continuous map $\phi:\mathbb{R}\times X\to X$ such that
\[
\phi_0=\mathrm{Id}_X,\qquad \phi_{t+s}=\phi_t\circ\phi_s\quad\forall\,s,t\in\mathbb{R}.
\]
We write $\phi_t(x_0)$ for the state at time $t$ starting from $x_0\in X$.
All flows in this paper arise from (globally) well-posed ODEs of the form $\dot{x} = f(x)$, where $f : X \to \R^d$ is Lipschitz continuous.

\subsection{Attractors}\label{subsec:attractors}
A set $A\subseteq X$ is \emph{forward-invariant} if $\phi_t(A)\subseteq A$ for all $t\ge 0$, and \emph{invariant} if $\phi_t(A) = A$ for all $t\in \mathbb{R}$. A point $x^\ast\in X$ is \emph{stationary} if $\phi_t(x^\ast)=x^\ast$ for all $t \ge 0$. Let $S\subseteq X$ be nonempty and closed and let $\mathcal{N}(S)$ denote the family of neighborhoods of $S$, then:
\begin{itemize}
\item \emph{Stability.} $S$ is \emph{stable} if for every $U\in\mathcal{N}(S)$ there exists $V\in\mathcal{N}(S)$ with
\[
x\in V \ \Longrightarrow\ \phi_t(x)\in U\quad\forall\,t\ge 0.
\]

\item \emph{Attraction.} $S$ is \emph{attracting} if there exists $U\in\mathcal{N}(S)$ such that
\[
x\in U \ \Longrightarrow\ \dist(\phi_t(x),S)\xrightarrow[t\to\infty]{}0.
\]
Any such $U$ is a \emph{basin of attraction} of $S$.

\item \emph{Asymptotic stability.} $S$ is \emph{asymptotically stable} if it is both stable and attracting.
\end{itemize}

A nonempty, compact, invariant set $A\subseteq X$ is an \emph{attractor} if there exists $U\in\mathcal{N}(A)$ with
\[
\lim_{t\to\infty}\ \sup_{x\in U} \dist(\phi_t(x),A)=0.
\]
Equivalently, $A$ is an attractor if and only if there exists an open, forward-invariant \emph{trapping region} $B$ and $t_0>0$ such that
$\phi_t(\operatorname{cl} B)\subset B$ for all $t\ge t_0$ and
\(
A=\bigcap_{t\ge 0}\phi_t(B).
\) In particular, every attractor is asymptotically stable, and any compact, invariant and asymptotically stable set is an attractor \cite{San10}. If $A$ is an attractor with trapping region $B$, the \emph{dual repellor} is
\[
A^\ast:=\bigcap_{t < 0}\phi_t\bigl(X\setminus B\bigr),
\]
which is an attractor for the time-reversed flow $\psi_t := \phi_{-t}$ and satisfies $A\cap A^\ast=\varnothing$.
Finally, we cite the following useful lemmas:
\begin{lemma}\label{lem:restrict-attract}
Let $A\subseteq X$ be an attractor, let $K\subseteq X$ be nonempty, compact and invariant and assume $A\cap K\neq\varnothing$. Then $A\cap K$ is an attractor for the restricted flow
$\phi|_K:\mathbb R\times K\to K$.
\end{lemma}

\subsection{Chain transitivity}
For $\varepsilon,T>0$, an \emph{$(\varepsilon,T)$--chain} from $x$ to $y$ is a finite sequence
\[
x=x_0,\ x_1,\ \dots,\ x_k=y,\qquad t_0,\dots,t_{k-1}\ge T,
\]
such that $\dist(\phi_{t_j}(x_j),x_{j+1})<\varepsilon$ for all $j\in \{0, \dots, k-1\}$.
There is a \emph{pseudo-orbit} from $x$ to $y$ (denoted $x\rightsquigarrow y$) if for every $\varepsilon,T>0$ there exists an $(\varepsilon,T)$--chain from $x$ to $y$. Note that pseudo-orbits cannot leave, but may enter, attractors\textemdash and pseudo-orbits may leave, but cannot enter a repellor.

With this we can define the notion of \emph{chain transitivity}. A nonempty compact set $K\subseteq X$ is \emph{internally chain transitive} for the flow $\phi$ if it is invariant and, for every $x,y\in K$, $x \rightsquigarrow y$ for the flow \emph{restricted to $K$}. We now mention the following alternative definition, due to \cite{Ben99}:
\begin{proposition}\label{prop:ict-attractor}
Let $K\subseteq X$ be nonempty and compact, then $K$ is internally chain transitive if and only if $K$ is invariant and the flow restricted to $K$ admits no proper attractor.
\end{proposition}

\section{Regularizers and choice maps}
\label{app:mirror}

Throughout, fix a player $i\in\N$ and write $\X_i=\Delta(\A_i)$, $\YY_i=\R^{\A_i}$,
$\langle\cdot,\cdot\rangle$ for the standard pairing on $\YY_i\times\YY_i$,
and $\one\in\YY_i$ for the all-ones vector.
We view $h_i:\X_i\to\R\cup\{+\infty\}$ as an extended-real convex function on $\YY_i$
by setting $h_i(x_i)=+\infty$ for $x_i\notin\X_i$.
We assume throughout that $h_i$ is \emph{decomposable}, \emph{continuous} and \emph{smooth on the interior} :
\[
h_i(x_i)=\sum_{\alpha_i\in\A_i}\theta_i(x_{i\alpha}),
\quad \theta_i: [0,1] \to \R, \quad \theta_i \in C^0([0,1]),\quad \theta_i\in C^2((0,1]).
\]
We will moreover assume that there exists a constant $K_i>0$ such that
\begin{equation}\label{eq:curvature}
\theta_i''(p)\ge K_i \qquad \forall\,p\in(0,1].
\end{equation}
As a consequence, $h_i$ is \emph{$K_i$-strongly convex} on $\X_i$ (with respect to $\|\cdot\|_2$), i.e.\ for all
$x_i,x_i'\in\X_i$ and all $t\in(0,1]$,
\[
h_i\bigl((1-t)x_i+t x_i'\bigr)
\;\le\;
(1-t)h_i(x_i)+t h_i(x_i')-\frac{K_i}{2}\,t(1-t)\,\|x_i-x_i'\|_2^2.
\]
Finally, recall from the main text that $h_i$ is \emph{steep} if $\theta_i'(p)\to-\infty$ as $p\downarrow 0$.

\subsection{Choice maps}

The \emph{convex conjugate} of $h_i$ is
\[
h_i^*(y_i)\;:=\;\sup_{x_i\in\X_i}\{\langle y_i,x_i\rangle-h_i(x_i)\},\qquad y_i\in\YY_i,
\]
and the regularized \emph{choice map} (also known as the \emph{mirror map}) is
\[
Q_i(y_i)\;:=\;\arg\max_{x_i\in\X_i}\{\langle y_i,x_i\rangle-h_i(x_i)\}.
\]
By strong convexity, the maximizer is unique, so $Q_i:\YY_i\to\X_i$ is single-valued. We shall also consider the \emph{subdifferential} of $h_i$ at $x_i\in\X_i$:
\[
\partial h_i(x_i)
\;:=\;
\Bigl\{g_i\in\YY_i:\
h_i(w_i)\ge h_i(x_i)+\langle g_i,w_i-x_i\rangle\ \ \forall w_i\in\X_i\Bigr\},
\]
and we write $\dom\partial h_i:=\{x_i\in\X_i:\partial h_i(x_i)\neq\varnothing\}$ and $\ImQ_i = Q_i(\YY_i)$ for the image of $Q_i$, also known as the \emph{prox-domain}.

Under our assumptions, the choice map enjoys the following properties:

\begin{proposition}\label{prop:choice-map}
For each $i\in\N$:
\begin{enumerate}
\item\label{it:Qgrad} $h_i^*$ is $C^1$ on $\YY_i$ and $Q_i=\nabla h_i^*$, in particular $Q_i$ is Lipschitz continuous.
\item\label{it:Q-image} $\ImQ_i=\dom\partial h_i$.
\item\label{it:shift} For all $c\in\R$, $Q_i(y_i+c\one)=Q_i(y_i)$.
\item\label{it:vanish} If $y_i^k\in\YY_i$ satisfies $y_{i\alpha}^k-y_{i\beta}^k\to-\infty$ for some $\alpha\neq\beta$,
then $Q_{i\alpha}(y_i^k)\to 0$.
\item \label{it:KKT} For $x=Q(y)$ there exist $\lambda_i\in\mathbb{R}$ and $\nu_i\in\mathbb{R}_+^{\mathcal{A}_i}$ such that, for all $\alpha_i \in \A_i$,
\[
y_{i\alpha}=\theta_i'(x_{i\alpha})+\lambda_i - \nu_{i\alpha},
\]
and if $x_{i\alpha} \neq 0$, then $\nu_{i\alpha}=0$.
\end{enumerate}
\end{proposition}

\begin{proof}
    See \cite{MS16}.
\end{proof}

\begin{lemma}\label{lem:steep-image}
If $h_i$ is steep, then $\ImQ_i=\X_i^\circ$.
\end{lemma}

\begin{proof}
Direct consequence of~\cref{prop:choice-map}(\labelcref{it:Q-image}), see \cite{MS16}.
\end{proof}

\begin{proposition}\label{prop:Q-injective}
If $h_i$ is steep, then, for all $y_i,y_i'\in\YY_i$,
\[
Q_i(y_i)=Q_i(y_i')
\quad\Longleftrightarrow\quad
\exists\,c\in\R:\ y_i'=y_i+c\one.
\]
Equivalently, the restriction of $Q_i$ to the quotient space $\YY_i/\mathrm{span}\{\one\}$ is injective.
\end{proposition}
\begin{proof}
Let $x_i:=Q_i(y_i)=Q_i(y_i')\in\X_i^\circ$.
Since $x_i$ is interior, by \cref{prop:choice-map}(\labelcref{it:KKT})
\[
y_{i\alpha}-\theta_i'(x_{i\alpha})=\lambda,\qquad
y_{i\alpha}'-\theta_i'(x_{i\alpha})=\lambda',
\qquad \forall\,\alpha_i\in\A_i.
\]
Subtracting gives $y_{i\alpha}'-y_{i\alpha}=\lambda'-\lambda$ for all $\alpha_i$, hence
$y_i'-y_i=c\one$ with $c:=\lambda'-\lambda$. Conversely, if $y_i'=y_i+c\one$, then $Q_i(y_i')=Q_i(y_i)$
by~\cref{prop:choice-map}(\labelcref{it:shift}).
\end{proof}

\subsection{Bregman divergence and Fenchel coupling}

Fix $p_i\in\X_i$ and define the \emph{face-neighborhood} of $p_i$ by
\[
\Delta_{p_i}\;:=\;\{x_i\in\X_i:\supp(x_i)\supseteq\supp(p_i)\}.
\]
For $x_i\in\Delta_{p_i}$, the segment $x_i+t(p_i-x_i)$ stays in $\X_i$ for all sufficiently small $t\downarrow 0$.
We therefore define the (one-sided) directional derivative
\[
h_i'(x_i;z_i)\;:=\;\lim_{t\downarrow 0}\frac{h_i(x_i+t z_i)-h_i(x_i)}{t},
\qquad z_i\in \YY_i \text{ with } x_i+t z_i\in\X_i \text{ for small } t>0,
\]
and the (extended) \emph{Bregman divergence} by
\begin{equation}\label{eq:Breg-ext}
D_i(p_i,x_i)\;:=\;h_i(p_i)-h_i(x_i)-h_i'(x_i;p_i-x_i),
\qquad x_i\in\Delta_{p_i}.
\end{equation}
If $x_i\in\X_i^\circ$ then $h_i$ is differentiable at $x_i$ and $h_i'(x_i;p_i-x_i)=\langle\nabla h_i(x_i),p_i-x_i\rangle$,
so \eqref{eq:Breg-ext} reduces to the usual Bregman divergence.

The associated \emph{Fenchel coupling} is
\[
F_i(p_i,y_i)\;:=\;h_i(p_i)+h_i^*(y_i)-\langle y_i,p_i\rangle,\qquad (p_i,y_i)\in\X_i\times\YY_i,
\]
and for profiles $p\in\X$, $y\in\YY$, we set $F_h(p,y):=\sum_{i\in\N}F_i(p_i,y_i)$.

\begin{proposition}\label{prop:fenchel}
Assume $h_i$ is $K_i$-strongly convex on $\X_i$ w.r.t.\ $\|\cdot\|$.
Then, for all $p_i\in\X_i$ and $y_i\in\YY_i$:
\begin{enumerate}
\item\label{it:fenchel-lb} $F_i(p_i,y_i)\ge \dfrac{K_i}{2}\|Q_i(y_i)-p_i\|^2$.
\item\label{it:fenchel-zero} $F_i(p_i,y_i^k)\to 0$ if and only if $Q_i(y_i^k)\to p_i$.
\item\label{it:fenchel-breg} If $x_i=Q_i(y_i)\in\Delta_{p_i}$, then $F_i(p_i,y_i)=D_i(p_i,x_i)$.
\end{enumerate}
\end{proposition}
\begin{proof}
See \cite{MS16}.
\end{proof}

In particular, $F_i(p_i,y_i)=0$ if and only if $p_i=Q_i(y_i)$. Thus the Fenchel coupling acts as a kind of primal-dual distance between strategies and scores.

\begin{lemma}\label{lem:fenchel-deriv}
Let $x=Q(y)$ be a solution orbit of \eqref{eq:FTRL}. Then for every fixed $p_i\in\X_i$,
\[
\frac{d}{dt}F_i(p_i,y_i) =\langle v_i(x),\,x_i-p_i\rangle.
\]
\end{lemma}
\begin{proof}
By \cref{prop:choice-map}(\labelcref{it:Qgrad}), $\nabla h_i^*=Q_i$. Hence, by the chain rule,
\[
\frac{d}{dt}F_i(p_i,y_i)
=
\big\langle \nabla h_i^*(y_i),\dot y_i\big\rangle-\langle \dot y_i(t),p_i\rangle
=
\big\langle \dot y_i, Q_i(y_i)-p_i\big\rangle.
\]
Using $Q_i(y_i)=x_i$ and $\dot y_i=v_i(x)$ yields the claim.
\end{proof}

\subsection{Examples of regularizers}\label{subsec:examples}

We give here some examples of standard regularizers \cite{BecTeb03, MS16, BM23, SS11} which verify our assumptions. Recall $s_i(p)\coloneqq 1/\theta_i''(p)$.

\begin{example}[Entropy]
    Let $\theta_i(p)=p\log p$ for $p \in [0,1]$ (with the convention $0\log 0=0$). Then
\[
\theta_i''(p)=\frac{1}{p}\ge 1 \quad (p\in(0,1]),
\]
so \eqref{eq:curvature} holds with $K_i=1$ and $s_i(p)=p$.
The regularizer is steep and the choice map is the softmax
\[
Q_i(y_i)=\Lambda_i(y_i)
:=\left(\frac{\exp(y_{i\alpha})}{\sum_{\beta_i\in\A_i}\exp(y_{i\beta})}\right)_{\alpha_i\in\A_i},
\qquad
h_i^*(y_i)=\log\Big(\sum_{\beta_i\in\A_i}\exp(y_{i\beta})\Big).
\]
On $\Delta_{p_i}$, the Bregman divergence is the \emph{Kullback--Leibler divergence}, or \emph{relative entropy}:
\[
D_i(p_i,x_i)=\sum_{\alpha_i\in\A_i} p_{i\alpha}\log\frac{p_{i\alpha}}{x_{i\alpha}}
=D_{\mathrm{KL}}(p_i\|x_i).
\]
\end{example}

\begin{example}[Euclidean]
Let $\theta_i(p)=\tfrac12 p^2$. Then $\theta_i''(p)=1$, so \eqref{eq:curvature} holds with $K_i=1$ and $s_i(p)=1.$
The regularizer is non-steep and
\[
Q_i(y_i)=\Pi_i(y_i):=\arg\min_{x_i\in\X_i}\|x_i-y_i\|_2^2.
\]
Moreover,
\[
h_i^*(y_i)
=\sup_{x_i\in\X_i}\left\{\langle y_i,x_i\rangle-\frac12\|x_i\|_2^2\right\}
=\frac12\|y_i\|_2^2-\frac12\,\dist(y_i,\X_i)^2,
\]
and on $\X_i^\circ$ one has $D_i(p_i,x_i)=\tfrac12\|p_i-x_i\|_2^2$.
\end{example}

\begin{example}[Square-root]
    Let $\theta_i(p)=-4\sqrt{p}$ on $[0,1]$. Then
\[
\theta_i''(p)=p^{-3/2}\ge 1,
\]
so \eqref{eq:curvature} holds with $K_i=1$ and $s_i(p)=p^{3/2}.$
The regularizer is steep, and on $\Delta_{p_i}$ the Bregman divergence is
\[
D_i(p_i,x_i)=2\sum_{\alpha_i\in\A_i}\left(\frac{(\sqrt{p_{i\alpha}}-\sqrt{x_{i\alpha}})^2}{\sqrt{x_{i\alpha}}}\right).
\]
\end{example}

\begin{example}[Tsallis entropy]
    Fix $\lambda\in(0,1)\cup(1,2]$ and let
\[
\theta_i(p)=\frac{p^\lambda}{\lambda(\lambda-1)},\qquad p\in[0,1].
\]
Then, since $\lambda-2\le 0$ we have
\[
\theta_i''(p)=p^{\lambda-2}\ge 1 \quad (p\in(0,1]),
\]
so \eqref{eq:curvature} holds with $K_i=1$ and $s_i(p)=p^{2-\lambda}.$
Moreover, if $\lambda\in(0,1)$ then $\theta_i'(p)\to-\infty$ as $p\downarrow 0$, so it is steep, while if $\lambda\in(1,2]$ then $\theta_i'(0+)=0$, so it is non-steep.
On $\Delta_{p_i}$,
\[
D_i(p_i,x_i)
=
\frac{1}{\lambda(\lambda-1)}
\sum_{\alpha_i\in\A_i}
\Big(
p_{i\alpha}^{\lambda}
-\lambda x_{i\alpha}^{\lambda-1}p_{i\alpha}
+(\lambda-1)x_{i\alpha}^{\lambda}
\Big).
\]
This family contains the Euclidean regularizer as the case $\lambda=2$ and the square-root regularizer as the case $\lambda=1/2$, while the entropy regularizer is recovered in the limit $\lambda\to1$, up to affine terms on the simplex.
\end{example}

\section{Strategy dynamics}
\label{app:dynamics}

Consider continuous-time \ac{FTRL}
\[
\dot y=v(x),\qquad x=Q(y).
\]
We start first with well-posedness in \emph{score space}:
\begin{proposition}\label{prop:well-posed}
    $v\circ Q$ is Lipschitz continuous and the ODE $\dot y=v(Q(y))$ admits a unique global solution $y(\cdot):\mathbb R\to \YY$ from every initial condition $y(0)$.
\end{proposition}
\begin{proof}
    $v$ is multilinear by definition and bounded by finiteness of the game, $Q$ is Lipschitz continuous by \cref{prop:choice-map}(\labelcref{it:Qgrad}), hence their composition is Lipschitz continuous. Well-posedness then follows from standard Cauchy--Lipschitz / Picard--Lindel\"of theorem.
\end{proof}

We would like now to establish the existence of a flow in \emph{strategy space}. For this purpose, fix $\B_i\subseteq\A_i$ and assume that on some interval $I$ the support of player $i$ is constant:
\[
x_i(t)\in\Delta(\B_i)^\circ \qquad \forall\,t\in I.
\]
Define
\[
s_i(p)\coloneqq \frac{1}{\theta_i''(p)},\quad
S_i(x_i)\coloneqq \sum_{\beta_i\in \B_i} s_i(x_{i\beta}), \quad \pi_{i\alpha}(x_i)
	\defeq \begin{dcases*}
		\frac{s_i(x_{i\alpha})}{S_i(x_i)}
			&if $\alpha_i \in \B_i$,
			\\
		0
			&otherwise.
		\end{dcases*}
\]
and write $\pi_i(x_i)=(\pi_{i\beta}(x_i))_{\beta_i\in\A_i}$.

\begin{proposition}\label{prop:strategy-ode}
On $I$, for every $\alpha_i\in \B_i$,
\[
\dot x_{i\alpha}
=
s_i(x_{i\alpha})\Big(v_{i\alpha}(x)-\langle \pi_i(x_i),v_i(x)\rangle\Big).
\]
\end{proposition}

\begin{proof}
Following \cite{MS16} and \cite{FVGL+20}, on $I$, we have $x_{i\alpha}(t)>0$ for $\alpha_i\in\B_i$, so $\theta_i$ is $C^2$ along these coordinates.
The KKT conditions for the constrained maximization defining $x_i=Q_i(y_i)$ on the face $\Delta(\B_i)$ yield
\[
y_{i\alpha}-\theta_i'(x_{i\alpha})=\lambda_i,
\qquad \alpha_i\in\B_i,
\]
for some scalar multiplier $\lambda_i=\lambda_i(t)$ enforcing $\sum_{\alpha_i\in\B_i}x_{i\alpha}=1$.
Differentiating in time and using $\dot y_{i\alpha}=v_{i\alpha}(x)$ gives
\[
v_{i\alpha}(x)-\theta_i''(x_{i\alpha})\,\dot x_{i\alpha}=\dot\lambda_i,
\]
hence
\[
\dot x_{i\alpha}=s_i(x_{i\alpha})\big(v_{i\alpha}(x)-\dot\lambda_i\big).
\]
Summing over $\alpha_i\in\B_i$ and using $\sum_{\alpha_i\in\B_i}\dot x_{i\alpha}=0$ yields
\[
0=\sum_{\alpha_i\in\B_i}s_i(x_{i\alpha})\big(v_{i\alpha}(x)-\dot\lambda_i\big)
\quad\Longrightarrow\quad
\dot\lambda_i=\sum_{\alpha_i\in\B_i}\pi_{i\alpha}(x_i)\,v_{i\alpha}(x)
=\langle \pi_i(x_i),v_i(x)\rangle,
\]
which gives the stated formula.
\end{proof}

\begin{remark}\label{rmrk:steep-ode}
    In the steep case, the choice map takes values in the relative interior, so the above facewise ODE applies globally
along each trajectory (with $\B_i=\A_i$).
\endenv
\end{remark}

\subsection{Construction of the strategy flow}

We now encode the induced dynamics directly on $\X$ by a globally defined vector field.
This requires a regularity strengthening on $s_i$, which is the following Assumption.

\begin{assumption}\label{asm:well-posed}
For each $i$, $s_i$ extends to a globally Lipschitz function on $[0,1]$ with $s_i(0)=0$
\end{assumption}

Notice that all the steep regularizers in \cref{subsec:examples} verify this assumption. Note also that this assumption implies steepness:
\begin{proposition}\label{prop:wp-steep}
Under \cref{asm:well-posed}, $h_i$ is steep.
\end{proposition}

\begin{proof}
Let $L>0$ be a Lipschitz constant of $s_i$ on $[0,1]$. Since $s_i(0)=0$, we have
$s_i(p)=|s_i(p)-s_i(0)|\le Lp$ for all $p\in[0,1]$.
Thus, for all $p\in(0,1]$,
\[
\theta_i''(p)=\frac{1}{s_i(p)}\ge \frac{1}{Lp}.
\]
Fix $p_0\in(0,1]$. Since $\theta_i\in C^2((0,1])$, $\theta_i'$ is absolutely continuous on $[p,p_0]$ for any $p\in(0,p_0)$, and
\[
\theta_i'(p)=\theta_i'(p_0)-\int_p^{p_0}\theta_i''(t)\,dt
\le \theta_i'(p_0)-\frac{1}{L}\int_p^{p_0}\frac{dt}{t}
= \theta_i'(p_0)-\frac{1}{L}\log\frac{p_0}{p}.
\]
Letting $p\downarrow 0$ yields $\theta_i'(p)\to -\infty$, i.e.\ $h_i$ is steep.
\end{proof}

We now state a quick consequence of this, which will be useful for the contruction of the global flow.

\begin{lemma}\label{lem:S-positive}
Under \cref{asm:well-posed}, for each $i$ there exists $\sigma_i>0$ such that $S_i(x_i)\ge \sigma_i$ for all $x_i\in\X_i$.
Consequently $\pi_i:\X_i\to\X_i$ is Lipschitz continuous.
\end{lemma}

\begin{proof}
Let $m_i=|\A_i|$. For any $x_i\in\X_i$ there exists $\alpha$ with $x_{i\alpha}\ge 1/m_i$.
Since $s_i$ is continuous on $[0,1]$ and strictly positive on $(0,1]$, the minimum
\[
\sigma_i\coloneqq \min_{p\in[1/m_i,1]} s_i(p)
\]
exists and satisfies $\sigma_i>0$. Hence $S_i(x_i)\ge s_i(x_{i\alpha})\ge \sigma_i$.
Lipschitz continuity of $\pi_i$ follows because it is a quotient of Lipschitz continuous functions with denominator bounded away from $0$.
\end{proof}

\begin{definition}\label{def:field}
Define the \emph{strategy field} $F:\X\to \YY$ by, for each $i\in \N$ and $\alpha_i\in\A_i$,
\[
F_{i\alpha}(x)\coloneqq s_i(x_{i\alpha})\Big(v_{i\alpha}(x)-\langle \pi_i(x_i),v_i(x)\rangle\Big).
\]
\end{definition}

\begin{proposition}\label{prop:flow}
Under \cref{asm:well-posed} the ODE
\[
\dot x=F(x),\qquad x(0)\in \X,
\]
admits a unique global solution $x(\cdot):\mathbb R\to \X$.
This yields a continuous flow $(\Theta_t)_{t\in\mathbb R}$ on $\X$\textemdash the \emph{strategy flow}\textemdash which is invariant on each face of $\X$.
\end{proposition}
\begin{proof}
\Cref{asm:well-posed} implies each $s_i$ is globally Lipschitz on $[0,1]$ and, by~\cref{lem:S-positive}, the map $x_i\mapsto \pi_i(x_i)$ is also globally Lipschitz on $\X_i$.
Moreover, by multilinearity, $v$ is globally Lipschitz on $\X$.
Therefore each component $F_{i\alpha}(x)$ is a composition of globally Lipschitz maps, hence $F$ is globally Lipschitz on $\X$.
Existence and uniqueness of global solutions for all $t\in\mathbb R$ follow from the Cauchy--Lipschitz / Picard--Lindel\"of theorem. Moreover, for each $i$, using $s_i(x_{i\alpha})=S_i(x_i)\pi_{i\alpha}(x_i)$ we obtain
\[
\sum_{\alpha_i}\dot x_{i\alpha}= \sum_{\alpha_i}F_{i\alpha}(x)
=
\sum_{\alpha_i} s_i(x_{i\alpha})v_{i\alpha}(x)
-
\Big(\sum_{\alpha_i}s_i(x_{i\alpha})\Big)\langle \pi_i(x_i),v_i(x)\rangle
=0.
\]
This shows $\sum_{\alpha_i}\dot x_{i\alpha}=0$, hence $\sum_{\alpha_i}x_{i\alpha}(t)=1$ for all $t$.
Moreover, if $x_{i\alpha_i}(t)=0$, then $s_i(0)=0$ implies $\dot x_{i\alpha_i}(t)=0$, so zero coordinates remain zero.
Consequently, for every initial condition $x_0\in\X$ and every $t\ge 0$,
\[
\supp(\Theta_t(x_0))\subseteq \supp(x_0),
\]
i.e.\ each face of $\X$ is forward-invariant. To upgrade this to invariance (two-sided), fix a face $\F=\prod_{i\in\N}\Delta(\B_i)$ and note that
\[
x\in\F \quad\Longleftrightarrow\quad x_{i \alpha}=0\ \ \text{for all }i\in\N,\ \alpha_i\in \A_i\setminus \B_i.
\]
By the previous paragraph, if $x\in\F$ then $\Theta_t(x)\in\F$ for all $t\ge 0$. Consider now the time-reversed ODE
\[
\dot x=-F(x).
\]
Its vector field is again globally Lipschitz. Moreover, the same argument as above applies verbatim: for each $i$,
\[
\sum_{\alpha_i}(-F_{i\alpha}(x))=0,
\]
and if $x_{i\alpha}=0$, then $s_i(0)=0$ implies
\[
-F_{i\alpha}(x)=0.
\]
Hence every face of $\X$ is forward-invariant under the time-reversed ODE as well. By uniqueness of solutions and the flow property, the flow of $\dot x=-F(x)$ is $(\Theta_{-t})_{t\in\R}$. Therefore, for every $t\ge 0$,
\[
\Theta_{-t}(\F)\subseteq \F.
\]
Applying $\Theta_t$ to both sides and using $\Theta_t\circ\Theta_{-t}=\Theta_0=\mathrm{Id}$ yields
\[
\F\subseteq \Theta_t(\F).
\]
Since we already have $\Theta_t(\F)\subseteq \F$ for $t\ge 0$, it follows that
\[
\Theta_t(\F)=\F
\]
for all $t\ge 0$. Replacing $t$ by $-t$ gives $\Theta_t(\F)=\F$ for all $t\in\R$.
\end{proof}

\begin{remark}\label{rmrk:reversal}
Since $F$ is globally Lipschitz, the time-reversed flow is obtained by solving $\dot x=-F(x)$, i.e.\ $\Theta_{-t}$
is the flow map of $-F$.
Moreover, if one considers the negated game with payoffs $-u_i$ (so the payoff field becomes $-v$),
then the induced strategy field is exactly $-F$ (replace $v$ by $-v$ in~\cref{def:field}).
Hence $(\Theta_{-t})_{t\in\mathbb R}$ coincides with the strategy flow of the negated game.
\endenv
\end{remark}

Naturally, when $h$ is steep, we shall always work under \cref{asm:well-posed} whenever we mention the strategy flow induced by \eqref{eq:SD}.

\begin{remark}\label{rmrk:flow-link}
    Consider a subgame $\B=\prod_{i\in\N}\B_i$ and its face $\F=\cont(\B)$, let $Q_\B:\YY\to\F^\circ$ be the face-restricted choice map  $Q_\B(y):=\arg \max_{w \in \F} \{ \langle y, w\rangle -h(w) \}$, and consider the facewise \ac{FTRL} dynamics
\begin{equation}
\label{eq:FTRL-B}
\dot y = v(x),
\qquad x = Q_\B(y).
\tag{FTRL-$\protect\FTRLtag$}
\end{equation}
Now, given $x_0 \in \F$, the strategy flow trajectory $(\Theta_t)$ induced by \eqref{eq:SD} with $\Theta_0 = x_0$ coincides, by construction, with the \eqref{eq:FTRL-B} orbit $x(\cdot)$ with $Q_\B(y(0)) = x_0$ (which indeed exists, by \cref{lem:steep-image}, and is unique by \cref{prop:Q-injective}). Hence the strategy flow can be seen as a globally Lipschitz continuous patchwork of these facewise \ac{FTRL} dynamics.
\endenv
\end{remark}

\section{Omitted proofs from Section \ref{sec:dynstable}}
\label{app:dynstable}

We start by stating two properties of \ac{FTRL} which will be useful for this section. We start with a powerful rationality property, proved in \cite{MS16}:
\begin{theorem}\label{thm:extinction}
Under \eqref{eq:FTRL}, any dominated mixed strategy becomes extinct along every orbit \(x(t)=Q(y(t))\).
\end{theorem}

Where we say that $x_i\in\X_i$ is \emph{dominated} by $x'_i\in\X_i$ if $x'_i$ yields a strictly higher payoff against every opponents' profile:
$u_i(x'_i,w_{-i})>u_i(x_i,w_{-i})$ for all $w_{-i}\in\X_{-i}$.

We will also use the following fact, proved by \citet{FVGL+20}.
\begin{theorem}\label{thm:no-interior} Let \(A\subseteq \mathcal{X}^\circ\). Then \(A\) is not asymptotically stable under \eqref{eq:FTRL}. Consequently, if $h$ is steep, no subset of the relative interior of any face is asymptotically stable under the strategy flow induced by \eqref{eq:SD}. In particular, every asymptotically stable set intersects the vertex set \(\mathcal{A}\). \end{theorem}

\subsection{Stability and attraction}

Recall first that we identify each pure action $\alpha_i\in\A_i$ with the corresponding vertex of the simplex $\X_i$
(and likewise each pure profile $\alpha\in\A$ with the corresponding vertex of $\X$).

The next argument hinges on two structural properties of the strategy flow $(\Theta_t)$: \emph{face-invariance} and \emph{elimination of dominated strategies}. In fact, the same reasoning applies verbatim to any flow on $\X$ that satisfies these two properties. Indeed,  they imply together that, on each one-dimensional face (edge) of the strategy space, the induced motion is aligned with player deviations. More precisely, trajectories are confined to the edge by face-invariance while elimination of dominated strategies forces them to drift toward the endpoint corresponding to the profitable unilateral deviation. In this sense, the flow on edges follows the orientation prescribed by the preference graph.

\begin{proof}[\textbf{Proof of \cref{prop:attract-club}.}]
    Suppose, toward a contradiction, that $\alpha\in S$ and $\alpha\to\alpha'$ in the preference graph, but $\alpha'\notin S$. There exists then a unique player $i$ such that $\alpha$ and $\alpha'$ are $i$--comparable and $u_i(\alpha')\geq u_i(\alpha)$. Let $U$ be an attracting neighborhood of $S$. Because $\alpha\in S\subset U$ and $\alpha'\notin S$, shrinking $U$ if necessary we may assume $\alpha'\notin U$. By connectedness of the segment $[\alpha,\alpha']$, there exists $\beta\in(\alpha,\alpha')\cap U$. By face invariance under steep regularizers (\cref{prop:flow}), $\Theta_t(\beta)\in[\alpha,\alpha']$ for all $t\ge0$.

    If $u_i(\alpha')>u_i(\alpha)$, on the segment $[\alpha,\alpha']$, player $i$'s action $\alpha'_i$ strictly dominates $\alpha_i$, so by elimination of dominated strategies applied to this one--player game (see \cref{thm:extinction}) we have $\Theta_t(\beta)\longrightarrow \alpha'$ as $t\to\infty$.
Since $\beta\in U$ and $U$ is an attracting neighborhood of $S$, every limit point of $(\Theta_t)$ starting from $\beta$ lies in $S$, hence $\alpha'\in S$, a contradiction.

If $u_i(\alpha')=u_i(\alpha)$, then every point of $[\alpha,\alpha']$ is stationary by face-invariance and  the fact that a game with constant payoffs admits constant dynamics under \eqref{eq:FTRL}. Define now $\alpha(\lambda):=(1-\lambda)\alpha+\lambda\alpha'$, $\lambda\in[0,1]$, and set
\[
\Lambda:=\{\lambda\in[0,1]:\ \alpha(\lambda)\in S\}.
\]
Then $0\in\Lambda$ and $\Lambda$ is closed (since $S$ is closed). Let $\lambda^\ast=\sup\Lambda$. If $\lambda^\ast<1$, choose $\lambda\in(\lambda^\ast,1)$ sufficiently close to $\lambda^\ast$ so that $\alpha(\lambda)\in U$. But $\alpha(\lambda)$ is stationary, and by attraction, this forces $\alpha(\lambda)\in S$, contradicting the choice of $\lambda^\ast$. Hence $\lambda^\ast=1$ and $\alpha'=\alpha(1)\in S$, again a contradiction.
\end{proof}

With this, we obtain the following useful corollary:

\begin{corollary}\label{cor:stable-club}
    Every asymptotically stable set for the strategy flow induced by \eqref{eq:SD} contains a set of pure profiles which is closed under better replies.
\end{corollary}

\begin{proof}
    By \cref{thm:no-interior}, every asymptotically stable set $S$ must intersect $\A$. \Cref{prop:attract-club} then forces $S \cap \A$ to be closed under better replies.
\end{proof}

To prove the next results, we will need the following lemma:

\begin{lemma}\label{lem:gap-freezing}
Let $i\in\N$ and $\alpha_i \in\A_i$. For $G\ge 0$ define the gap region
\[
\Gamma_i(\alpha_i,G):=\{y_i\in\YY_i:\ y_{i\alpha}\ge y_{i\beta}+G\ \ \forall \beta\neq \alpha\},
\]
then, for every $r>0$ ,there exists a threshold $G_i(\alpha_i,r)>0$ such that
for every $y_i\in\Gamma_i(\alpha_i,G_i(\alpha_i,r))$ one has $\|Q_i(y_i)-\alpha_i\|_1\le r$.
\end{lemma}

\begin{proof}
Fix $r>0$ and assume, toward a contradiction, that no such $G_i(\alpha_i,r)$ exists.
Then for each integer $m\ge 1$ we can find $y_i^{m}\in\Gamma_i(\alpha_i,m)$ such that
$\|Q_i(y_i^{m})-\alpha_i\|_1>r$.
By \cref{prop:choice-map}(\labelcref{it:shift}),
replacing $y_i^{m}$ by $z_i^{m}:=y_i^{m}-y_{i\alpha}^{m}\mathbf 1$ does not change $Q_i$ and yields
$z_{i\alpha}^{m}=0$ and $z_{i\beta}^{m}\le -m$ for all $\beta_i\neq \alpha_i$.
Hence $z_{i\beta}^{m}-z_{i\alpha}^{m}\to -\infty$ for every $\beta_i\neq \alpha_i$.
By \cref{prop:choice-map}(\labelcref{it:vanish}),
this implies $Q_{i\beta}(z_i^{m})\to 0$ for every $\beta_i\neq \alpha_i$, and therefore
$Q_i(z_i^{m})\to \alpha_i$ in $\|\cdot\|_1$, contradicting $\|Q_i(y_i^{m})-\alpha_i\|_1>r$.
\end{proof}

In the next proof, due to the lack of face-invariance, we instead initialize the score dynamics so that not only play starts
arbitrarily close to $\alpha$, but also all opponents are forced to remain close to
$\alpha_{-i}$ over a prescribed time window by imposing large initial score gaps.
During that window, the deviator's score difference between $\alpha_i$ and the strict
better reply $\alpha_i'$ drifts at a uniform negative rate, and we pick the window
length so that the choice map flips (up to a prescribed tolerance) toward $\alpha_i'$,
forcing the orbit to exit a stable neighborhood.

\begin{proof}[\textbf{Proof of \cref{thm:stable-sclub}.}]
Let $S\subseteq\X$ be nonempty, closed, and stable under \eqref{eq:FTRL}, i.e.
\[
x(0)\in V\cap\ImQ \ \Longrightarrow\ x(t)\in U\quad\forall t\ge 0
\]
for every neighborhood $U$ of $S$ and some neighborhood $V$ of $S$. Assume toward a contradiction that $\alpha\in S\cap\A$ and $\alpha\to\alpha'$ is a strict better
reply but $\alpha'\notin S$. Let $i$ be the deviating player, so $\alpha'=(\alpha_i',\alpha_{-i})$, and set
\[
\delta:=u_i(\alpha_i',\alpha_{-i})-u_i(\alpha_i,\alpha_{-i})>0.
\]
Since $S$ is closed and $\alpha'\notin S$, let $r:=\dist(\alpha',S)>0$ and define the open neighborhood
\[
U:=\{x\in\X:\dist(x,S)<r/2\},
\]
so $\alpha'\notin U$. By stability, there exists an open neighborhood $V\supseteq S$ such that every
solution orbit $x=Q(y)$ with $x(0)\in V\cap\ImQ$ satisfies $x(t)\in U$ for all $t\ge 0$. Because $\alpha\in S\subseteq V$ and $V$ is open, pick $r_0\in(0,r/8)$ such that
$B_1(\alpha,r_0)\subseteq V$. Define the continuous function
\[
g(x_{-i}):=u_i(\alpha_i',x_{-i})-u_i(\alpha_i,x_{-i}),\qquad x_{-i}\in\X_{-i}.
\]
Since $g(\alpha_{-i})=\delta$, there exists $\kappa_0>0$ such that
\begin{equation}\label{eq:margin-ct}
\|x_{-i}-\alpha_{-i}\|_1\le \kappa_0\quad\Longrightarrow\quad g(x_{-i})\ge \delta/2.
\end{equation}
Set $\kappa:=\min\{\kappa_0,r/8,r_0/4\}$. If $|\N|\ge 2$, choose $\rho:=\min\{\kappa/2,r_0/2\}$ and
numbers $(\rho_j)_{j\neq i}$ with $\rho_j>0$ and $\sum_{j\neq i}\rho_j=\rho$. If $|\N|=1$ set $\rho:=0$
and ignore all $j\neq i$ below. Next, we invoke \cref{lem:gap-freezing}: for each $j\neq i$ choose $G_j>0$ such that
\begin{equation}\label{eq:Gj-ct}
y_j\in\Gamma_j(\alpha_j,G_j)\quad\Longrightarrow\quad \|Q_j(y_j)-\alpha_j\|_1\le \rho_j.
\end{equation}
Likewise, apply the same property twice for player $i$ to obtain $B>0$ and $G_i'>0$ such that
\begin{align}
y_i\in\Gamma_i(\alpha_i,B)\ &\Longrightarrow\ \|Q_i(y_i)-\alpha_i\|_1\le r_0/2, \label{eq:B-ct}\\
y_i\in\Gamma_i(\alpha_i',G_i')\ &\Longrightarrow\ \|Q_i(y_i)-\alpha_i'\|_1\le r/8. \label{eq:Gi-ct}
\end{align}
Fix now the horizon
\[
T_*:=\frac{2(B+G_i')}{\delta},
\qquad\text{and set }T:=T_*.
\]
For each player $j\neq i$ define the payoff spread
\[
M_j:=\max_{x\in\X}\ \max_{\beta_j,\gamma_j\in\A_j}\ |v_{j\beta}(x)-v_{j\gamma}(x)|<\infty.
\]
Initialize the opponents' scores by
\[
y_{j\alpha}(0)=G_j+M_jT,\qquad y_{j\beta}(0)=0,\quad\beta_j\neq \alpha_j.
\]
For each $\beta_j\neq \alpha_j$ consider the gap $z_{j\beta}(t):=y_{j\alpha}(t)-y_{j\beta}(t)$.
Then
\[
\dot z_{j\beta}(t)=v_{j\alpha}(x(t))-v_{j\beta}(x(t))\ge -M_j,
\]
so for $t\in[0,T]$,
\[
z_{j\beta}(t)\ge z_{j\beta}(0)-M_j t\ge (G_j+M_jT)-M_jT=G_j.
\]
Thus $y_j(t)\in\Gamma_j(\alpha_j,G_j)$ for all $t\in[0,T]$, and \eqref{eq:Gj-ct} yields
\begin{equation}\label{eq:opponents-close}
\|x_j(t)-\alpha_j\|_1=\|Q_j(y_j(t))-\alpha_j\|_1\le \rho_j,\qquad \forall t\in[0,T],\ \forall j\neq i.
\end{equation}
In particular,
\begin{equation}\label{eq:opp-margin}
\|x_{-i}(t)-\alpha_{-i}\|_1\le \sum_{j\neq i}\rho_j=\rho\le \kappa\le \kappa_0,
\qquad \forall t\in[0,T].
\end{equation}
By \eqref{eq:margin-ct} and \eqref{eq:opp-margin}, for $t\in[0,T]$ we have
\[
u_i(\alpha_i',x_{-i}(t))-u_i(\alpha_i,x_{-i}(t)) = g(x_{-i}(t)) \ge \delta/2.
\]
Now initialize player $i$'s scores by
\[
y_{i\alpha}(0)=B,\qquad y_{i\alpha'}(0)=0,\qquad y_{i\gamma}(0)=-R\quad(\gamma_i\notin\{\alpha_i,\alpha_i'\}),
\]
where $R$ will be chosen momentarily.
Then $y_i(0)\in\Gamma_i(\alpha_i,B)$, so by \eqref{eq:B-ct} we have
$\|x_i(0)-\alpha_i\|_1\le r_0/2$.
Combining with \eqref{eq:opponents-close} at $t=0$ gives
\[
\|x(0)-\alpha\|_1
\le \|x_i(0)-\alpha_i\|_1+\sum_{j\neq i}\|x_j(0)-\alpha_j\|_1
\le r_0/2+\rho
\le r_0,
\]
hence $x(0)\in B_1(\alpha,r_0)\subseteq V$. Since also $x(0)=Q(y(0))\in\ImQ$, stability implies $x(t)\in U$
for all $t\ge 0$.

Define the deviator's score difference $z(t):=y_{i\alpha}(t)-y_{i\alpha'}(t)$. For $t\in[0,T]$,
\[
\dot z(t)=v_{i\alpha}(x(t))-v_{i\alpha'}(x(t))
=u_i(\alpha_i,x_{-i}(t))-u_i(\alpha_i',x_{-i}(t))
\le -\delta/2,
\]
so
\[
z(T)\le z(0)-(\delta/2)T = B-(\delta/2)\cdot \frac{2(B+G_i')}{\delta}=-G_i',
\]
i.e. $y_{i\alpha'}(T)\ge y_{i\alpha}(T)+G_i'$.

It remains to ensure that $\alpha_i'$ dominates every other action at time $T$ by margin $G_i'$.
Let
\[
M_i:=\max_{x\in\X}\ \max_{\beta_i,\gamma_i\in\A_i}\ |v_{i\beta}(x)-v_{i\gamma}(x)|,
\]
and for each $\gamma_i\notin\{\alpha_i,\alpha_i'\}$ define $w_\gamma(t):=y_{i\alpha'}(t)-y_{i\gamma}(t)$.
Then $\dot w_\gamma(t)=v_{i\alpha'}(x(t))-v_{i\gamma}(x(t))\ge -M_i$, so for $t\in[0,T]$,
\[
w_\gamma(t)\ge w_\gamma(0)-M_i t=R-M_i t.
\]
Choose $R\ge G_i'+M_iT$. Then $w_\gamma(T)\ge G_i'$ for all $\gamma_i\notin\{\alpha_i,\alpha_i'\}$, and we already have
$y_{i\alpha'}(T)\ge y_{i\alpha}(T)+G_i'$. Hence $y_i(T)\in\Gamma_i(\alpha_i',G_i')$, and \eqref{eq:Gi-ct} yields
\begin{equation}\label{eq:i-close}
\|x_i(T)-\alpha_i'\|_1\le r/8.
\end{equation}
Also, by \eqref{eq:opponents-close} at $t=T$,
\[
\|x_{-i}(T)-\alpha_{-i}\|_1\le \sum_{j\neq i}\rho_j=\rho\le r/8.
\]
Therefore,
\[
\|x(T)-\alpha'\|_1 \le \|x_i(T)-\alpha_i'\|_1+\|x_{-i}(T)-\alpha_{-i}\|_1 \le r/8+r/8=r/4.
\]
Since $x\mapsto \dist(x,S)$ is $1$-Lipschitz in $\|\cdot\|_1$, we have
\[
\dist(x(T),S)\ \ge\ \dist(\alpha',S)-\|x(T)-\alpha'\|_1\ \ge\ r-r/4\ =\ 3r/4\ >\ r/2,
\]
so $x(T)\notin U$. This contradicts stability.
\end{proof}

\subsection{Connectedness}

The next lemma works by showing that the span of a strongly connected set cannot contain a proper attractor. Indeed, we proceed by contradiction, showing first that such an attractor would have to contain all vertices of the set, and then, by induction, show how that would force the existence of a fully interior repellor, which is impossible.

The next theorem shows that strong connectivity of pure profiles propagates to dynamical connectedness of their mixed span.

\begin{lemma}\label{lem:ict-span}
If $\HH \subseteq \A$ is strongly connected, then $\cont(\HH)$ is internally chain transitive for the strategy flow induced by \eqref{eq:SD}.
\end{lemma}

\begin{proof}
Let $\HH\subseteq\A$ be strongly connected.
By \cref{prop:ict-attractor}, it suffices to show that $\cont(\HH)$ is invariant and that the flow restricted to $\cont(\HH)$ admits no proper attractor.

By \cref{prop:flow}, every face of $\X$ is invariant under the strategy flow. Since $\cont(\HH)$ is the union of all faces whose vertices lie in $\HH$, it follows that $\cont(\HH)$ is invariant. Let $A\subseteq \cont(\HH)$ be an attractor for the restricted flow on $\cont(\HH)$.
We show that $A=\cont(\HH)$.

Pick a face $\F\subseteq \cont(\HH)$ with $A\cap \F\neq\varnothing$.
Because $\F$ is invariant, \cref{lem:restrict-attract} implies that $A\cap \F$ is an attractor for the flow restricted to $\F$.
Hence $A\cap \F$ is asymptotically stable for the dynamics on $\F$, so by \cref{thm:no-interior} it intersects the vertex set of $\F$.
Choose $\alpha\in (A\cap \F)\cap\A$. Since $\alpha\in \cont(\HH)$ and $\alpha$ is a vertex, we have $\alpha\in\HH$.

We claim that $\HH\subseteq A$.
If not, pick $\alpha'\in \HH\setminus A$. By strong connectivity of $\HH$, there exists a directed path
\[
\alpha=\alpha^0\to\alpha^1\to\cdots\to\alpha^k=\alpha'
\quad\text{with all }\alpha^j\in\HH.
\]
Let $j$ be the smallest index such that $\alpha^j\in A$ and $\alpha^{j+1}\notin A$, and set $E:=[\alpha^j,\alpha^{j+1}]$.
By \cref{prop:flow}, $E$ is invariant, and since $\alpha^j\in A\cap E$, \cref{lem:restrict-attract} shows that $A\cap E$ is an attractor for the restricted flow on $E$.
Applying \cref{prop:attract-club} to this (one-dimensional) restricted flow yields that $(A\cap E)\cap\A$ is closed under better replies on $E$ and since $\alpha^j\to\alpha^{j+1}$, we get $\alpha^{j+1}\in A\cap E$, a contradiction.
Thus $\HH\subseteq A$.

Next, we prove that every face $\F\subseteq \cont(\HH)$ satisfies $\F\subseteq A$, by induction on $d:=\dim \F$.
For $d=0$, $\F$ is a vertex in $\A$, hence lies in $\HH\subseteq A$.
Assume $d\ge 1$ and the claim holds for all proper subfaces $\F'\subsetneq \F$.
Note first that the boundary of $\F$ verifies $\partial \F=\bigcup_{\F'\subsetneq \F}\F'\subseteq A$.
Let $A_\F:=A\cap \F$, by \cref{lem:restrict-attract}, $A_\F$ is an attractor for the restricted flow on $\F$.
If $A_\F\neq \F$, consider the corresponding dual repellor $R_\F$ which is nonempty, compact, invariant in $\F$, and disjoint from $A_\F$.
Since $\partial \F\subseteq A_\F$, we have $R_\F\cap\partial \F=\varnothing$, hence $R_\F\subseteq \F^\circ$.
For the time-reversed flow on $\F$, the set $R_\F$ is an attractor, and by \cref{rmrk:reversal}, this time-reversed flow is itself a strategy flow (of the negated game) on $\F$.
This contradicts \cref{thm:no-interior}, therefore $A_\F=\F$, i.e.\ $\F\subseteq A$.

By induction, every face $\F\subseteq \cont(\HH)$ is contained in $A$, so $\cont(\HH)\subseteq A$, since always $A\subseteq \cont(\HH)$, we conclude $A=\cont(\HH)$.
Thus the restricted flow on $\cont(\HH)$ admits no proper attractor, and \cref{prop:ict-attractor} yields that $\cont(\HH)$ is internally chain transitive.
\end{proof}

\begin{proof}[\textbf{Proof of \cref{thm:sc-span}}]
Let $\HH \subseteq \A$ be again strongly connected. Since $\cont(\HH)$ is invariant, \cref{lem:restrict-attract} applied to the attractor $A$ and $\cont(\HH)$ shows that $A\cap \cont(\HH)$ is an attractor for the flow restricted to $\cont(\HH)$ (nonempty because $\HH\subseteq A\cap\cont(\HH)$).
By \cref{lem:ict-span}, $\cont(\HH)$ is internally chain transitive, hence the restricted flow on $\cont(\HH)$ admits no proper attractor.
Therefore $A\cap\cont(\HH)=\cont(\HH)$, i.e.\ $\cont(\HH)\subseteq A$.

In the case the whole preference graph is strongly connected, applying this to $\HH = \A$ yields that $\cont(\HH)= \X$ admits no proper attractor.
\end{proof}

\section{Omitted proofs from Section \ref{sec:prefstable}}
\label{app:prefstable}

\subsection{Energy functions and asymptotic stability}
Before we start, we shall first define the notion of \emph{energy functions}, introduced in \cite{HMC21,MZ19,MHC24}, which allow us to convert local dissipation in \emph{score space} into asymptotic stability in \emph{strategy space}. They are defined as follows:

\begin{definition}\label{def:local-energy}
Let $S\subseteq\X$ be a nonempty closed set. Under \eqref{eq:FTRL}, a map $E:\YY\to[0,\infty)$ is a
\emph{local energy function} for $S$ if:
\begin{enumerate}
\item\label{it:LE-smooth} $E$ is Lipschitz and $C^1$ on $\YY$.
\item\label{it:LE-zero} For every sequence $(y^k)_{k\in\mathbb N}\subseteq\YY$,
\[
Q(y^k)\to S \quad\Longleftrightarrow\quad E(y^k)\to 0.
\]
\item\label{it:LE-dissip} There exists $\bar E>0$ such that for all $0<E^-<E^+\le \bar E$,
\[
\sup\bigl\{\dot E(y):\ E^-<E(y)<E^+\bigr\}<0,
\]
where $\dot E(y):= \langle \nabla E(y),\,v(Q(y))\rangle$ is the derivative of $E$ along \eqref{eq:FTRL}.
\end{enumerate}
\end{definition}

The link with the long-run behavior of the continuous-time dynamics, which we prove, is as follows:

\begin{theorem}\label{thm:energy-ct}
Let $S\subseteq \X$ be nonempty and closed, if $S$ admits a local energy function, then $S$ is asymptotically stable under \eqref{eq:FTRL}.
\end{theorem}
\begin{proof}
    Let $y(\cdot)$ be a solution of $\dot y=v(Q(y))$ and set $x(t)=Q(y(t))$ and $\mathcal E(t):=E(y(t))$.
By \cref{def:local-energy}(\labelcref{it:LE-smooth}) and the chain rule, $\mathcal E$ is $C^1$ and
$\dot{\mathcal E}(t)=\dot E(y(t))$.

First, for every neighborhood $U$ of $S$ in $\X$ there exists $\delta_U>0$ such that
\begin{equation}\label{eq:smallE-near}
E(y)<\delta_U \ \Longrightarrow\ Q(y)\in U.
\end{equation}
Otherwise, there exist $U$ and $y_k$ with $E(y_k)\to 0$ but $Q(y_k)\notin U$, contradicting
$Q(y_k)\to S$ from \cref{def:local-energy}(\labelcref{it:LE-zero}).
Second, for every $\delta>0$ there exists a neighborhood $V_\delta$ of $S$ such that
\begin{equation}\label{eq:near-smallE}
Q(y)\in V_\delta \ \Longrightarrow\ E(y)<\delta.
\end{equation}
Otherwise, there exist $\delta>0$ and $y_m$ with $\dist(Q(y_m),S)<1/m$ but $E(y_m)\ge\delta$, contradicting
\cref{def:local-energy}(\labelcref{it:LE-zero}).

By \cref{def:local-energy}(\labelcref{it:LE-dissip}), fix $\bar E>0$ such that for every $0<E^-<E^+\le\bar E$,
\[
\kappa(E^-,E^+)\coloneqq-\sup\{\dot E(y):\ E^-<E(y)<E^+\}\;>\;0.
\tag{$\ast$}
\]
Hence, whenever $\mathcal E(t)\in(E^-,E^+)$, we have $\dot{\mathcal E}(t)\le -\kappa(E^-,E^+)$.
In particular, if $0<a<b\le\bar E$ and $\mathcal E(t_0)\le a$, then $\mathcal E(t)\le a$ for all $t\ge t_0$:
otherwise, letting $\tau=\inf\{t\ge t_0:\mathcal E(t)>a\}$ gives $\mathcal E(\tau)=a$ by continuity, and then continuity yields
$\varepsilon>0$ with $\mathcal E(t)\in(a,b)$ for $t\in(\tau,\tau+\varepsilon)$, so $\dot{\mathcal E}(t)\le -\kappa(a,b)<0$ on
$(\tau,\tau+\varepsilon)$ by ($\ast$), contradicting $\mathcal E(t)>a$ for $t>\tau$ close enough.

For stability, let $U$ be any neighborhood of $S$. Choose $\delta_U\in(0,\bar E]$ so that \eqref{eq:smallE-near} holds,
pick $a\in(0,\delta_U)$, and apply \eqref{eq:near-smallE} with $\delta=a$ to obtain a neighborhood $V$ of $S$ such that
$x(0)=Q(y(0))\in V$ implies $\mathcal E(0)<a$. The previous paragraph (with $(a,b)=(a,\delta_U)$) yields
$\mathcal E(t)\le a<\delta_U$ for all $t\ge 0$, hence $x(t)\in U$ for all $t\ge 0$ by \eqref{eq:smallE-near}, hence $S$ is Lyapunov stable.

For attraction, fix $b\in(0,\bar E]$ and assume $\mathcal E(0)<b$. Let $a\in(0,b)$ be arbitrary.
If $\mathcal E(0)\le a$, set $T_a=0$. Otherwise, if $\mathcal E(0)>a$ and $\mathcal E(t)>a$ for all $t\ge 0$, then
$\mathcal E(t)\in(a,b)$ for all $t\ge 0$, so $\dot{\mathcal E}(t)\le -\kappa(a,b)$ for all $t\ge 0$ and thus
\[
\mathcal E(t)\le \mathcal E(0)-\kappa(a,b)\,t,
\]
which is impossible for $t>\mathcal E(0)/\kappa(a,b)$ because $\mathcal E(t)\ge 0$. Hence there exists $T_a$ with $\mathcal E(T_a)\le a$,
and the previous paragraph implies $\mathcal E(t)\le a$ for all $t\ge T_a$. Since this holds for every $a\in(0,b)$, we obtain
$\mathcal E(t)\to 0$ as $t\to\infty$.
If $\dist(x(t),S)\not\to 0$, there exist $\varepsilon>0$ and $t_k\to\infty$ with $\dist(x(t_k),S)\ge\varepsilon$ for all $k$.
But $\mathcal E(t_k)=E(y(t_k))\to 0$, so \cref{def:local-energy}(\labelcref{it:LE-zero}) implies $x(t_k)=Q(y(t_k))\to S$, a contradiction, therefore $\dist(x(t),S)\to 0$.
Finally, applying \eqref{eq:near-smallE} with $\delta=b$ yields a neighborhood $U_0$ of $S$ such that
$x(0)\in U_0\cap\ImQ$ implies $\mathcal E(0)<b$, the above then shows $\dist(x(t),S)\to 0$ for all such initial conditions,
so $S$ is attracting.
\end{proof}

\subsection{When preferences are enough}
We start by constructing an energy function for spans of \ac{club} subgames.

\begin{lemma}\label{lem:face-energy}
Let $\B\subseteq\A$ be a subgame.
If $\B$ is closed under better replies, then the Fenchel gap
\[
F_\B(y):=h^*(y)-h_\B^*(y),
\qquad
h_\B^*(y):=\max_{w\in\cont(\B)}\{\langle y,w\rangle-h(w)\},
\]
is a local energy function for $\cont(\B)$.
\end{lemma}
\begin{proof}
Write $\bar\B_i:=\A_i\setminus\B_i$ and $\F:=\cont(\B) =\prod_{i\in\N}\Delta(\B_i)$.
For each $i\in\N$, define the restricted choice map\footnote{By abuse, we denote here $\Delta(\B_i)$ as the subset of mixed strategies of $w_i\in\X_i$ with $\supp(w_i) \subseteq\B_i$.}
\[
Q_{\B_i}(y_i):=\arg\max_{x_i\in\Delta(\B_i)}\{\langle y_i,x_i\rangle-h_i(x_i)\},
\qquad
Q_\B(y):=(Q_{\B_i}(y_i))_{i\in\N}\in\F,
\]
and set, throughout the proof,
\[
x:=Q(y),\qquad x^\B:=Q_\B(y),\qquad \Delta_i:=x_i-x_i^\B,\qquad m_i:=\sum_{\alpha\in\bar\B_i}x_{i\alpha}.
\]
We divide the proof into 6 steps.

\emph{Step 1:}
By \cref{prop:choice-map}(\labelcref{it:Qgrad}), $h^*$ is $C^1$ and $\nabla h^*(y)=Q(y)$.
Applying the same result playerwise to the reduced simplex $\Delta(\B_i)$ shows that the restricted conjugate
\[
h_{\B_i}^*(y_i):=\max_{w_i\in\Delta(\B_i)}\{\langle y_i,w_i\rangle-h_i(w_i)\}
\]
is $C^1$ with $\nabla h_{\B_i}^*(y_i)=Q_{\B_i}(y_i)$.
Since $h=\sum_i h_i$ and $\F$ is a product face, we have $h_\B^*(y)=\sum_i h_{\B_i}^*(y_i)$ and thus
\[
\nabla h_\B^*(y)=Q_\B(y).
\]
Consequently,
\begin{equation}\label{eq:grad-gap-FB}
\nabla F_\B(y)=\nabla h^*(y)-\nabla h_\B^*(y)=Q(y)-Q_\B(y)=x-x^\B.
\end{equation}
Because $x,x^\B\in\X$ for all $y$, the gradient \eqref{eq:grad-gap-FB} is bounded on $\YY$, hence $F_\B$ is globally Lipschitz and $C^1$, verifying \cref{def:local-energy}(\labelcref{it:LE-smooth}).

\emph{Step 2:}
By definition,
\[
h^*(y)=\max_{x\in\X}\{\langle y,x\rangle-h(x)\},
\qquad
h_\B^*(y)=\max_{w\in\F}\{\langle y,w\rangle-h(w)\},
\]
so
\begin{equation}\label{eq:face-coupling}
F_\B(y)
=h^*(y)-h_\B^*(y)
=\min_{w\in\F}\bigl(h(w)+h^*(y)-\langle y,w\rangle\bigr)
=\min_{w\in\F}F_h(w,y),
\end{equation}
where $F_h(w,y)=\sum_{i\in\N}F_i(w_i,y_i)$ is the (playerwise) Fenchel coupling already defined in \cref{app:mirror}.
By Fenchel--Young, $F_i(\cdot,y_i)\ge 0$ for each $i$, hence $F_\B(y)\ge 0$ for all $y$.
Moreover, $F_\B(y)=0$ if and only if there exists $w\in\F$ with $F_h(w,y)=0$, then each $F_i(w_i,y_i)=0$, so $w_i=Q_i(y_i)$ for all $i$ by \cref{prop:fenchel}(\labelcref{it:fenchel-zero}), i.e.\ $Q(y)=w\in\F$.
Conversely, if $Q(y)\in\F$, choosing $w=Q(y)$ in \eqref{eq:face-coupling} gives $F_\B(y)\le F_h(Q(y),y)=0$, hence $F_\B(y)=0$.
Thus,
\begin{equation}\label{eq:FB-zero-set}
F_\B(y)=0\quad\Longleftrightarrow\quad Q(y)\in\F.
\end{equation}

\emph{Step 3:}
Let $K_{\min}:=\min_{i\in\N}K_i$.
Fix $y\in\YY$ and $x=Q(y)$.
Using \eqref{eq:face-coupling} and the quadratic lower bound of the Fenchel coupling from \cref{prop:fenchel}(\labelcref{it:fenchel-lb}),
\[
F_\B(y)
=\min_{w\in\F}\sum_{i\in\N}F_i(w_i,y_i)
\ge \min_{w\in\F}\sum_{i\in\N}\frac{K_i}{2}\,\|x_i-w_i\|_2^2
\ge \frac{K_{\min}}{2}\,\dist(x,\F)^2.
\]
Therefore, if $F_\B(y^k)\to 0$ along some sequence $(y^k)$, then $\dist(Q(y^k),\F)\to 0$, i.e.\ $Q(y^k)\to\F$ because $\F$ is closed.
Conversely, assume $\dist(Q(y^k),\F)\to 0$ and set $x^k:=Q(y^k)$.
Choose $w^k\in\F$ such that $\|x^k-w^k\|_2\le \dist(x^k,\F)+1/k$, hence $\|x^k-w^k\|_2\to 0$.
By compactness of $\F$, pass to a subsequence (not relabeled) with $w^k\to w\in\F$, then also $x^k\to w$.
For each $i$, $Q_i(y_i^k)=x_i^k\to w_i$, so \cref{prop:fenchel}(\labelcref{it:fenchel-zero}) yields $F_i(w_i,y_i^k)\to 0$.
Summing over $i$, we get $F_h(w,y^k)\to 0$ and then, by \eqref{eq:face-coupling},
\[
0\le F_\B(y^k)=\min_{u\in\F}F_h(u,y^k)\le F_h(w,y^k)\to 0.
\]
Hence $F_\B(y^k)\to 0$.
This proves \cref{def:local-energy}(\labelcref{it:LE-zero}).

\emph{Step 4:}
Fix $i\in\N$, $\beta_i\in\B_i$ and $\alpha_i\in\bar\B_i$.
For every pure opponents' profile $\gamma_{-i}\in\prod_{j\neq i}\B_j$, closure of $\B$ under better replies means that $(\alpha_i,\gamma_{-i})$ cannot be a weakly improving deviation from $(\beta_i,\gamma_{-i})$, hence
\[
u_i(\beta_i,\gamma_{-i})>u_i(\alpha_i,\gamma_{-i})
\qquad
\forall\,\gamma_{-i}\in\prod_{j\neq i}\B_j.
\]
By multilinearity, the same strict inequality holds for all mixed opponents' profiles $w_{-i}\in\prod_{j\neq i}\Delta(\B_j)$, i.e., for all $w \in \F$:
\[
v_{i\beta}(w)>v_{i\alpha}(w).
\]
Consider now the gap $v_{i\beta}(w)-v_{i\alpha}(w)$, this is continuous on the compact set $\F$ and strictly positive everywhere, so\footnote{With the convention that the minimum over an empty index set is $+\infty$.}
\[
\delta'_i:=\min_{\beta_i\in\B_i,\ \alpha_i\in\bar\B_i}\ \min_{w\in\F}\; v_{i\beta}(w)-v_{i\alpha}(w) >0
\]
Set $\delta_i:=\delta'_i/2$ for those $i$ with $\bar\B_i\neq\varnothing$.
By uniform continuity of $v$ on compact $\X$, there exists $\varepsilon>0$ such that whenever $\dist(x,\F)<\varepsilon$, we have,
for all $i\in\N$ with $\bar\B_i\neq\varnothing$, $\beta_i\in\B_i$, $\alpha_i\in\bar\B_i$,
\begin{equation}\label{eq:gap-near-F}
v_{i\beta}(x)-v_{i\alpha}(x)\ge \delta_i.
\end{equation}

\emph{Step 5:}
By \cref{prop:choice-map}(\labelcref{it:KKT}), there exist $\lambda_i\in\R$ and multipliers $\nu_i\in\R_+^{\A_i}$ such that for all $\gamma_i\in \A_i$,
\[
y_{i\gamma}=\theta_i'(x_{i\gamma})+\lambda_i-\nu_{i\gamma},
\qquad
x_{i\gamma}>0\ \Rightarrow\ \nu_{i\gamma}=0.
\]
Likewise, applying the same KKT statement to the restricted maximization over $\Delta(\B_i)$, there exist $\lambda_i^\B\in\R$ and $\nu_i^\B\in\R_+^{\B_i}$ such that for all $\beta_i\in \B_i$,
\[
y_{i\beta}=\theta_i'(x^\B_{i\beta})+\lambda_i^\B-\nu^\B_{i\beta},
\qquad
x^\B_{i\beta}>0\ \Rightarrow\ \nu^\B_{i\beta}=0,
\qquad
x^\B_{i\alpha}=0\ \ (\alpha_i\in \A_i\setminus \B_i).
\]
Define $\theta_i'(0+):=\lim_{p\downarrow 0}\theta_i'(p)\in[-\infty,\infty)$ and the truncated inverse
\[
\varphi_i(z):=
\begin{cases}
0, & z\le \theta_i'(0+),\\
(\theta_i')^{-1}(z), & z>\theta_i'(0+),
\end{cases}
\]
which is nondecreasing.
We claim that
\begin{equation}\label{eq:choice-repr}
x_{i\gamma}=\varphi_i(y_{i\gamma}-\lambda_i)\quad(\gamma_i\in \A_i),
\qquad
x^\B_{i\beta}=\varphi_i(y_{i\beta}-\lambda_i^\B)\quad(\beta_i\in \B_i).
\end{equation}
Indeed, if $x_{i\gamma}>0$, then $\nu_{i\gamma}=0$ and $y_{i\gamma}-\lambda_i=\theta_i'(x_{i\gamma})$, so $x_{i\gamma}=(\theta_i')^{-1}(y_{i\gamma}-\lambda_i)=\varphi_i(y_{i\gamma}-\lambda_i)$.
If $x_{i\gamma}=0$, then $y_{i\gamma}-\lambda_i=\theta_i'(0+)-\nu_{i\gamma}\le \theta_i'(0+)$ (interpreting $\theta_i'(x_{i\gamma})$ as $\theta_i'(0+)$ at $x_{i\gamma=0}$), hence $\varphi_i(y_{i\gamma}-\lambda_i)=0=x_{i\gamma}$.
The same reasoning applies to $x^\B$ on $\B_i$.
Summing \eqref{eq:choice-repr} over $\A_i$ and $\B_i$ respectively gives
\[
\sum_{\gamma\in \A_i}\varphi_i(y_{i\gamma}-\lambda_i)=1,
\qquad
\sum_{\beta\in \B_i}\varphi_i(y_{i\beta}-\lambda_i^\B)=1.
\]
Since $\varphi_i\ge 0$,
\[
\sum_{\beta\in \B_i}\varphi_i(y_{i\beta}-\lambda_i)\le \sum_{\gamma\in \A_i}\varphi_i(y_{i\gamma}-\lambda_i)=1.
\]
Because $\varphi_i$ is nondecreasing, the function $\tau\mapsto \sum_{\beta\in \B_i}\varphi_i(y_{i\beta}-\tau)$ is nonincreasing, thus the equality
$\sum_{\beta\in \B_i}\varphi_i(y_{i\beta}-\lambda_i^\B)=1$ and the inequality at $\tau=\lambda_i$ imply $\lambda_i^\B\le \lambda_i.$\footnote{Unless both sums happen to equal 1 at $\tau = \lambda_i$, in which case the maximizer of the unrestricted problem is in $\F$ anyway and one might take $\lambda_i=\lambda_i^\B$.}
Therefore, for every $\beta\in B$,
\[
y_{i\beta}-\lambda_i \le y_{i\beta}-\lambda_i^\B
\quad\Rightarrow\quad
x_{i\beta}=\varphi_i(y_{i\beta}-\lambda_i)\le \varphi_i(y_{i\beta}-\lambda_i^\B)=x^\B_{i\beta}.
\]
In summary,
\begin{equation}\label{eq:mass-transfer}
x_{i\beta}\le x^\B_{i\beta}\ \ (\beta\in \B_i),
\qquad
x^\B_{i\alpha}=0\ \ (\alpha\in \bar \B_i).
\end{equation}
Consequently, we have the sign and mass identities
\begin{equation}\label{eq:mass-signs}
\Delta_{i\alpha}\ge 0\ (\alpha\in \bar \B_i),\qquad
\Delta_{i\beta}\le 0\ (\beta\in \B_i),\qquad
m_i=\sum_{\alpha\in \bar \B_i}\Delta_{i\alpha}=-\sum_{\beta_i\in \B_i}\Delta_{i\beta}.
\end{equation}

\emph{Step 6:}
Along continuous-time \ac{FTRL}, $\dot y=v(Q(y))=v(x)$, so by \eqref{eq:grad-gap-FB},
\begin{equation}\label{eq:Fdot-basic}
\dot F_\B(y)=\langle \nabla F_\B(y),\dot y\rangle
=\langle x-x^\B, v(x)\rangle
=\sum_{i\in\N}\langle \Delta_i, v_i(x)\rangle.
\end{equation}
Assume now that $\dist(x,\F)<\varepsilon$, with $\varepsilon$ from \eqref{eq:gap-near-F}.
Fix $i$ with $\bar\B_i\neq\varnothing$, then \eqref{eq:gap-near-F} yields $v_{i\alpha}(x)\le  \min_{\beta_i\in\B_i}v_{i\beta}(x)-\delta_i$ for all $\alpha_i\in\bar\B_i$ and $v_{i\beta}(x)\ge \min_{\beta_i\in\B_i}v_{i\beta}(x)$ for all $\beta_i\in\B_i$.
Using \eqref{eq:mass-signs},
\begin{align}
\langle \Delta_i, v_i(x)\rangle
&=\sum_{\alpha_i\in\bar\B_i}v_{i\alpha}(x)\Delta_{i\alpha}+\sum_{\beta_i\in\B_i}v_{i\beta}(x)\Delta_{i\beta}
\notag\\
&\le (\min_{\beta_i\in\B_i}v_{i\beta}(x)-\delta_i)\sum_{\alpha_i\in\bar\B_i}\Delta_{i\alpha}+\min_{\beta_i\in\B_i}v_{i\beta}(x)\sum_{\beta_i\in\B_i}\Delta_{i\beta}
\notag\\
&=(\min_{\beta_i\in\B_i}v_{i\beta}(x)-\delta_i)m_i-\min_{\beta_i\in\B_i}v_{i\beta}(x)\, m_i
=-\delta_i m_i
=-\delta_i\sum_{\alpha_i\in\bar\B_i}x_{i\alpha}.
\end{align}
Summing over $i$ in \eqref{eq:Fdot-basic} gives, whenever $\dist(x,\F)<\varepsilon$,
\begin{equation}\label{eq:Fdot-drift}
\dot F_\B(y)
\le -\sum_{i\in\N:\,\bar\B_i\neq\varnothing}\delta_i\sum_{\alpha\in\bar\B_i}x_{i\alpha}.
\end{equation}

Now set $K_{\min}:=\min_i K_i$ and define
\[
\bar E:=\frac{K_{\min}}{2}\,\varepsilon^2.
\]
By Step 3, $F_\B(y)\le \bar E$ implies $\dist(Q(y),\F)\le\varepsilon$, so the drift bound \eqref{eq:Fdot-drift} holds on the entire sublevel set $\{F_\B\le \bar E\}$.
Fix $0<E^-<E^+\le \bar E$ and consider the domain $\mathcal{D}:=\{y\in\YY:\ E^-<F_\B(y)<E^+\}$.
Define the continuous outside-mass functional
\[
M(y):=\sum_{i\in\N}\sum_{\alpha_i\in\bar\B_i}Q_{i\alpha}(y_i)=\sum_{i\in\N} m_i.
\]
Let $\delta_{\min}:=\min_{i:\,\bar\B_i\neq\varnothing}\delta_i>0$ (if $\bar\B_i=\varnothing$ for all $i$, then $\F=\X$ and $F_\B= 0$ everywhere, so the dissipation requirement is vacuous).
Then \eqref{eq:Fdot-drift} gives, for all $y\in\mathcal{D}$,
\[
\dot F_\B(y)\le -\delta_{\min} M(y).
\]
It remains to show that $M$ is bounded away from $0$ on $\mathcal{D}$.
Suppose not, then there exists a sequence $y^k\in\mathcal{D}$ with $M(y^k)\to 0$.
Let $x^k:=Q(y^k)$ and define $w^k\in\F$ by renormalizing $x_i^k$ on $\B_i$:
if $m_i^k:=\sum_{\alpha_i\in\bar\B_i}x_{i\alpha}^k$ and $1-m_i^k>0$, set
\[
w_{i\beta}^k:=\frac{x_{i\beta}^k}{1-m_i^k}\ \ (\beta_i\in\B_i),
\qquad
w_{i\alpha}^k:=0\ \ (\alpha_i\in\bar\B_i).
\]
Then $w^k\in\F$ and, for each $i$,
\[
\|x_i^k-w_i^k\|_2^2
=\sum_{\beta_i\in\B_i}\left(x_{i\beta}^k-\frac{x_{i\beta}^k}{1-m_i^k}\right)^2+\sum_{\alpha_i\in\bar\B_i}(x_{i\alpha}^k)^2
=\left(\frac{m_i^k}{1-m_i^k}\right)^2\sum_{\beta_i\in\B_i}(x_{i\beta}^k)^2+\sum_{\alpha_i\in\bar\B_i}(x_{i\alpha}^k)^2.
\]
Using $\sum_{\beta_i\in\B_i}(x_{i\beta}^k)^2\le (\sum_{\beta_i\in\B_i}x_{i\beta}^k)^2=(1-m_i^k)^2$ and $\sum_{\alpha_i\in\bar\B_i}(x_{i\alpha}^k)^2\le (\sum_{\alpha_i\in\bar\B_i}x_{i\alpha}^k)^2=(m_i^k)^2$, we obtain
\[
\|x_i^k-w_i^k\|_2^2\le (m_i^k)^2+(m_i^k)^2=2(m_i^k)^2,
\qquad\text{so}\qquad
\|x_i^k-w_i^k\|_2\le \sqrt{2}\,m_i^k.
\]
Therefore,
\[
\dist(x^k,\F)\le \|x^k-w^k\|_2
\le \sqrt{\sum_{i\in\N}\|x_i^k-w_i^k\|_2^2}
\le \sqrt{2}\sum_{i\in\N}m_i^k
=\sqrt{2}\,M(y^k)\to 0.
\]
By Step 3, this implies $F_\B(y^k)\to 0$, contradicting $F_\B(y^k)>E^-$ for all $k$.
Hence $\inf_{y\in\mathcal{D}}M(y)=:\kappa(E^-,E^+)>0$, and thus
\[
\sup_{y\in\mathcal{D}}\dot F_\B(y)\le -\delta_{\min}\kappa(E^-,E^+)<0,
\]
which is exactly \cref{def:local-energy}(\labelcref{it:LE-dissip}).

All three items of \cref{def:local-energy} are verified, so $F_\B$ is a local energy function for $\F=\cont(\B)$.
\end{proof}

We may now harvest what we sowed:
\begin{proof}[\textbf{Proof of \cref{thm:subgame-stab}}]
Combine \cref{thm:energy-ct} with \cref{lem:face-energy}.
\end{proof}

\begin{proof}[\textbf{Proof of \cref{cor:subgame-equiv}}]
Combine \cref{thm:stable-sclub} with \cref{thm:subgame-stab}.
\end{proof}

\begin{proof}[\textbf{Proof of \cref{cor:subgame-attract}}]
We shall construct a basin of attraction on all of $\X$ by gluing together all the facewise basins of \eqref{eq:FTRL} restricted to each face.

Let $\F=\cont(\B)=\prod_{i\in\N}\Delta(\B_i)$. First, assume that $\B$ is closed under better replies. We show that $\F$ is an attractor for the strategy flow induced by \eqref{eq:SD}, i.e., $\F$ is invariant and asymptotically stable. Let $\F'=\cont(\B')=\prod_{i\in\N}\Delta(\B_i')$ be any face of $\X$ such that $\F\cap \F'\neq\varnothing$. Then $\B_i\cap\B_i'\neq\varnothing$ for all $i$, and
\begin{equation}\label{eq:face-intersect}
\F\cap \F' = \prod_{i\in\N}\Delta(\B_i\cap \B_i')\;=\;\cont(\B|_{\B'}),\qquad
\B|_{\B'}\coloneqq \prod_{i\in\N}(\B_i\cap \B_i').
\end{equation}
Indeed, $x\in\F\cap\F'$ if and only if for every $i$, $x_i\in\Delta(\B_i)\cap\Delta(\B_i')=\Delta(\B_i\cap\B_i')$. Moreover, $\B|_{\B'}$ is closed under better replies in the preference graph restricted to the vertex set $\B'$: if $\alpha\in\B|_{\B'}$ and $\alpha\to\beta$ with $\beta\in\B'$, then $\alpha\in\B$ and the arc $\alpha\to\beta$ is also an arc of the full preference graph, hence $\beta\in\B$ by closure of $\B$, so $\beta\in\B\cap\B'=\B|_{\B'}$.

Fix such a face $\F'=\cont(\B')$ and consider the facewise \ac{FTRL} dynamics {\renewcommand{\FTRLtag}{\B'}\eqref{eq:FTRL-B}} on $\F'$. Applying \cref{thm:subgame-stab} to the restricted game on action sets $\B'$ yields that
\[
\F\cap\F'=\cont(\B|_{\B'})
\]
is asymptotically stable under {\renewcommand{\FTRLtag}{\B'}\eqref{eq:FTRL-B}}. Under \cref{asm:well-posed}, \cref{lem:steep-image} gives $\ImQ_{\B'}=\F'{}^\circ$, and by \cref{rmrk:flow-link} the restriction of the strategy flow to $\F'$ coincides with the induced strategy orbit of {\renewcommand{\FTRLtag}{\B'}\eqref{eq:FTRL-B}} for initial conditions in $\F'{}^\circ$. Hence, for every $\varepsilon>0$ and every such $\F'$, there exists a neighborhood $U_{\F'}^\varepsilon\subseteq \F'$ of $\F\cap\F'$ in the relative topology of $\F'$ such that every strategy flow trajectory with initial condition $x_0\in U_{\F'}^\varepsilon\cap \F'{}^\circ$ satisfies
\begin{equation}\label{eq:face-stability}
\dist(\Theta_t(x_0),\,\F\cap\F')<\varepsilon\ \ \forall t\ge 0,
\qquad
\lim_{t\to+\infty}\dist(\Theta_t(x_0),\,\F\cap\F')=0.
\end{equation}
Let now $\mathfrak{F}$ be the (finite) set of faces $\F'$ of $\X$ such that $\F\cap\F'\neq\varnothing$. Fix $\varepsilon>0$. For each $\F'\in\mathfrak{F}$, let $U_{\F'}^\varepsilon$ be as in \eqref{eq:face-stability}. Fix $x\in\F$ and define $\mathfrak{F}(x)\coloneqq\{\F'\in\mathfrak{F}:\ x\in\F'\}$. Since $\X$ has finitely many faces, set
\[
r_x\coloneqq \frac12\min\{\dist(x,\F''):\ \F'' \text{ a face of }\X,\ x\notin \F''\}>0,
\]
so that $B_{r_x}(x)\cap\X$ meets only faces $\F'$ with $x\in\F'$ (equivalently, only faces in $\mathfrak{F}(x)$). Next, for each $\F'\in\mathfrak{F}(x)$, since $U_{\F'}^\varepsilon$ is a neighborhood of $\F\cap\F'$ in $\F'$ and $x\in\F\cap\F'$, there exists $\rho_{\F'}^x>0$ such that
\[
B_{\rho_{\F'}^x}(x)\cap \F' \subseteq U_{\F'}^\varepsilon.
\]
Define
\[
\rho_x\coloneqq \min\Bigl(r_x,\min_{\F'\in\mathfrak{F}(x)}\rho_{\F'}^x\Bigr)>0,
\qquad
W_x\coloneqq B_{\rho_x}(x)\cap\X,
\]
so $W_x$ is open in $\X$. Let $z\in W_x$ and let $\F_z$ denote the (unique) minimal face of $\X$ containing $z$. By construction of $r_x$, every face of $\X$ intersecting $W_x$ contains $x$, so in particular $x\in\F_z$, hence $\F_z\in\mathfrak{F}(x)\subseteq\mathfrak{F}$. Moreover,
\[
z\in B_{\rho_x}(x)\cap \F_z\subseteq U_{\F_z}^\varepsilon
\]
by construction of $\rho_x$. By definition of minimal face, $z\in \F_z^\circ$ and by \cref{prop:flow}, faces are invariant for the strategy flow, so $\Theta_t(z)\in\F_z$ for all $t\ge 0$. Applying \eqref{eq:face-stability} with $\F'=\F_z$ yields
\[
\dist(\Theta_t(z),\F\cap\F_z)<\varepsilon\ \ \forall t\ge 0,
\qquad
\dist(\Theta_t(z),\F\cap\F_z)\to 0.
\]
Since $\F\cap\F_z\subseteq \F$, we have $\dist(\Theta_t(z),\F)\le \dist(\Theta_t(z),\F\cap\F_z)$, hence
\[
\dist(\Theta_t(z),\F)<\varepsilon\ \ \forall t\ge 0,
\qquad
\dist(\Theta_t(z),\F)\to 0.
\]
Because $\F$ is compact, the open cover $\{W_x\}_{x\in\F}$ admits a finite subcover $\F\subseteq\bigcup_{k=1}^m W_{x^k}$. Set
\[
W^\varepsilon\coloneqq \bigcup_{k=1}^m W_{x^k},
\]
an open neighborhood of $\F$ in $\X$. Then for all $z\in W^\varepsilon$,
\[
\dist(\Theta_t(z),\F)<\varepsilon\ \ \forall t\ge 0,
\qquad
\dist(\Theta_t(z),\F)\to 0.
\]
Since $\varepsilon>0$ was arbitrary, this proves that $\F$ is stable and attracting for the strategy flow. Moreover, $\F$ is invariant by \cref{prop:flow}. Hence $\F$ is an invariant asymptotically stable set, i.e., an attractor.

Conversely, if $\F$ is an attractor for the strategy flow, then in particular it is attracting. By \cref{prop:attract-club}, $\F\cap\A$ is closed under better replies. Since $\F=\cont(\B)$, we have $\F\cap\A=\B$, so $\B$ is closed under better replies.
\end{proof}

\begin{proof}[\textbf{Proof of \cref{thm:wa-attract}}]
First, we show strict \aclp{NE} are attractors. Let $\alpha^\ast\in\A$ be such an equilibrium and identify it with the corresponding vertex of $\X$.
Then for every player $i$ and every $\alpha_i\neq \alpha_i^\ast$ we have
$u_i(\alpha_i,\alpha_{-i}^\ast)<u_i(\alpha^\ast)$, hence there is no arc $\alpha^\ast\to\beta$ in the preference graph with $\beta\neq \alpha^\ast$.
Therefore $\{\alpha^\ast\}$ is a \ac{club} subgame. By \cref{cor:subgame-attract}, $\cont(\{\alpha^\ast\})=\{\alpha^\ast\}$
is an attractor for the strategy flow induced by \eqref{eq:SD}. Since it is a singleton, it contains no proper (nonempty)
subset, hence it is minimal among attractors.

Conversely, let $M\subseteq\X$ be a minimal attractor for the strategy flow. In particular $M$ is asymptotically stable.
By \cref{cor:stable-club}, there exists a nonempty set of pure profiles $\HH\subseteq M\cap\A$ which is closed under better replies.
Fix $\alpha^0\in\HH$. Since the game is weakly acyclic, there exists a finite better-reply path
\[
\alpha^0\to \alpha^1\to \cdots \to \alpha^k
\]
ending at a pure \acl{NE} $\alpha^k$. Because $\HH$ is closed under better replies and $\alpha^0\in\HH$,
an immediate induction gives $\alpha^j\in\HH$ for all $j$, hence $\alpha^k\in\HH\subseteq M$. Moreover, since the game has no ties, $\alpha^k$ must be a strict \acl{NE}. Now, by the first implication, strictness implies that $\{\alpha^k\}$ is an attractor. Since $\{\alpha^k\}\subseteq M$ and $M$ is minimal among
attractors, we conclude $M=\{\alpha^k\}$, i.e.\ $M$ is a strict \acl{NE}.
\end{proof}

\begin{proof}[\textbf{Proof of \cref{rmrk:potential}.}]
Note that, by definition, a finite game is an \emph{ordinal potential game} if there exists $\Phi:\A\to\R$ such that for every player $i$ and every $i$-comparable
$\alpha,\alpha'\in\A$,
\[
u_i(\alpha')>u_i(\alpha)\quad\Longleftrightarrow\quad \Phi(\alpha')>\Phi(\alpha).
\]

Now, if $\Phi$ is an ordinal potential, every arc $\alpha\to\alpha'$ yields $u_i(\alpha')>u_i(\alpha)$ (no ties), hence $\Phi(\alpha')>\Phi(\alpha)$.
Thus $\Phi$ strictly increases along directed edges, so no directed cycle exists. Conversely, if the preference graph is acyclic, fix a topological ordering $\alpha^1,\dots,\alpha^{|\A|}$ and set $\Phi(\alpha^k)=k$.
Then $\alpha\to\alpha'$ implies $\Phi(\alpha')>\Phi(\alpha)$. For $i$-comparable $\alpha,\alpha'$, the unique arc between them points toward the larger payoff.
\end{proof}

\subsection{When preference are \emph{not} enough}

We start with the counterexample.
\begin{proof}[\textbf{Proof of \cref{prop:span-counter}}]
Players have two actions each:
\[
\A_1=\{\mathrm F,\mathrm B\}\quad(\text{Front/Back}),\qquad
\A_2=\{\mathrm L,\mathrm R\}\quad(\text{Left/Right}),\qquad
\A_3=\{\mathrm T,\mathrm B\}\quad(\text{Top/Bottom}).
\]
We abbreviate a pure profile $\alpha=(\alpha_1,\alpha_2,\alpha_3)\in\A$ by the three-letter word $\alpha_1\alpha_2\alpha_3$, e.g. $\mathrm{FLT}=(\mathrm F,\mathrm L,\mathrm T).$ Payoffs $u(\alpha)=(u_1,u_2,u_3)$ are:

\begin{table}[htbp]
\centering
\begin{minipage}{0.45\linewidth}
\centering
\textbf{Front:}
\begin{tabular}{c|cc}
   & \(\mathrm L\) & \(\mathrm R\) \\ \hline
\(\mathrm T\) & \((1,0,0)\) & \((0,1,0)\) \\
\(\mathrm B\) & \((2,0,10)\) & \((0,2,0)\)
\end{tabular}
\end{minipage}\hspace{1.5em}
\begin{minipage}{0.45\linewidth}
\centering
\textbf{Back:}
\begin{tabular}{c|cc}
   & \(\mathrm L\) & \(\mathrm R\) \\ \hline
\(\mathrm T\) & \((0,1,1)\) & \((1,0,1)\) \\
\(\mathrm B\) & \((0,0,0)\) & \((1,1,0)\)
\end{tabular}
\end{minipage}
\label{tab:counter}
\end{table}
Let $\mathcal H:=\A\setminus\{\mathrm{BLB}\}$ and $\mathcal S:=\cont(\HH)\subseteq\X$. In the preference graph of this game, $\mathrm{BLB}$ has no incoming arc and $\mathcal H$ is the unique \ac{club} set (see \cref{fig:counter-cube}). In the cube $\X=\Delta(\A_1)\times\Delta(\A_2)\times\Delta(\A_3)$, write the coordinates
\[
p:=x_{1\mathrm F}\in[0,1],\qquad q:=x_{2\mathrm R}\in[0,1],\qquad r:=x_{3\mathrm T}\in[0,1].
\]
Then $\mathrm{BLB}$ corresponds to $(p,q,r)=(0,0,0)$ and one checks
\begin{equation}\label{eq:counter-span}
\mathcal S=\{p=1\}\ \cup\ \{q=1\}\ \cup\ \{r=1\},
\end{equation}
i.e.\ the union of the three square faces not containing $\mathrm{BLB}$. Let also $\F_{\mathrm T}:=\{r=1\}$, i.e.\ player $3$ plays $\mathrm T$ purely. Restricting to $\F_{\mathrm T}$ gives the $2$-player game between players $(1,2)$:
\[
\begin{array}{c|cc}
        & \mathrm L & \mathrm R\\ \hline
\mathrm B& (0,1) & (1,0)\\
\mathrm F& (1,0) & (0,1)
\end{array}
\]
hence the unique mixed equilibrium is
\[
p^\ast=\tfrac12,\qquad q^\ast=\tfrac12.
\]
Define
\[
x^\ast:=(p^\ast,q^\ast,1)=\Bigl(\tfrac12,\tfrac12,1\Bigr)\in \F_{\mathrm T}.
\]
At $x^\ast$, player $3$ has a strict profitable deviation to $\mathrm B$:
\[
v_{3\mathrm T}(x^\ast)=1-p^\ast=\tfrac12,
\qquad
v_{3\mathrm B}(x^\ast)=10\,p^\ast(1-q^\ast)=10\cdot\tfrac14=\tfrac52,
\]
so $v_{3\mathrm B}(x^\ast)-v_{3\mathrm T}(x^\ast)=2>0$.

\begin{figure}[tbp]
\centering
\begin{tikzpicture}[
  scale=3,
  line join=round,
  line cap=round,
  edge/.style={line width=0.7pt,draw=black!50},
  sinkspan/.style={draw=red!70!black, line width=1.6pt, opacity=0.80},
  midorient/.style={-{Stealth[length=2.3mm]}, line width=0.9pt, draw=black},
  sinkorient/.style={-{Stealth[length=2.3mm]}, line width=0.9pt, draw=red!70!black},
  sinkorientbi/.style={{Stealth[length=2.3mm]}-{Stealth[length=2.3mm]}, line width=0.9pt, draw=red!70!black}
]

\coordinate (BLB) at (0,0);
\coordinate (BRB) at (1,0);
\coordinate (BRT) at (1,1);
\coordinate (BLT) at (0,1);

\coordinate (FLB) at (0.40,0.30);
\coordinate (FRB) at (1.40,0.30);
\coordinate (FRT) at (1.40,1.30);
\coordinate (FLT) at (0.40,1.30);

\fill[red!25,opacity=0.25] (FLB)--(FRB)--(FRT)--(FLT)--cycle;
\fill[red!25,opacity=0.20] (BRB)--(FRB)--(FRT)--(BRT)--cycle;
\fill[red!25,opacity=0.15] (BLT)--(BRT)--(FRT)--(FLT)--cycle;

\draw[edge] (FLB)--(FRB)--(FRT)--(FLT)--cycle;
\draw[edge] (BLB)--(BRB)--(BRT)--(BLT)--cycle;
\draw[edge] (FLB)--(BLB);
\draw[edge] (FRB)--(BRB);
\draw[edge] (FRT)--(BRT);
\draw[edge] (FLT)--(BLT);

\draw[sinkspan] (FLB)--(FRB)--(FRT)--(FLT)--cycle;
\draw[sinkspan] (FRB)--(BRB)--(BRT)--(FRT)--cycle;
\draw[sinkspan] (FLT)--(FRT)--(BRT)--(BLT)--cycle;

\draw[midorient]  ($(BLB)!0.48!(FLB)$) -- ($(BLB)!0.52!(FLB)$);
\draw[sinkorient] ($(BLT)!0.48!(FLT)$) -- ($(BLT)!0.52!(FLT)$);
\draw[sinkorient] ($(FRB)!0.48!(BRB)$) -- ($(FRB)!0.52!(BRB)$);
\draw[sinkorient] ($(FRT)!0.48!(BRT)$) -- ($(FRT)!0.52!(BRT)$);

\draw[midorient]  ($(BLB)!0.48!(BRB)$) -- ($(BLB)!0.52!(BRB)$);
\draw[sinkorient] ($(FLB)!0.48!(FRB)$) -- ($(FLB)!0.52!(FRB)$);
\draw[sinkorient] ($(BRT)!0.48!(BLT)$) -- ($(BRT)!0.52!(BLT)$);
\draw[sinkorient] ($(FLT)!0.48!(FRT)$) -- ($(FLT)!0.52!(FRT)$);

\draw[midorient]  ($(BLB)!0.48!(BLT)$) -- ($(BLB)!0.52!(BLT)$);
\draw[sinkorient] ($(BRB)!0.48!(BRT)$) -- ($(BRB)!0.52!(BRT)$);
\draw[sinkorient] ($(FLT)!0.48!(FLB)$) -- ($(FLT)!0.52!(FLB)$);

\foreach \V in {FLB,FRB,FRT,FLT,BRB,BRT,BLT}{
  \fill[red!70!black] (\V) circle (0.9pt);
}
\fill[blue!70!black] (BLB) circle (1.1pt);

\coordinate (xstar) at ($(BLT)!0.5!(FRT)$);
\fill[magenta!80!black] (xstar) circle (1pt);

\end{tikzpicture}
\caption{Cube $\X$ with coordinates $(p,q,r)=(x_{1\mathrm F},x_{2\mathrm R},x_{3\mathrm T})$.
The red vertices represent the unique \ac{club} set, and the red shaded faces are its span. The blue vertex is a \emph{source vertex}, meaning it has no incoming edge. In particular it is a repellor of the dynamics. The purple point $x^\ast=(\tfrac12,\tfrac12,1)$ is a \acl{NE} of the game restricted to the face $\F_{\mathrm T}$. Nonetheless, Player 3 has a deviation to the bottom in the neighborhood of that point, hence the dynamics are locally repelled from $\F_{\mathrm T}$ near $x^*$. This is what shall drive instability.}
\label{fig:counter-cube}
\end{figure}
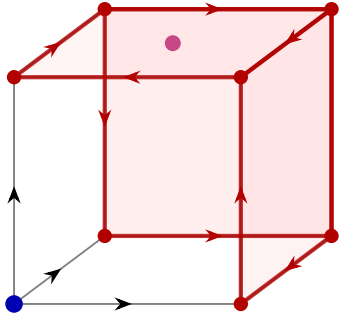

We now consider the dynamics \eqref{eq:FTRL}. Write
\[
a:=p-\tfrac12,\qquad b:=q-\tfrac12,\qquad s:=1-r,
\]
and let
\[
\widehat x:=(x_1,x_2,\mathrm T)=(p,q,1)\in\F_{\mathrm T}
\]
be the projection of $x=(p,q,r)$ onto the top face obtained by replacing player $3$'s mixed action by the pure action $\mathrm T$. The point $x^*=(\tfrac12,\tfrac12,1)$ is \emph{variationally stable }for the game restricted to $\F_{\mathrm T}$. Indeed, on the top face,
\[
v_{1\mathrm F}(\widehat x)-v_{1\mathrm B}(\widehat x)=-2b,
\qquad
v_{2\mathrm R}(\widehat x)-v_{2\mathrm L}(\widehat x)=2a,
\]
so
\begin{align}\label{eq:counter-vs}
&\left\langle v_1(\widehat x),x_1-x_1^*\right\rangle
+\left\langle v_2(\widehat x),x_2-x_2^*\right\rangle
\notag\\
&\qquad
=a\bigl(v_{1\mathrm F}(\widehat x)-v_{1\mathrm B}(\widehat x)\bigr)
+b\bigl(v_{2\mathrm R}(\widehat x)-v_{2\mathrm L}(\widehat x)\bigr)
=-2ab+2ab=0.
\end{align}
This cancellation is what will keep the first two players close to $(1/2,1/2)$ while player $3$ moves away from the top face. To quantify this, define the restricted Fenchel energy
\[
\mathcal E(y)
:=F_{1}(x^*_1,y_1) + F_{2}(x^*_2,y_2).
\]
Since $\nabla h_i^*(y_i)=Q_i(y_i)=x_i$, differentiation along \eqref{eq:FTRL} gives
\begin{equation}\label{eq:counter-Edot-gen}
\dot{\mathcal E}
=\langle v_1(x),x_1-x_1^*\rangle + \langle v_2(x),x_2-x_2^*\rangle.
\end{equation}
A direct calculation from the payoff table yields
\[
v_{1\mathrm F}(x)-v_{1\mathrm B}(x)
=\frac{s}{2}-(2+s)b,
\]
\[
v_{2\mathrm R}(x)-v_{2\mathrm L}(x)
=\frac{3s}{2}+(2-s)a,
\]
and
\begin{equation}\label{eq:counter-gap}
v_{3\mathrm B}(x)-v_{3\mathrm T}(x)
=2+6a-5b-10ab.
\end{equation}
Substituting the first two identities in \eqref{eq:counter-Edot-gen}, we obtain
\begin{equation}\label{eq:counter-Edot}
\dot{\mathcal E}
=a\left(\frac{s}{2}-(2+s)b\right)
+b\left(\frac{3s}{2}+(2-s)a\right)
=s\left(\frac a2+\frac{3b}{2}-2ab\right).
\end{equation}
This \enquote{tangential} energy is thus exactly conserved when $s=0$, in agreement with \eqref{eq:counter-vs}; away from the top face, its drift is proportional to the  $s$. Let
\[
K:=\min\{K_1,K_2\}>0,
\]
where $K_i$ is the strong convexity constant of $h_i$. \Cref{prop:fenchel}(\labelcref{it:fenchel-lb}) gives
\[
F_{i}(x_i^*,y_i)
\ge \frac{K_i}{2}\|x_i-x_i^*\|_2^2.
\]
Since $\|x_1-x_1^*\|_2^2=2a^2$ and $\|x_2-x_2^*\|_2^2=2b^2$, it follows that
\begin{equation}\label{eq:counter-coercive}
\mathcal E(y)\ge K(a^2+b^2).
\end{equation}
Moreover, $|a|,|b|\le1/2$, so
\[
\left|\frac a2+\frac{3b}{2}-2ab\right|
\le \frac32(|a|+|b|)
\le \frac{3}{\sqrt2}\sqrt{a^2+b^2}.
\]
Combining this estimate with \eqref{eq:counter-Edot} and \eqref{eq:counter-coercive}, we get
\[
\dot{\mathcal E}
\le C s\sqrt{\mathcal E},
\qquad
C:=\frac{3}{\sqrt{2K}}.
\]
Applying this inequality to $\sqrt{\mathcal E+\delta}$ and then letting $\delta\downarrow0$ gives, for every interval on which the trajectory is defined,
\begin{equation}\label{eq:counter-Ebound}
\sqrt{\mathcal E(y(t))}
\le \sqrt{\mathcal E(y(0))}
+\frac C2\int_0^t s(u)\,du.
\end{equation}
\begin{figure*}[t]
\centering
\includegraphics[height=25ex]{Figures/Counter1-flow.pdf}
\quad
\includegraphics[height=25ex]{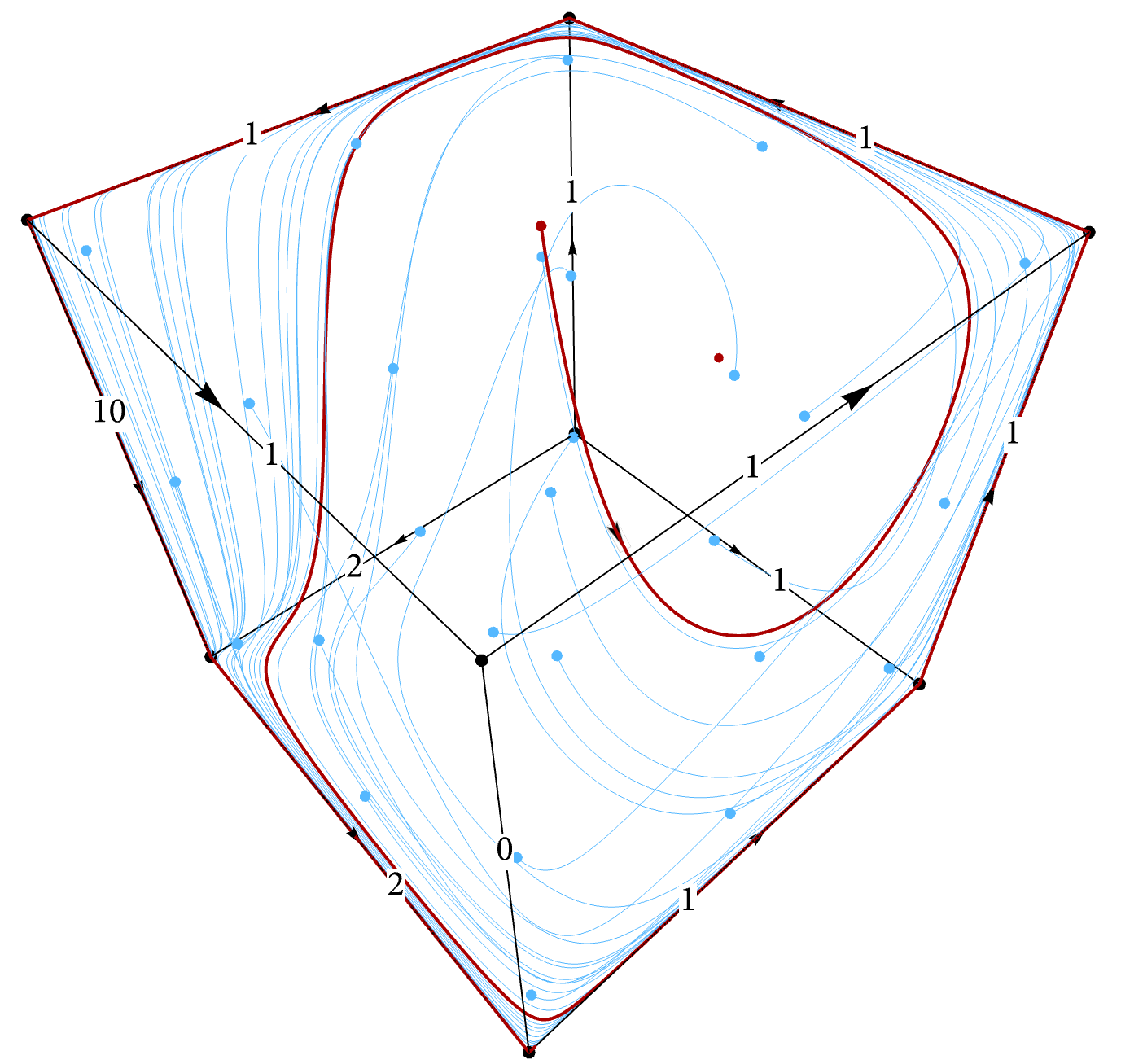}
\quad
\includegraphics[height=25ex]{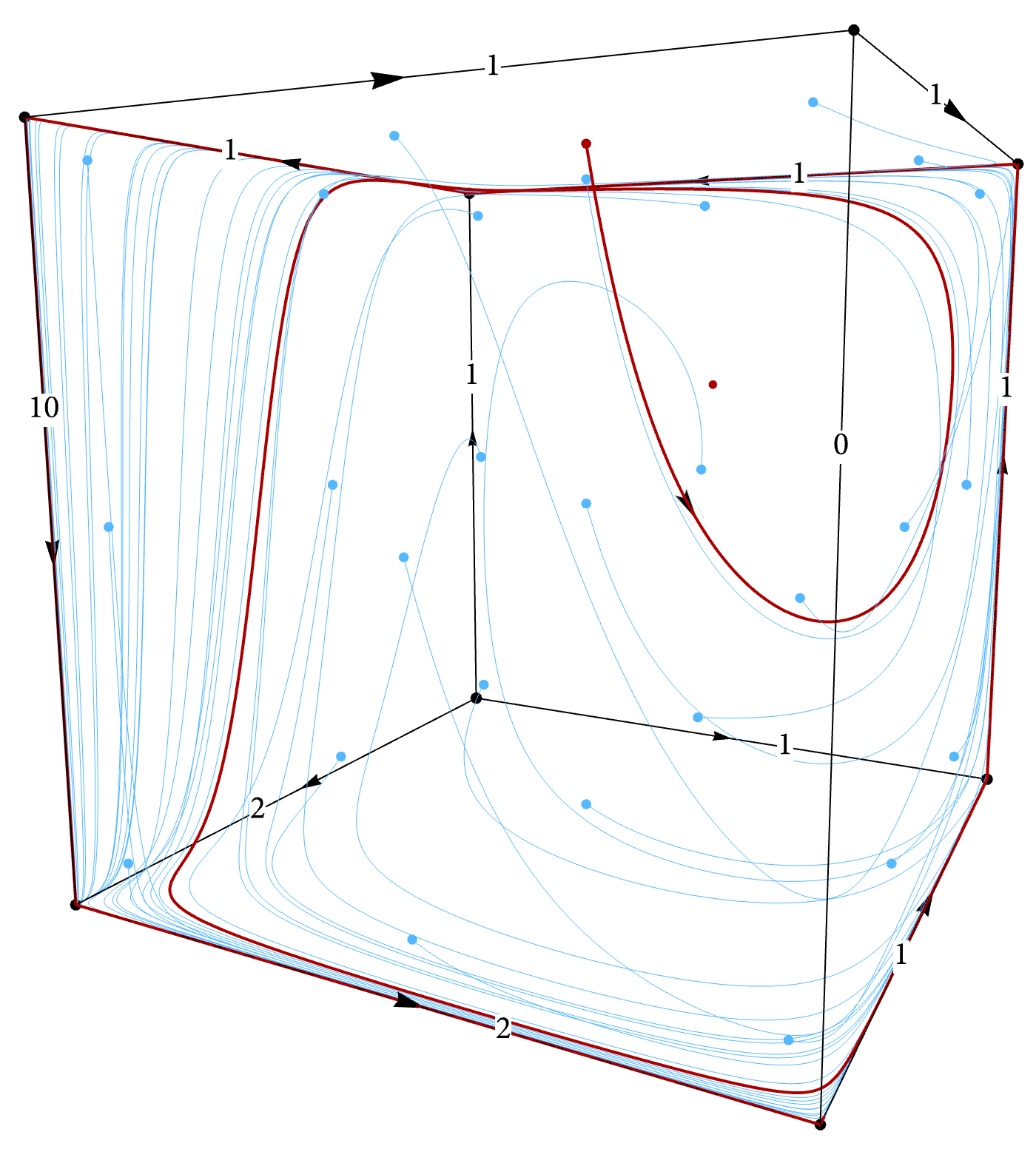}
\caption{Three perspectives of \eqref{eq:FTRL} under entropic regularization for this counterexample game.}
\label{fig:counter-flow}
\end{figure*}

We next estimate the growth of $s=x_{3\mathrm B}$. Fix
\[
\rho:=\frac1{20},
\qquad
e_0:=K\rho^2,
\qquad
d:=2-11\rho-10\rho^2>0.
\]
By \eqref{eq:counter-coercive},
\begin{equation}\label{eq:counter-smallE}
\mathcal E(y)<e_0
\quad\Longrightarrow\quad
|a|<\rho\ \text{ and }\ |b|<\rho.
\end{equation}
Hence, whenever $\mathcal E(y)<e_0$, \eqref{eq:counter-gap} gives
\[
v_{3\mathrm B}(x)-v_{3\mathrm T}(x)
\ge 2-6|a|-5|b|-10|ab|
\ge d.
\]
For $u\in(0,1)$, set
\[
\sigma(u):=\frac{1}{\theta_3''(u)+\theta_3''(1-u)}>0.
\]
As long as $0<s<1$, the support of player $3$ is fixed, and the strategy equation \eqref{eq:SD} reduces, in $\Delta(\{\mathrm T,\mathrm B\})$, to
\[
\dot s
=\sigma(s)\bigl(v_{3\mathrm B}(x)-v_{3\mathrm T}(x)\bigr).
\]
Indeed, writing $m(u)=1/\theta_3''(u)$, the coefficient of the payoff difference in \eqref{eq:SD} is
\[
\frac{m(s)m(1-s)}{m(s)+m(1-s)}
=\frac{1}{\theta_3''(s)+\theta_3''(1-s)}.
\]
Therefore, on the region $\{\mathcal E<e_0\}$,
\begin{equation}\label{eq:counter-sdot-lb}
\dot s\ge d\,\sigma(s)>0.
\end{equation}
In particular, $s$ is strictly increasing there. Moreover the quantity controlling the accumulated drift is
\[
\Gamma(\eta)
:=\int_0^\eta\frac{u}{\sigma(u)}\,du
=\int_0^\eta
u\bigl(\theta_3''(u)+\theta_3''(1-u)\bigr)\,du.
\]
We claim that
\begin{equation}\label{eq:counter-Gamma0}
\Gamma(\eta)\longrightarrow0
\qquad\text{as }\eta\downarrow0.
\end{equation}
To see this, convexity gives, for $u\in(0,1/2)$,
\[
\theta_3(u)-\theta_3(0)
\le u\theta_3'(u)
\le \theta_3(2u)-\theta_3(u).
\]
Since $\theta_3$ is continuous at $0$, both outer terms converge to $0$, so $u\theta_3'(u)\to0$. Integration by parts then yields
\[
\int_0^\eta u\theta_3''(u)\,du
=\eta\theta_3'(\eta)-\theta_3(\eta)+\theta_3(0)
\longrightarrow0.
\]
Also, $\theta_3''$ is bounded near $1$, so
\[
\int_0^\eta u\theta_3''(1-u)\,du=O(\eta^2).
\]
This proves \eqref{eq:counter-Gamma0}. Now, choose
\begin{equation}\label{eq:counter-eta}
0<\eta<\frac12-\rho
\end{equation}
small enough that
\begin{equation}\label{eq:counter-Gamma-ub}
\frac{C}{2d}\Gamma(\eta)<\frac{\sqrt{e_0}}{2},
\end{equation}
and define the neighborhood
\[
U:=\left\{x\in\X:
\dist(x,\mathcal S)<\frac{\sqrt2}{2}\eta
\right\}.
\]
Let $V$ be an arbitrary neighborhood of $\mathcal S$. For $\varepsilon>0$, consider
\[
x^\varepsilon(0)
:=\left(\frac12,\frac12,1-\varepsilon\right).
\]
As $\varepsilon\downarrow0$, this point converges to $x^*\in\mathcal S$, so
$x^\varepsilon(0)\in V$ for all sufficiently small $\varepsilon$. It is also an admissible initialization for \eqref{eq:FTRL}. Indeed, for a binary mixed action $(w,1-w)$ with $w\in(0,1)$, the score pair
\[
\bigl(\theta_i'(w),\theta_i'(1-w)\bigr)
\]
has $(w,1-w)$ as its unique regularized best response: the derivative of the one-dimensional maximization problem vanishes at $w$, and strict convexity gives uniqueness. Thus we may realize $x^\varepsilon(0)$ by taking
\[
y_{1\mathrm F}(0)=y_{1\mathrm B}(0)=\theta_1'(1/2),
\qquad
y_{2\mathrm L}(0)=y_{2\mathrm R}(0)=\theta_2'(1/2),
\]
and
\[
y_{3\mathrm T}(0)=\theta_3'(1-\varepsilon),
\qquad
y_{3\mathrm B}(0)=\theta_3'(\varepsilon)
\]
for player $3$. With this choice,
\begin{equation}\label{eq:counter-E0}
\mathcal E(y(0))=0.
\end{equation}
Fix
\[
0<\varepsilon<\frac\eta2
\]
small enough that $x^\varepsilon(0)\in V$, and let $x^\varepsilon(t)$ be the corresponding trajectory. Define
\[
\tau:=\inf\left\{t\ge0:
\mathcal E(y(t))=e_0\ \text{or}\ s(t)=\eta
\right\}.
\]
For $0\le t<\tau$, we have $\mathcal E(y(t))<e_0$ and $s(t)<\eta$. Hence \eqref{eq:counter-sdot-lb} applies, and, because $s$ is increasing,
\begin{align}
\int_0^t s(u)\,du
&\le \frac1d\int_0^t\frac{s(u)}{\sigma(s(u))}\dot s(u)\,du\notag\\
&=\frac1d\int_\varepsilon^{s(t)}\frac{v}{\sigma(v)}\,dv
\le\frac1d\Gamma(\eta).
\label{eq:counter-s-int}
\end{align}
Combining \eqref{eq:counter-Ebound},
\eqref{eq:counter-E0},
\eqref{eq:counter-s-int}, and
\eqref{eq:counter-Gamma-ub}, we obtain
\begin{equation}\label{eq:counter-bootstrap}
\sqrt{\mathcal E(y(t))}
<\frac{\sqrt{e_0}}2
\qquad\forall t<\tau.
\end{equation}
By continuity, the same strict bound holds at $t=\tau$ whenever $\tau<\infty$, so the event $\mathcal E=e_0$ cannot occur at time $\tau$. Finally, $\tau$ is finite. Indeed, $\sigma$ is continuous and strictly positive on $[\varepsilon,\eta]$, so
\[
\underline\sigma_{\varepsilon,\eta}
:=\min_{u\in[\varepsilon,\eta]}\sigma(u)>0.
\]
If $\tau=\infty$, then $\mathcal E(y(t))<e_0$ and $s(t)<\eta$ for every $t\ge0$, and \eqref{eq:counter-sdot-lb} would give
\[
\dot s(t)\ge d\,\underline\sigma_{\varepsilon,\eta}>0
\qquad\forall t\ge0,
\]
forcing $s$ to reach $\eta$ in finite time, a contradiction. Hence $\tau<\infty$; since the energy event is impossible, necessarily
\[
s(t_\eta)=\eta
\]
for $t_\eta:=\tau$, and \eqref{eq:counter-bootstrap} holds up to $t_\eta$. At this time, \eqref{eq:counter-smallE} gives
\[
p(t_\eta)<\frac12+\rho,
\qquad
q(t_\eta)<\frac12+\rho.
\]
By \eqref{eq:counter-eta},
\[
1-p(t_\eta)>\eta,
\qquad
1-q(t_\eta)>\eta.
\]
Finally, by \eqref{eq:counter-span}, the Euclidean distance to $\mathcal S$ is
\[
\dist(x,\mathcal S)
=\sqrt2\min\{1-p,1-q,s\}.
\]
Therefore
\[
\dist(x(t_\eta),\mathcal S)=\sqrt2\,\eta
>\frac{\sqrt2}{2}\eta,
\]
so $x(t_\eta)\notin U$. On the other hand,
\[
\dist(x^\varepsilon(0),\mathcal S)
=\sqrt2\,\varepsilon
<\frac{\sqrt2}{2}\eta,
\]
so $x^\varepsilon(0)\in U\cap V\cap\ImQ$. We have found a fixed neighborhood $U$ of $\mathcal S$ such that every neighborhood $V$ of $\mathcal S$ contains an admissible initial condition whose trajectory under \eqref{eq:FTRL} leaves $U$. Hence $\mathcal S$ is not stable under \eqref{eq:FTRL}.

\end{proof}
\begin{remark}
\label{rmrk:not-span}
Note that in the game considered in the above proof, the vertex BLB is a repellor. Indeed, it has no incoming edge (see the blue vertex in \cref{fig:counter-cube}), so it is a strict \acl{NE} of the negated game, hence an attractor for that game, and therefore (see \cref{rmrk:reversal}) a repellor for our game. Consider then the dual attractor of BLB (see \cref{subsec:attractors} for definition of dual repellor/attractor). This attractor must contain the span of the unique \ac{club} set (combining \cref{thm:sc-span} and \cref{cor:stable-club}), but it cannot coincide with that span since the latter is not stable by the previous proposition. Moreover, any larger span containing it would necessarily contain BLB, hence the dual attractor cannot be such a span either. Consequently, this attractor is not a span of pure profiles.
\endenv
\end{remark}

\begin{remark}
\label{rmrk:biggar}
\citet{BP25} introduce the notion of a \emph{local source}, defined as follows. Let $x\in\X$, let $\mathcal H\subseteq\A$ be a sink equilibrium (recall: a minimal nonempty \ac{club} set), and set $\mathcal S=\cont(\mathcal H)$. Let $\B$ be a subgame and let $\F=\cont(\B)$ be its span. Then $x$ is a local source of $\mathcal H$ in $\F$ if: (1) $x\in\mathcal S\cap\F$, (2) $\F\nsubseteq \mathcal S$, and (3) $x$ is a quasi-strict \acl{NE} of the negated game $-u$ restricted to $\F$.

In our game, the point
\[
x^*=\Bigl(\tfrac12,\tfrac12,1\Bigr)
\]
is indeed a local source, with $\F=\X$. Conditions (1) and (2) hold because $x^*\in\mathcal S$ while $\X\nsubseteq\mathcal S$. For condition (3), players $1$ and $2$ are indifferent between their two actions at $x^*$ and both actions belong to their supports. Player $3$ plays $\mathrm T$ purely, and the strict payoff inequality
\[
v_{3\mathrm B}(x^*)-v_{3\mathrm T}(x^*)=2>0
\]
in the original game implies that $\mathrm T$ is the unique best response of player $3$ in the negated game. Hence $x^*$ is a quasi-strict \acl{NE} of $-u$ on $\X$.

Accordingly, this example is consistent with the local source instability result of \cite{BP25}. It is also special relative to their results in two respects. First, the instability proof above uses full-support initializations $x^\varepsilon(0)\in\X^\circ$ converging to $x^*$, so the escape occurs from the interior rather than from an initialization a proper subface, which would be incompatible with the \eqref{eq:FTRL}-adapted notion of stability where initializations are defined as $x(0) = Q(y(0))$, and thus must be contained in $\ImQ$. Second, the argument applies to every \eqref{eq:FTRL} dynamics satisfying our regularizer assumptions, whereas \cite{BP25} work with \eqref{eq:RD} only.
\endenv
\end{remark}

\section{Omitted proofs from Section \ref{sec:rad}}
\label{app:rad}

We begin by proving that \ac{radness} implies \acl{clubness}:

\begin{proof}[\textbf{Proof of \cref{prop:resilient-club}.}]
We prove the elementary identity from which both claims follow directly. Let
$\beta\in\HH$ and let $\alpha\in\A\setminus\HH$ be $i$-comparable with $\beta$,
so $\alpha_{-i}=\beta_{-i}$ and $\alpha_i\neq\beta_i$. Then, for every
$j\neq i$, we have $\alpha_j=\beta_j$, and hence
\[
u_j(\alpha_j,\beta_{-j})=u_j(\beta_j,\beta_{-j})=u_j(\beta).
\]
Therefore
\begin{align}
\Phi(\beta,\alpha)
&=\sum_{j\in\N}\bigl(u_j(\alpha_j,\beta_{-j})-u_j(\beta)\bigr)
\notag\\
&=u_i(\alpha_i,\beta_{-i})-u_i(\beta)
\notag\\
&=u_i(\alpha)-u_i(\beta),
\end{align}
where the last equality uses $\alpha_{-i}=\beta_{-i}$.
\end{proof}
We have also an upgrade to mixed profiles:

\begin{lemma}\label{lem:resilience-mixed}
    Let $\HH \subseteq \A$. Then $\HH$ is \ac{rad} (resp.~\ac{srad}) if and only if for every $\alpha \notin \HH$, and for every face $\F \subseteq \cont(\HH)$, we have $ \sup_{x \in \F} \Phi(x,\alpha) \leq 0$ (resp. $\sup_{x \in \F} \Phi(x,\alpha) < 0)$.
\end{lemma}
\begin{proof}
The forward direction is trivial (just take $\F$ to be pure profile of $\HH$). For the other direction, fix $\alpha\notin\HH$ and a face $\F\subseteq \cont(\HH)$ and $\F=\cont(\B)$ for some subgame $\B$,
so $\B\subseteq \HH$. We have
\[
\Phi(x,\alpha)=\sum_{i\in\N}\bigl(v_{i\alpha}(x)-u_i(x)\bigr),
\]
so $\Phi(\cdot,\alpha)$ is multilinear. Next, we claim that for any multilinear $f$ on $\F$,
\begin{equation}\label{eq:vertex-max}
\sup_{x\in\F} f(x)=\max_{\beta\in\B} f(\beta).
\end{equation}
This is classic, to see why it's true, simply fix $x_{-1}$ and consider $x_1\mapsto f(x_1,x_{-1})$, which is affine on
$\Delta(\B_1)$, thus it attains its maximum at a vertex $\beta_1$ with $\beta_1\in\B_1$, giving
$f(x)\le f(\beta_1,x_{-1})$. Iterating this argument over $i=2,\dots,n$ yields some
$\beta=(\beta_1,\dots,\beta_n)\in\B$ with $f(x)\le f(\beta)$ for all $x\in\F$. Applying \eqref{eq:vertex-max} to $f=\Phi(\cdot,\alpha)$ and using $\B\subseteq\HH$, we get by hypothesis
\[
\sup_{x\in\F} \Phi(x,\alpha)=\max_{\beta\in\B} \Phi(\beta,\alpha)\le 0
\quad\text{(resp.\ $<0$).}
\]
Since $\alpha\notin\HH$ and $\F\subseteq\cont(\HH)$ were arbitrary, $\HH$ is \ac{rad} (resp.~\ac{srad}).
\end{proof}

To proceed, we shall prove the following lemma.

\begin{lemma}
    \label{lem:res-energy}
Let $\HH \subseteq \A$ be \ac{srad} and assume $h$ is steep. Define the function
\begin{equation}\label{eq:res-energy-app}
\bar F_\HH(y)\coloneqq \sum_{\alpha\notin\HH} e^{-F_h(\alpha,y)}\in[0,\infty),
\; y\in \YY,
\end{equation}
with the convention $e^{-\infty}=0$. Then $\bar F_\HH$ is a local energy function for $\cont(\HH)$.
\end{lemma}

\begin{proof}
Set $\cS\coloneqq \cont(\HH)$ and define $E:\YY\to[0,\infty)$ by
\[
E(y)\coloneqq \bar F_\HH(y)=\sum_{\alpha\notin\HH} e^{-F_h(\alpha,y)}.
\]
If $\A\setminus\HH=\varnothing$, then $\cS=\X$ and $E= 0$, so all three items of \cref{def:local-energy} hold trivially. Assume henceforth that $\A\setminus\HH\neq\varnothing$.

We start with regularity.
By \cref{prop:choice-map}(\labelcref{it:Qgrad}), $h^*$ is $C^1$ on $\YY$ and $\nabla h^*(y)=Q(y)$, hence, for every fixed $\alpha\in\A$,
\[
F_h(\alpha,y)=h(\alpha)+h^*(y)-\langle y,\alpha\rangle
\quad\text{is $C^1$ with}\quad
\nabla F_h(\alpha,y)=Q(y)-\alpha.
\]
Moreover, $F_h(\alpha,y)\ge 0$ by Fenchel--Young, so $\exp(-F_h(\alpha,y))\in(0,1]$.
Differentiating gives
\[
\nabla\!\left(e^{-F_h(\alpha,\cdot)}\right)(y)
=-e^{-F_h(\alpha,y)}\bigl(Q(y)-\alpha\bigr).
\]
Since $Q(y)\in\X$ for all $y$ and $\alpha\in\X$ is fixed, the vectors $Q(y)-\alpha$ are uniformly bounded, therefore the gradients above are uniformly bounded as well (because $e^{-F_h(\alpha,y)}\le 1$). It follows that each map $y\mapsto e^{-F_h(\alpha,y)}$ is globally Lipschitz and $C^1$ on $\YY$. Summing over the finite set $\A\setminus\HH$ shows that $E$ is globally Lipschitz and $C^1$, proving \cref{def:local-energy}(\labelcref{it:LE-smooth}).

Now positive semi-definiteness.
Let $(y^k)$ be any sequence in $\YY$ and set $x^k\coloneqq Q(y^k)$. We first show that $x^k\to\cS$ implies $E(y^k)\to 0$.
Fix $\alpha\notin\HH$. Since $x\mapsto x_\alpha=\prod_i x_{i\alpha}$ is continuous and $x_\alpha=0$ for all $x\in\cS$ by definition of $\cont(\HH)$, we have $x^k_\alpha\to 0$. Let
\[
\delta^k_\alpha\coloneqq \min_{i\in\N} x^k_{i\alpha_i}.
\]
Because $0\le x^k_{i\alpha_i}\le 1$ and $x^k_\alpha=\prod_i x^k_{i\alpha_i}\to 0$, we must have $\delta^k_\alpha\to 0$.
We claim that $F_h(\alpha,y^k)\to+\infty$, which implies $e^{-F_h(\alpha,y^k)}\to 0$. Indeed, \cref{asm:well-posed} implies that $h$ is steep, so $Q(y^k)\in\X^\circ$ for all $k$, in particular all coordinates of $x^k$ are strictly positive. Thus, for each $k$ and each player $i$, the KKT conditions in \cref{prop:choice-map}(\labelcref{it:KKT}) hold with $\nu_i\equiv 0$, yielding a scalar $\lambda_i^k\in\R$ such that
\[
y^k_{i\beta}=\theta_i'(x^k_{i\beta})+\lambda_i^k,\qquad \forall\,\beta\in\A_i.
\]
Fix $k$ and choose an index $j$ attaining the minimum in $\delta^k_\alpha$, i.e.\ $x^k_{j\alpha}=\delta^k_\alpha$.
Also choose $\beta_j\in\A_{j}$ with $x^k_{j\beta}\ge 1/|\A_{j}|$.
Using the definition of $h_{j}^*$ and testing it at the pure action $\beta_j$ gives
\[
h_{j}^*(y^k_{j})\ge \langle y^k_{j},\beta_j\rangle-h_{j}(\beta_j)
= y^k_{j}-h_{j}(\beta_j),
\]
so
\[
F_{j}(\alpha_{j},y^k_{j})
=h_{j}({\alpha_{j}})+h_{j}^*(y^k_{j})-y^k_{j\alpha}
\ge h_{j}({\alpha}_j)-h_{j}({\beta_j})+\bigl(y^k_{j\beta}-y^k_{j\alpha}\bigr).
\]
By the KKT identity, $y^k_{j\beta}-y^k_{j\alpha}=\theta_{j}'(x^k_{j\beta})-\theta_{j}'(x^k_{j\alpha})$.
Since $\theta_{j}'$ is increasing and $x^k_{j\beta}\ge 1/|\A_{j}|$, we have
$\theta_{j}'(x^k_{j\beta})\ge \theta_{j}'(1/|\A_{j}|)$, hence
\[
F_{j}(\alpha_{j},y^k_{j})
\ge \Bigl(\theta_{j}'(1/|\A_{j}|)-\bigl(\max_{\gamma\in\A_{j}}h_{j}(\gamma_j)-\min_{\gamma\in\A_{j}}h_{j}(\gamma_j)\bigr)\Bigr)
-\theta_{j}'(\delta^k_\alpha).
\]
Define for each $i$ the constant
\[
C_i\coloneqq \theta_i'(1/|\A_i|)-\bigl(\max_{\gamma_i\in\A_i}h_i(\gamma_i)-\min_{\gamma_i\in\A_i}h_i(\gamma_i)\bigr)\in\R.
\]
Then the preceding inequality yields
\[
F_h(\alpha,y^k)=\sum_{j\in\N}F_j(\alpha_j,y^k_j)\ \ge\ F_{j}(\alpha_{j},y^k_{j})
\ \ge\ \min_{i\in\N}\bigl(C_i-\theta_i'(\delta^k_\alpha)\bigr).
\]
Since we are in the steep regime, $\theta_i'(p)\to-\infty$ as $p\downarrow 0$ for each $i$. Because $\delta^k_\alpha\to 0$ and $\N$ is finite, the right-hand side diverges to $+\infty$. Hence $F_h(\alpha,y^k)\to+\infty$ and therefore $e^{-F_h(\alpha,y^k)}\to 0$. As $\A\setminus\HH$ is finite, summing over $\alpha\notin\HH$ yields $E(y^k)\to 0$.

Conversely, assume $E(y^k)\to 0$. We show that $x^k\to\cS$.
Fix $\alpha\notin\HH$ and suppose, for contradiction, that $x^k_\alpha\nrightarrow 0$. Then there exist $\varepsilon>0$ and a subsequence (not relabeled) such that $x^k_\alpha\ge\varepsilon$ for all $k$, so in particular $x^k_{i\alpha_i}\ge\varepsilon$ for all players $i$ (since each factor is $\le 1$).
Fix such an $i$. By the same KKT identity as above (valid because $x^k\in\X^\circ$), we have
\[
\max_{\beta_i\in\A_i} y^k_{i\beta}-y^k_{i\alpha_i}
=\max_{\beta_i\in\A_i}\theta_i'(x^k_{i\beta})-\theta_i'(x^k_{i\alpha_i})
\le \theta_i'(1)-\theta_i'(\varepsilon)\eqqcolon M_i(\varepsilon)<\infty.
\]
Hence $y^k_{i\alpha}\ge \max_\beta y^k_{i\beta}-M_i(\varepsilon)$.
On the other hand, by definition of $h_i^*$,
\[
h_i^*(y^k_i)=\sup_{w_i\in\X_i}\bigl(\langle y^k_i,w_i\rangle-h_i(w_i)\bigr)
\le \max_{\beta_i\in\A_i} y^k_{i\beta}-\min_{w_i\in\X_i}h_i(w_i).
\]
Combining these two bounds gives
\[
F_i(\alpha_i,y^k_i)
=h_i(\alpha_i)+h_i^*(y^k_i)-y^k_{i\alpha}
\le h_i(\alpha_i)-\min_{\X_i}h_i + M_i(\varepsilon)
\eqqcolon C'_i(\varepsilon)<\infty.
\]
Summing over $i$ yields a finite constant $ C'(\varepsilon)=\sum_i  C'_i(\varepsilon)$ such that
$F_h(\alpha,y^k)\le  C'(\varepsilon)$ along the subsequence, hence
$e^{-F_h(\alpha,y^k)}\ge e^{- C'(\varepsilon)}>0$ along the same subsequence.
This contradicts $E(y^k)\to 0$, since $E(y^k)\ge e^{-F_h(\alpha,y^k)}$.
Therefore $x^k_\alpha\to 0$ for every $\alpha\notin\HH$. Now define the continuous function $g:\X\to[0,\infty)$ by $g(x)\coloneqq \sum_{\alpha\notin\HH} x_\alpha$.
We have $g(x)=0$ if and only if $x\in\cS$, and the above shows $g(x^k)\to 0$.
If $x^k\nrightarrow \cS$, there exist $\varepsilon_0>0$ and a subsequence $x^{k_j}$ with $\dist(x^{k_j},\cS)\ge\varepsilon_0$.
By compactness of $\X$, pass to a further subsequence with $x^{k_j}\to \bar x\in\X$.
Continuity of $g$ gives $g(\bar x)=0$, hence $\bar x\in\cS$, contradicting $\dist(x^{k_j},\cS)\ge\varepsilon_0$ for $j$ large.
Thus $x^k\to\cS$.
This proves \cref{def:local-energy}(\labelcref{it:LE-zero}).

Finally, dissipativity.
Let $y$ be any solution orbit of \eqref{eq:FTRL} and set $x\coloneqq Q(y)$.
For each $\alpha\notin\HH$, summing \cref{lem:fenchel-deriv} over players gives
\[
\frac{d}{dt}F_h(\alpha,y)
=\sum_{i\in\N}\langle v_i(x),\,x_i-\alpha_i\rangle
=\sum_{i\in\N}\bigl(u_i(x)-v_{i\alpha}(x)\bigr)
=-\,\Phi(x,\alpha),
\]
where we used $u_i(x)=\langle v_i(x),x_i\rangle$ and the definition of $\Phi(\cdot,\alpha)$.
Therefore,
\[
\frac{d}{dt}e^{-F_h(\alpha,y)}
=-e^{-F_h(\alpha,y)}\frac{d}{dt}F_h(\alpha,y)
=e^{-F_h(\alpha,y)}\,\Phi(x,\alpha).
\]
Summing over $\alpha\notin\HH$ yields the derivative of $E$ along \eqref{eq:FTRL}:
\begin{equation}\label{eq:dotE-flux}
\dot E(y)=\sum_{\alpha\notin\HH} e^{-F_h(\alpha,y)}\,\Phi(Q(y),\alpha).
\end{equation}

We now exploit strict \ac{radness}.
Fix $\alpha\notin\HH$. For any $x\in\cS$, let $\F$ be the (minimal) face of $\X$ containing $x$, then $\F\subseteq\cS$.
By \cref{lem:resilience-mixed}, $\sup_{z\in\F} \Phi(z,\alpha)<0$, hence
$\Phi(x,\alpha)\le \sup_{z\in\F} \Phi(z,\alpha)<0$.
Thus $\Phi(\cdot,\alpha)<0$ on $\cS$. Since $\Phi(\cdot,\alpha)$ is continuous and $\cS$ is compact, the maximum
$m_\alpha\coloneqq \max_{x\in\cS} \Phi(x,\alpha)$ exists and satisfies $m_\alpha<0$.
Because $\A\setminus\HH$ is finite, the constant
\[
c\coloneqq -\max_{\alpha\notin\HH} m_\alpha \;>\;0
\]
is well-defined and satisfies $\Phi(x,\alpha)\le -c$ for all $x\in\cS$ and all $\alpha\notin\HH$.
By continuity, there exists a neighborhood $U$ of $\cS$ in $\X$ such that
\[
\Phi(x,\alpha)\le -\frac{c}{2}
\qquad\forall\,x\in U,\ \forall\,\alpha\notin\HH.
\]
By \cref{def:local-energy}(\labelcref{it:LE-zero}) already established, there exists $\bar E>0$ such that
$E(y)<\bar E$ implies $Q(y)\in U$ (otherwise, one could find $y^k$ with $E(y^k)\to 0$ but $Q(y^k)\notin U$, contradicting $Q(y^k)\to\cS$).
Hence, for any $y$ with $E(y)\le \bar E$, we have $x=Q(y)\in U$ and thus, using \eqref{eq:dotE-flux},
\[
\dot E(y)
=\sum_{\alpha\notin\HH} e^{-F_h(\alpha,y)}\,\Phi(x,\alpha)
\le -\frac{c}{2}\sum_{\alpha\notin\HH} e^{-F_h(\alpha,y)}
=-\frac{c}{2}\,E(y).
\]
Consequently, if $0<E^-<E^+\le \bar E$ and $E^-<E(y)<E^+$, then $\dot E(y)\le -(c/2)\,E^-<0$.
Taking the supremum over $\{y:\ E^-<E(y)<E^+\}$ gives \cref{def:local-energy}(\labelcref{it:LE-dissip}).

All three conditions of \cref{def:local-energy} are satisfied, so $\bar F_\HH$ is a local energy function for $\cont(\HH)$.
\end{proof}

We may now use this to deduce the main theorem.
\begin{proof}[\textbf{Proof of \cref{thm:strict-resilience}.}]
We shall proceed again by gluing all the facewise basins into a basin on all of $\X$. Let $\cS\coloneqq \cont(\HH)$. Under \cref{asm:well-posed}, the regularizer is steep, so the facewise \ac{FTRL} dynamics are well-defined. Moreover, $\cS$ is a union of faces, so compact and each face is invariant under the strategy flow, hence $\cS$ is invariant.

It remains to prove that $\cS$ is asymptotically stable for the strategy flow.
Fix any face $\F$ of $\X$ with $\cS_\F\coloneqq \cS\cap \F\neq\varnothing$, and write $\F=\cont(\B)$ for the corresponding subgame $\B=\prod_i \B_i\subseteq \A$. Define $\HH_\F\coloneqq \HH\cap \B$. We claim
\begin{equation}
\label{eq:face-span}
\cS_\F
    = \cont(\HH_\F).
\end{equation}
Indeed, if $x\in\cS\cap\F$, then $x_\alpha=0$ for every $\alpha\notin \HH$ (since $x\in\cont(\HH)$) and also $x_\alpha=0$ for every $\alpha\notin\B$ (since $x\in\cont(\B)$), therefore $x_\alpha=0$ for every $\alpha\notin(\HH\cap\B)=\HH_\F$, i.e.\ $x\in\cont(\HH_\F)$.
Conversely, $\cont(\HH_\F)\subseteq \cont(\HH)\cap \cont(\B)=\cS\cap\F$, proving \eqref{eq:face-span}.

Next, we show that $\HH_\F$ is \ac{srad} for the restricted game on $\B$.
Fix $\alpha\in \B\setminus \HH_\F$. Then $\alpha\notin\HH$.
Let $\G$ be any face with $\G\subseteq \cont(\HH_\F)$. By \eqref{eq:face-span}, $\G\subseteq \cont(\HH)$.
Since $\HH$ is \ac{srad} (cf.~\cref{lem:resilience-mixed}), we have
\[
\sup_{x\in\G} \Phi(x,\alpha)<0.
\]
Since $\alpha\in\B$ and $\G\subseteq\cont(\B)$, this is exactly the \ac{sradness} requirement for $\HH_\F$ within $\B$.

Now consider the face-restricted choice map $Q_\B:\YY\to\F^\circ$ and the facewise \ac{FTRL} dynamics
\[
\dot y=v(x),\qquad x=Q_\B(y),
\tag{FTRL-$\FTRLtag$}
\]
as in \cref{rmrk:flow-link}. Since $h$ is steep and $\HH_\F$ is strictly \ac{rad}, the \ac{radness} energy construction
(cf.\ \cref{lem:res-energy}) provides a local energy function for $\cont(\HH_\F)=\cS_\F$ of the form
\[
\bar F_{\HH_\F}(y)=\sum_{\alpha\in \B\setminus\HH_\F} e^{-F_{h|\F}(\alpha,y)}.
\]
Therefore, by the continuous-time energy theorem (cf.\ \cref{thm:energy-ct}), $\cS_\F$ is asymptotically stable under the
facewise \ac{FTRL} dynamics. Finally, by \cref{rmrk:flow-link}, for every initial condition $x_0\in\F$ the
strategy flow trajectory starting at $x_0$ coincides with the appropriate facewise \ac{FTRL} orbit on the minimal face of $\F$
containing $x_0$. It follows that $\cS_\F$ is asymptotically stable for the restriction of the strategy flow to $\F$.
Concretely: for each such face $\F$ there exists a neighborhood $U_\F\subseteq \F$ of $\cS_\F$ (relative topology) such that
every strategy flow trajectory starting in $U_\F$ remains in $U_\F$ and satisfies $\dist(\Theta_t,\cS_\F)\to 0$ as $t\to\infty$.

We now patch these facewise basins into a basin for $\cS$ in $\X$.
Let $\mathfrak{F}_\cS$ be the (finite) family of faces $\F$ of $\X$ with $\cS_\F\neq\varnothing$, and fix for each
$\F\in\mathfrak{F}_\cS$ a neighborhood $U_\F\subseteq \F$ with the invariance/attraction property above.
Fix $x\in\cS$ and set
\[
\mathfrak{F}(x)\coloneqq \{\F\in\mathfrak{F}_\cS:\ x\in\F\}.
\]
Because $\X$ has finitely many faces, the number
\[
r_x \coloneqq \tfrac12\min\{\dist(x,\F'): \F'\text{ a face of }\X,\ x\notin \F'\}
\]
is well-defined and strictly positive. In particular, $B_{r_x}(x)\cap \X$ meets only faces that contain $x$.
Moreover, for each $\F\in\mathfrak{F}(x)$, since $U_\F$ is a neighborhood of $x$ in $\F$ (relative topology), there exists
$\rho^x_{\F}>0$ such that $B_{\rho^x_{\F}}(x)\cap \F\subseteq U_\F$. Define
\[
\rho_x \coloneqq \min\Bigl(r_x,\ \min_{\F\in\mathfrak{F}(x)}\rho^x_{\F}\Bigr)>0.
\]
Then for every $z\in B_{\rho_x}(x)\cap\X$, every face of $\X$ containing $z$ must contain $x$ (since it meets
$B_{r_x}(x)\cap\X$), so the minimal face $\F_z$ containing $z$ contains $x$, hence $\F_z\in\mathfrak{F}(x)$, by construction,
this implies $z\in U_{\F_z}$. Consequently, the family $\{B_{\rho_x}(x)\cap\X\}_{x\in\cS}$ covers $\cS$.
By compactness of $\cS$, choose $x^1,\dots,x^m\in\cS$ with
\[
\cS\subseteq U \coloneqq \bigcup_{k=1}^m \bigl(B_{\rho_{x^k}}(x^k)\cap \X\bigr),
\]
so $U$ is an open neighborhood of $\cS$ in $\X$.

\emph{Attraction.}
Take any $\Theta_0\in U$ and let $\F$ be the minimal face of $\X$ containing $\Theta_0$.
By construction of $U$, there exists some $k$ such that $\Theta_0\in B_{\rho_{x^k}}(x^k)\cap\X$, hence $\Theta_0\in U_\F$ by the
preceding argument. Since $\F$ is invariant under the strategy flow, the trajectory $t\mapsto \Theta_t(\Theta_0)$ remains in $\F$. Since
$\Theta_0\in U_\F$, it remains in $U_\F$ and satisfies
\[
\dist\bigl(\Theta_t(\Theta_0),\cS_\F\bigr)\xrightarrow[t\to\infty]{}0.
\]
Because $\cS_\F\subseteq \cS$, we have $\dist(\Theta_t(\Theta_0),\cS)\le \dist(\Theta_t(\Theta_0),\cS_\F)\to 0$, proving that $\cS$ attracts $U$.

\emph{Stability.}
Let $W$ be an arbitrary neighborhood of $\cS$ in $\X$.
For each $\F\in\mathfrak{F}_\cS$, the set $W\cap \F$ is a neighborhood of $\cS_\F$ in the relative topology of $\F$.
Since $\cS_\F$ is stable for the restricted flow on $\F$, there exists a neighborhood $V_\F\subseteq \F$ of $\cS_\F$ such that
$x_0\in V_\F$ implies $\Theta_t(x_0)\in W\cap\F$ for all $t\ge 0$.
Repeating the compactness/ball-cover construction above with $(V_\F)_{\F\in\mathfrak{F}_\cS}$ in place of
$(U_\F)_{\F\in\mathfrak{F}_\cS}$ yields an open neighborhood $V$ of $\cS$ in $\X$ such that every $x_0\in V$ lies in
$V_{\F_{x_0}}$, where $\F_{x_0}$ is the minimal face containing $x_0$, hence $\Theta_t(x_0)\in W$ for all $t\ge 0$.
Thus $\cS$ is stable.

We have shown that $\cS$ is compact, invariant, and asymptotically stable for the strategy flow.
It follows then that
$\cont(\HH)=\cS$ is an attractor for the strategy flow induced by \eqref{eq:SD}.
\end{proof}

In the entropic setting, we shall work with a much simpler energy function:

\begin{lemma}\label{lem:mass-energy}
Let $\HH\subseteq\A$ and assume $h$ is entropic. Define
\[
\bar W_\HH(x)\coloneqq \sum_{\alpha\notin\HH} x_\alpha,\qquad x\in\X.
\]
Then, if $\HH$ is \ac{rad} and closed under better replies, $\bar W_\HH\circ Q$ is a local energy function for $\cont(\HH)$.
\end{lemma}

\begin{proof}
Set $\cS\coloneqq \cont(\HH)$ and define
\[
E:\YY\to[0,1],\qquad E(y)\coloneqq \bar W_\HH(Q(y)).
\]
If $\A\setminus\HH=\varnothing$, then $\cS=\X$ and $\bar W_\HH=0$, hence $E=0$, and the three items of \cref{def:local-energy} hold trivially. Assume henceforth $\A\setminus\HH\neq\varnothing$.

We start with regularity. Since $h$ is entropic, $Q$ is the softmax, hence $C^\infty$. Moreover by \cref{prop:choice-map}(\labelcref{it:Qgrad}), it is globally Lipschitz. Next, $\bar W_\HH$ is a polynomial in the coordinates $x_{i\alpha}$ because each monomial $x_\alpha=\prod_i x_{i\alpha}$ is polynomial, thus $\bar W_\HH\in C^\infty(\X)$ and is Lipschitz on compact $\X$. Therefore $E=\bar W_\HH\circ Q$ is $C^1$ and globally Lipschitz on $\YY$.

For positive semi-definiteness, $\bar W_\HH(x)$ is a sum of nonnegative terms, so
\[
\bar W_\HH(x)=0
\quad\Longleftrightarrow\quad
x_\alpha=0\ \ \forall \alpha\notin\HH
\quad\Longleftrightarrow\quad
x\in \cS,
\]
hence
\begin{equation}\label{eq:zero-set-W}
\cS=\bar W_\HH^{-1}(0).
\end{equation}
We do again the basic subsequence argument. Let $(y^k)$ be any sequence and set $x^k\coloneqq Q(y^k)\in\X$. If $x^k\to\cS$, then by continuity of $\bar W_\HH$ and \eqref{eq:zero-set-W}, $E(y^k)=\bar W_\HH(x^k)\to 0$. Conversely, if $E(y^k)=\bar W_\HH(x^k)\to 0$ but $x^k\not\to\cS$, then there exist $\varepsilon_0>0$ and a subsequence $x^{k_j}$ with $\dist(x^{k_j},\cS)\ge \varepsilon_0$ for all $j$. By compactness of $\X$, pass to a further subsequence (not relabeled) with $x^{k_j}\to \bar x\in\X$. Continuity gives $\bar W_\HH(\bar x)=0$, so $\bar x\in\cS$ by \eqref{eq:zero-set-W}. Since $x\mapsto \dist(x,\cS)$ is continuous and $\dist(\bar x,\cS)=0$, we get $\dist(x^{k_j},\cS)\to 0$, contradicting $\dist(x^{k_j},\cS)\ge \varepsilon_0$ for $j$ large. Hence $x^k\to\cS$.

We now turn to dissipativity. Let $y$ be any solution orbit of \eqref{eq:FTRL} and set $x\coloneqq Q(y)$. In the entropic case, the induced strategy flow coincides with the replicator dynamics
\begin{equation}\label{eq:RD-local}
\dot x_{i\alpha}=x_{i\alpha}\bigl(v_{i\alpha}(x)-u_i(x)\bigr),
\qquad i\in\N,\ \alpha_i\in\A_i.
\end{equation}
By steepness, $x_{i\alpha}>0$ for all $i$ and $\alpha_i\in\A_i$. The product rule and \eqref{eq:RD-local} then give
\[
\frac{d}{dt}x_\alpha
=
x_\alpha\sum_{i\in\N}\frac{\dot x_{i\alpha_i}}{x_{i\alpha_i}}
=
x_\alpha\sum_{i\in\N}\bigl(v_{i\alpha}(x)-u_i(x)\bigr).
\]
Also
\[
\Phi(x,\alpha)=\sum_{i\in\N}\bigl(v_{i\alpha}(x)-u_i(x)\bigr),
\]
so we obtain the pointwise identity
\begin{equation}\label{eq:xalpha-dot}
\dot x_\alpha=x_\alpha\,\Phi(x,\alpha).
\end{equation}
Summing \eqref{eq:xalpha-dot} over $\alpha\notin\HH$ yields, along the orbit,
\begin{equation}\label{eq:Edot}
\frac{d}{dt} \bar W_\HH(x)=\sum_{\alpha\notin\HH}x_\alpha\,\Phi(x,\alpha).
\end{equation}
Equivalently, for every $y\in\YY$,
\begin{equation}\label{eq:dotE-formula}
\dot E(y)
=\sum_{\alpha\notin\HH}x_\alpha\,\Phi(x,\alpha),
\qquad x=Q(y).
\end{equation}
Define the continuous function
\[
\phi(x)\coloneqq \sum_{\alpha\notin\HH}x_\alpha\,\Phi(x,\alpha),\qquad x\in\X.
\]
By \eqref{eq:dotE-formula}, it is enough to show that $\phi$ is strictly negative whenever the ``mass outside $\HH$'', namely $\bar W_\HH(x)$, is positive but small: we will prove that there exists $\bar E>0$ such that
\begin{equation}\label{eq:neg-sublevel}
0<\bar W_\HH(x)\le \bar E
\quad\Longrightarrow\quad
\varphi(x)<0.
\end{equation}

Fix $x^\star\in\cS$. Set $\B_i\coloneqq \supp(x^\star_i)$ and $\B\coloneqq \prod_{i\in\N}\B_i$. Since $x^\star\in\cS$, we have $\B\subseteq \HH$. Let $\F\coloneqq \cont(\B)$, then $x^\star\in \F^\circ$. For $\alpha\in\A$ define the \emph{deviating coalition} from $\B$:
\[D_\B(\alpha)\coloneqq\{i\in\N:\alpha_i\notin \B_i\},\]
and on $\A\setminus\HH$ introduce the preorder $\preceq_\B$ by: $\alpha\preceq_\B \nu$ if and only if $D_\B(\alpha)\subseteq D_\B(\nu)$ and $\alpha_i=\nu_i$ for all $i\in D_\B(\alpha)$. Let $M_\B\subseteq \A\setminus\HH$ be the set of $\preceq_\B$-minimal excluded profiles, and set
\[
\bar W_\B(x)\coloneqq \sum_{\alpha\in M_\B}x_\alpha.
\]
Notice that the profiles in $M_\B$ are those which are \enquote{closest} to $\B$, in the sense that they're one unilateral deviation away from it: this is how we can exploit \ac{clubness} to get a negative contribution.

\emph{Claim 1.} If $\alpha\in M_\B$, $i\in D_\B(\alpha)$, and $\beta_i\in \B_i$, then $(\beta_i,\alpha_{-i})\in\HH$.
Indeed, let $\gamma=(\beta_i,\alpha_{-i})$. Then $D_\B(\gamma)=D_\B(\alpha)\setminus\{i\}\subsetneq D_\B(\alpha)$ and $\gamma_j=\alpha_j$ for all $j\in D_\B(\gamma)$, hence $\gamma\preceq_\B \alpha$ with $\gamma\neq \alpha$. If $\gamma\notin\HH$, this contradicts the $\preceq_\B$-minimality of $\alpha\in M_\B$, therefore $\gamma\in\HH$.

\emph{Claim 2.} There exist a neighborhood $U$ of $x^\star$ in $\X$ and $\gamma\in(0,1)$ such that
\begin{equation}\label{eq:min-dom}
x\in U\setminus \cS \ \Longrightarrow\ \bar W_\B(x)\ge \gamma\,\bar W_\HH(x).
\end{equation}
Indeed, let $n\coloneqq |\N|$ and set
\[
\delta\coloneqq \frac12\min_{i\in\N,\ \beta_i\in \B_i}x^\star_{i\beta}>0.
\]
By continuity of $x\mapsto x_{i\beta}$, there exists a neighborhood $U_0$ of $x^\star$ such that
\begin{equation}\label{eq:delta-lb}
x_{i\beta}\ge \delta
\qquad \forall x\in U_0,\ \forall i\in\N,\ \forall \beta_i\in \B_i.
\end{equation}
Define
\[
\varepsilon(x)\coloneqq \max_{i\in\N}\sum_{\alpha_i\notin \B_i}x_{i\alpha_i}.
\]
Since $\varepsilon(x^\star)=0$, shrink $U_0$ so that $\varepsilon(x)\le \varepsilon_0\le \delta$ for all $x\in U_0$, where $\varepsilon_0>0$ will be fixed below. Fix $\nu\in(\A\setminus\HH)\setminus M_\B$. Let $\mu=\mu(\nu)$ be any $\preceq_\B$-minimal element of the nonempty set $\{\alpha\in\A\setminus\HH:\ \alpha\preceq_\B \nu\}$. Then $\mu\in M_\B$ and $\mu\preceq_\B \nu$. Moreover $D_\B(\mu)\subsetneq D_\B(\nu)$: if $D_\B(\mu)=D_\B(\nu)$, then $\mu_i=\nu_i$ for all $i\in D_\B(\nu)$ and also $\mu\preceq_\B \nu$ and $\nu\preceq_\B \mu$, so $\nu$ would be $\preceq_\B$-minimal as well, contradicting $\nu\notin M_\B$. We claim that for all $x\in U_0$,
\begin{equation}\label{eq:termwise-dom}
x_\nu \le C\,\varepsilon(x)\,x_\mu,
\qquad C\coloneqq \delta^{-n}.
\end{equation}
If $x_\mu=0$, then $x_\nu=0$ as well, hence \eqref{eq:termwise-dom} holds: indeed, if $x_\mu=0$ then some index $j$ has $x_{j\mu}=0$. By \eqref{eq:delta-lb} this forces $j\in D_\B(\mu)$, and since $\mu_j=\nu_j$ for all $j\in D_\B(\mu)$ we get $x_{j\nu}=0$ and thus $x_\nu=0$. Assume now $x_\mu>0$ and write
\[
\frac{x_\nu}{x_\mu}=\prod_{i\in\N}\frac{x_{i\nu}}{x_{i\mu}}.
\]
If $i\in D_\B(\mu)$ then $\mu_i=\nu_i$, so the factor equals $1$. If $i\in D_\B(\nu)\setminus D_\B(\mu)$ then $\mu_i\in \B_i$ and $\nu_i\notin \B_i$, hence by \eqref{eq:delta-lb} and the definition of $\varepsilon(x)$,
\[
x_{i\mu}\ge \delta,
\qquad
x_{i\nu}\le \sum_{\alpha_i\notin \B_i}x_{i\alpha}\le \varepsilon(x),
\quad\Rightarrow\quad
\frac{x_{i\nu}}{x_{i\mu}}\le \frac{\varepsilon(x)}{\delta}.
\]
Finally, if $i\notin D_\B(\nu)$ then $\nu_i\in \B_i$, and also $\mu_i\in\B_i$ (else $i\in D_\B(\mu)\subseteq D_\B(\nu)$), so $x_{i\mu}\ge \delta$ by \eqref{eq:delta-lb} and trivially $x_{i\nu}\le 1$, hence $x_{i\nu}/{x_{i\mu}}\le 1/\delta$. Since $\nu\notin\HH$ and $\B\subseteq\HH$, we have $\nu\notin\B$ and thus $D_\B(\nu)\neq\varnothing$, implying $n-|D_\B(\nu)|\le n-1$. Also $D_\B(\mu)\subsetneq D_\B(\nu)$ implies $|D_\B(\nu)\setminus D_\B(\mu)|\ge 1$. Using $\varepsilon(x)\le \delta$ on $U_0$ gives
\[
\frac{x_\nu}{x_\mu}
\le
\Big(\frac{\varepsilon(x)}{\delta}\Big)\Big(\frac{1}{\delta}\Big)^{n-1}
=
\delta^{-n}\,\varepsilon(x)
=
C\,\varepsilon(x),
\]
proving \eqref{eq:termwise-dom}. Summing \eqref{eq:termwise-dom} over $\nu\in(\A\setminus\HH)\setminus M_\B$ yields, for all $x\in U_0$,
\[
\bar W_\HH(x)
=
\bar W_\B(x)+\sum_{\nu\in(\A\setminus\HH)\setminus M_\B}x_\nu
\le
\bar W_\B(x)+C\,\varepsilon(x)\,|\A\setminus\HH|\,\bar W_\B(x).
\]
Choose $\varepsilon_0$ so that $C\,\varepsilon_0\,|\A\setminus\HH|\le \tfrac12$, and set $U\coloneqq U_0$ with this choice. Then for all $x\in U$ we have $\bar W_\HH(x)\le \tfrac32\bar W_\B(x)$, i.e.\ \eqref{eq:min-dom} holds with $\gamma\coloneqq 2/3$.

Next, since $\HH$ is \ac{club}, every unilateral deviation from $\HH$ to $\A\setminus\HH$ is strictly unprofitable. Formally, let $\mathcal E$ be the finite set of triples $(\beta,\alpha,i)$ such that $\beta\in\HH$, $\alpha\notin\HH$, and $\alpha$ and $\beta$ $i$-comparable. For $(\beta,\alpha,i)\in\mathcal E$, closure under better replies implies $u_i(\alpha)<u_i(\beta)$, hence define
\[
\eta\coloneqq \min_{(\beta,\alpha,i)\in\mathcal E}\bigl(u_i(\beta)-u_i(\alpha)\bigr)>0.
\]
If $\alpha,\beta$ differ only in $i$, then for $j\neq i$ we have $u_j(\alpha_j,\beta_{-j})=u_j(\beta)$, so
\begin{equation}\label{eq:eta-flux}
\Phi(\beta,\alpha)
=
\sum_{j\in\N}\bigl(u_j(\alpha_j,\beta_{-j})-u_j(\beta)\bigr)
=
u_i(\alpha)-u_i(\beta)
\le -\eta.
\end{equation}
We now estimate $\phi(x)$ for $x$ near $\cS$. Fix $\alpha\in\A$. Since $x\mapsto \Phi(x,\alpha)$ is multilinear (being a finite sum of multilinear payoff terms), for every $x\in\X$ we have the multilinear interpolation identity
\begin{equation}\label{eq:flux-expand}
\Phi(x,\alpha)=\sum_{\beta\in\A}x_\beta\,\Phi(\beta,\alpha),
\end{equation}
obtained by expanding successively in each block $x_i$. Using \eqref{eq:flux-expand},
\begin{equation}\label{eq:varphi-decomp}
\phi(x)
=
\sum_{\alpha\notin\HH}\sum_{\beta\in\HH}x_\alpha x_\beta\,\Phi(\beta,\alpha)
+
\sum_{\alpha\notin\HH}\sum_{\beta\notin\HH}x_\alpha x_\beta\,\Phi(\beta,\alpha).
\end{equation}
Let $D\coloneqq \max_{\alpha,\beta\in\A}|\Phi(\beta,\alpha)|<\infty$. Then the second term in \eqref{eq:varphi-decomp} is bounded above by
\[
\sum_{\alpha\notin\HH}\sum_{\beta\notin\HH}x_\alpha x_\beta\,|\Phi(\beta,\alpha)|
\le
D\Big(\sum_{\alpha\notin\HH}x_\alpha\Big)^2
=
D\,\bar W_\HH(x)^2.
\]
For the first term, \ac{radness} gives $\Phi(\beta,\alpha)\le 0$ for all $\alpha\notin\HH$ and all $\beta\in\HH$, hence the first term is $\le 0$. To obtain strict negativity when $x\notin\cS$ but is close to $x^\star$, we focus on the minimal excluded profiles. Fix $\alpha\in M_\B$ and choose any $i(\alpha)\in D_\B(\alpha)$. Pick $\beta_{i(\alpha)}\in \B_{i(\alpha)}$ maximizing $x_{i(\alpha)\beta}$ over $\B_{i(\alpha)}$, and set
\[
\beta(\alpha)\coloneqq (\beta_{i(\alpha)},\alpha_{-i(\alpha)}).
\]
By Claim 1, $\beta(\alpha)\in\HH$, and $\alpha$ and $\beta(\alpha)$ differ only in player $i(\alpha)$, so by \eqref{eq:eta-flux}, $\Phi(\beta(\alpha),\alpha)\le -\eta$. Therefore, using again $\Phi(\beta,\alpha)\le 0$ for $\beta\in\HH$,
\[
\sum_{\alpha\notin\HH}\sum_{\beta\in\HH}x_\alpha x_\beta\,\Phi(\beta,\alpha)
\le
\sum_{\alpha\in M_\B} x_\alpha x_{\beta(\alpha)}\,\Phi(\beta(\alpha),\alpha)
\le
-\eta \sum_{\alpha\in M_\B} x_\alpha x_{\beta(\alpha)}.
\]
For $x\in U$ from Claim 2 and $\alpha\in M_\B$, we have $\beta_{i(\alpha)}\in\B_{i(\alpha)}$, so \eqref{eq:delta-lb} gives $x_{i(\alpha)\beta}\ge \delta$, while $\alpha_{i(\alpha)}\notin \B_{i(\alpha)}$ implies $x_{i(\alpha)\alpha}\le \varepsilon(x)$. Hence, whenever $x_\alpha>0$,
\[
x_{\beta(\alpha)} = x_\alpha\,\frac{x_{i(\alpha)\beta}}{x_{i(\alpha)\alpha}}
\ge x_\alpha\,\frac{\delta}{\varepsilon(x)},
\]
and the same inequality holds trivially if $x_\alpha=0$. Thus
\[
x_\alpha x_{\beta(\alpha)} \ge \frac{\delta}{\varepsilon(x)}\,x_\alpha^2,
\qquad \alpha\in M_\B.
\]
Consequently,
\[
\sum_{\alpha\notin\HH}\sum_{\beta\in\HH}x_\alpha x_\beta\,\Phi(\beta,\alpha)
\le
-\eta\,\frac{\delta}{\varepsilon(x)}\sum_{\alpha\in M_\B}x_\alpha^2
\le
-\eta\,\frac{\delta}{|M_\B|\,\varepsilon(x)}\,\bar W_\B(x)^2,
\]
where we used $\sum_{\alpha\in M_\B}x_\alpha^2\ge |M_\B|^{-1}\big(\sum_{\alpha\in M_\B}x_\alpha\big)^2
=|M_\B|^{-1}\bar W_\B(x)^2$. Combining with the bound on the second term in \eqref{eq:varphi-decomp} and \eqref{eq:min-dom} yields, for all $x\in U\setminus\cS$,
\[
\phi(x)
\le
-\eta\,\frac{\delta}{|M_\B|\,\varepsilon(x)}\,\bar W_\B(x)^2
+
D\,\bar W_\HH(x)^2
\le
\bar W_\B(x)^2\Bigl(
-\eta\,\frac{\delta}{|M_\B|\,\varepsilon(x)} + D\,\gamma^{-2}
\Bigr).
\]
Note that on $U\setminus\cS$ we have $\varepsilon(x)>0$ (if $\varepsilon(x)=0$ then $x_i$ is supported on $\B_i$ for every $i$, and since $\B\subseteq\HH$ this forces $x\in\cS$). Since $\varepsilon(x^\star)=0$ and $\varepsilon$ is continuous, we may shrink $U$ (keeping $x^\star\in U$) so that
\[
\varepsilon(x)\le \frac{\eta\,\delta\,\gamma^2}{2D\,|M_\B|}
\qquad \forall x\in U.
\]
Then the bracket is strictly negative for all $x\in U\setminus\cS$, hence $\phi(x)<0$ on $U\setminus\cS$.

Since $x^\star\in\cS$ was arbitrary, these neighborhoods form an open cover of compact $\cS$. Extracting a finite subcover and letting $U$ be the union, we obtain an open neighborhood $U$ of $\cS$ such that
\begin{equation}\label{eq:neg-neigh}
\phi(x)<0\qquad \forall x\in U\setminus\cS.
\end{equation}
Because $\bar W_\HH$ is continuous, $\cS=\bar W_\HH^{-1}(0)$, and $\X\setminus U$ is compact and disjoint from $\cS$,
\[
\bar E\coloneqq \inf_{x\in \X\setminus U}\bar W_\HH(x)>0.
\]
Then $\{x\in\X:0<\bar W_\HH(x)\le \bar E\}\subseteq U\setminus\cS$, and \eqref{eq:neg-sublevel} follows from \eqref{eq:neg-neigh}. Finally, let $0<E^-<E^+\le \bar E$ and consider the compact set
\[
K\coloneqq \{x\in\X: E^-\le \bar W_\HH(x)\le E^+\}\subseteq \{x:0<\bar W_\HH(x)\le \bar E\}.
\]
By \eqref{eq:neg-sublevel}, $\phi(x)<0$ for all $x\in K$, hence $\max_{x\in K}\phi(x)<0$ by continuity. Using \eqref{eq:dotE-formula}, for any $y$ with $E^-<E(y)<E^+$ we have $x=Q(y)\in K$ and $\dot E(y)=\phi(x)$, so
\[
\sup\{\dot E(y): E^-<E(y)<E^+\}
\le
\max_{x\in K}\phi(x)
<0,
\]
which is exactly \cref{def:local-energy}(\labelcref{it:LE-dissip}). Hence $E=\bar W_\HH\circ Q$ is a local energy function for $\cS=\cont(\HH)$.
\end{proof}

From this, as usual, we get asymptotic stability.
\begin{proof}[\textbf{Proof of \cref{thm:resilience-RD}.}]
If $\A\setminus\HH=\varnothing$, then $\cont(\HH)=\X$, hence it is (trivially) an attractor for \eqref{eq:RD}.
Assume henceforth $\A\setminus\HH\neq\varnothing$, and set $\cS\coloneqq \cont(\HH)$. Since $h$ is entropic the induced strategy flow
$(\Theta_t)_{t\in\R}$ is well-defined and face-invariant, moreover it coincides with the replicator
dynamics \eqref{eq:RD}. In particular, by \cref{prop:flow} every face of $\X$ is invariant
for $(\Theta_t)$. As $\cS$ is a union of faces, it follows that $\cS$ is compact and invariant. We will prove that $\cS$
is asymptotically stable for $(\Theta_t)$, hence an attractor.

Comes again the time for some gluing: we shall patch the facewise basins into a basin on all of $\X$. Let $\F=\prod_{i\in\N}\Delta(\B_i)$ be a face of $\X$ and write $\A(\F)\coloneqq\prod_{i\in\N}\B_i$ for its vertex set.
Define
\[
\HH_\F\coloneqq \HH\cap \A(\F),\qquad \cS_\F\coloneqq \cS\cap \F.
\]
Then
\begin{equation}\label{eq:face-span-RD}
\cS_\F=\cont(\HH_\F).
\end{equation}
Indeed, if $x\in\cS_\F$, then $\supp(x)\subseteq \A(\F)$ (since $x\in\F$) and $\supp(x)\subseteq\HH$ (since $x\in\cont(\HH)$),
so $\supp(x)\subseteq \HH\cap\A(\F)=\HH_\F$, i.e.\ $x\in\cont(\HH_\F)$, and the reverse inclusion is immediate. We claim now that $\HH_\F$ is \ac{rad} and \ac{club} for the dynamics restricted to $\F$.
For \ac{radness}, let $\F'\subseteq \F$ be any face with $\F'\subseteq \cont(\HH_\F)$ and let $\alpha\in \A(\F)\setminus \HH_\F$.
Then $\alpha\notin\HH$ and $\F'\subseteq \cont(\HH)\!=\cS$, so by \ac{radness} of $\HH$ and \cref{lem:resilience-mixed},
\[
\sup_{x\in \F'}\,\Phi(x,\alpha)\le 0.
\]
the preceding inequality is exactly the \ac{rad} condition for $\HH_\F$ on $\F$. For closure under better replies, let $\beta\in\HH_\F$ and suppose $\beta\to\alpha$ is a better-reply edge with $\alpha\in\A(\F)$.
Then also $\beta\in\HH$ and $\beta\to\alpha$ is a better-reply edge in the full game, so $\alpha\in\HH$ because $\HH$ is closed under
better replies. As $\alpha\in\A(\F)$, this gives $\alpha\in\HH\cap\A(\F)=\HH_\F$.

Fix a face $\F$ with $\cS_\F\neq\varnothing$, and consider the face-restricted choice map $Q_\B:\YY\to\F^\circ$ and the corresponding
facewise \ac{FTRL} dynamics \eqref{eq:FTRL-B}. Apply \cref{lem:mass-energy} to the subgame on $\F$ (action sets $\B_i$) and the set
$\HH_\F\subseteq \A(\F)$: by the previous paragraph, $\HH_\F$ is \ac{rad} and \ac{club} on $\F$, so
\[
E_\F(y)\coloneqq \bar W_{\HH_\F}\bigl(Q_\B(y)\bigr)
\]
is a local energy function for $\cS_\F=\cont(\HH_\F)$ under \eqref{eq:FTRL-B}. Hence, by \cref{thm:energy-ct}, $\cS_\F$ is
asymptotically stable under \eqref{eq:FTRL-B}. By \cref{rmrk:flow-link}, the strategy flow trajectory starting from any
$x_0\in\F$ coincides with the orbit of \eqref{eq:FTRL-B} on the minimal face of $\F$ containing $x_0$, therefore $\cS_\F$ is
asymptotically stable for the restriction of the strategy flow to $\F$. Concretely: for each face $\F$ with $\cS_\F\neq\varnothing$,
there exists a neighborhood $U_\F\subseteq \F$ of $\cS_\F$ (in the relative topology of $\F$) such that every strategy flow trajectory
starting in $U_\F$ remains in $U_\F$ and satisfies $\dist(\Theta_t,\cS_\F)\to 0$.

Let now $\mathfrak{F}_\cS$ be the (finite) collection of faces $\F$ of $\X$ such that $\cS_\F\neq\varnothing$, and for each such $\F$ let
$U_\F\subseteq\F$ be as above. Fix any point $x\in\cS$. Let $\mathfrak{F}(x)\coloneqq\{\F\in\mathfrak{F}_\cS:\ x\in \F\}$. Because $\X$ has finitely many faces, the quantity
\[
r_x \coloneqq \tfrac12\min\{\dist(x,\F'): \F'\text{ a face of }\X,\ x\notin \F'\}
\]
is well-defined and strictly positive. In particular, $B_{r_x}(x)\cap \X$ meets only faces that contain $x$.
Moreover, for each $\F\in\mathfrak{F}(x)$, since $U_\F$ is a neighborhood of $x$ in $\F$, there exists $\rho^x_{\F}>0$ such that
$B_{\rho^x_{\F}}(x)\cap \F\subseteq U_\F$. Set
\[
\rho_x \coloneqq \min\Bigl(r_x,\ \min_{\F\in\mathfrak{F}(x)}\rho^x_{\F}\Bigr)>0.
\]
Then for every $z\in B_{\rho_x}(x)\cap\X$, the (unique) minimal face $\F_z$ containing $z$ must contain $x$, hence belongs to $\mathfrak{F}(x)$,
and therefore $z\in U_{\F_z}$. Thus the family $\{B_{\rho_x}(x)\cap\X\}_{x\in\cS}$ covers $\cS$. By compactness of $\cS$, there exist $x^1,\dots,x^m\in \cS$ such that
\[
\cS\subseteq U \coloneqq \bigcup_{k=1}^m \bigl(B_{\rho_{x^k}}(x^k)\cap \X\bigr),
\]
and $U$ is an open neighborhood of $\cS$ in $\X$.
Now take any $\Theta_0\in U$ and let $\F$ be the minimal face of $\X$ containing $\Theta_0$. By construction of $U$, $\F$ contains some $x^k\in\cS$ and
$\Theta_0\in U_\F$. Since faces are invariant (\cref{prop:flow}), $\Theta_t\in\F$ for all $t\ge 0$, and we have $\dist(\Theta_t,\cS_\F)\to 0$. As
$\cS_\F\subseteq\cS$, this implies $\dist(\Theta_t,\cS)\to 0$. This proves that $\cS$ is attracting. The stability part is obtained by the same
patching argument, using the stability component of asymptotic stability on each $\F\in\mathfrak{F}_\cS$.

Therefore $\cS=\cont(\HH)$ is compact, invariant, and asymptotically stable for the strategy flow, hence an attractor. Since this flow
coincides with \eqref{eq:RD} in the entropic case, $\cont(\HH)$ is an attractor under \eqref{eq:RD}.
\end{proof}

\bibliographystyle{icml}
\bibliography{bibtex/IEEEabrv,bibtex/Bibliography}

\end{document}